%% file: aistats2023.tex
\documentclass[twoside]{article}

\usepackage[accepted]{aistats2023}
\usepackage[round]{natbib}

\usepackage{graphicx}

\input{preamble}

\begin{document}

%
\runningtitle{Last-Iterate Convergence with Full and Noisy Feedback in Two-Player Zero-Sum Games}

%
\runningauthor{Kenshi Abe, Kaito Ariu, Mitsuki Sakamoto, Kentaro Toyoshima, Atsushi Iwasaki}

\twocolumn[

\aistatstitle{Last-Iterate Convergence with Full and Noisy Feedback \\ in Two-Player Zero-Sum Games}

\aistatsauthor{ Kenshi Abe \And Kaito Ariu}
\aistatsaddress{CyberAgent, Inc. \And CyberAgent, Inc., KTH Royal Institute of Technology}
\aistatsauthor{Mitsuki Sakamoto \And Kentaro Toyoshima \And Atsushi Iwasaki}
\aistatsaddress{University of Electro-Communications}

]

\input{abstract}

\input{main}

\subsubsection*{Acknowledgements}
Atsushi Iwasaki was supported by JSPS KAKENHI Grant Numbers 21H04890 and 20K20752.
We are indebted to Alexandre Proutiere for discussing the problem setting and theoretical results. 
We also thank Riku Togashi for commenting on our earlier drafts. 

\bibliography{ref}

\appendix
\onecolumn
\input{appendix}

\vfill

\end{document}

%% file: preamble.tex
\usepackage{amsmath}
\usepackage{amssymb}
\usepackage{amsthm}
\usepackage{bbm}
\usepackage{subcaption}
\usepackage{algorithmic}
\usepackage{algorithm}
\usepackage{hyperref}

\newcommand{\argmin}{\mathop{\rm arg~min}\limits}
\newtheorem{theorem}{Theorem}[section]

\newtheorem{lemma}[theorem]{Lemma}
\newtheorem{corollary}[theorem]{Corollary}
\newtheorem{definition}[theorem]{Definition}
\newtheorem{assumption}[theorem]{Assumption}
\newtheorem{remark}[theorem]{Remark}

\allowdisplaybreaks[1]

\usepackage{dsfont}
\newcommand{\indicator}{\mathds{1}}
\newcommand{\Hess}{\textnormal{Hess}}

\newif\ifarxiv

\usepackage{enumitem}
\setlist{nolistsep}
\setlist{nosep}

%% file: abstract.tex
\begin{abstract}
This paper proposes Mutation-Driven Multiplicative Weights Update (M2WU) for learning an equilibrium in two-player zero-sum normal-form games and proves that it exhibits the last-iterate convergence property in both full and noisy feedback settings. In the former, players observe their exact gradient vectors of the utility functions. In the latter, they only observe the noisy gradient vectors. Even the celebrated Multiplicative Weights Update (MWU) and Optimistic MWU (OMWU) algorithms may not converge to a Nash equilibrium with noisy feedback. On the contrary, M2WU exhibits the last-iterate convergence to a stationary point near a Nash equilibrium in both feedback settings. We then prove that it converges to an exact Nash equilibrium by iteratively adapting the mutation term. We empirically confirm that M2WU outperforms MWU and OMWU in exploitability and convergence rates.
\end{abstract}

%% file: main.tex
\ifarxiv
\else
\begin{figure*}[t!]
    \centering
    \begin{minipage}[t]{0.24\textwidth}
        \centering
        \includegraphics[width=\linewidth]{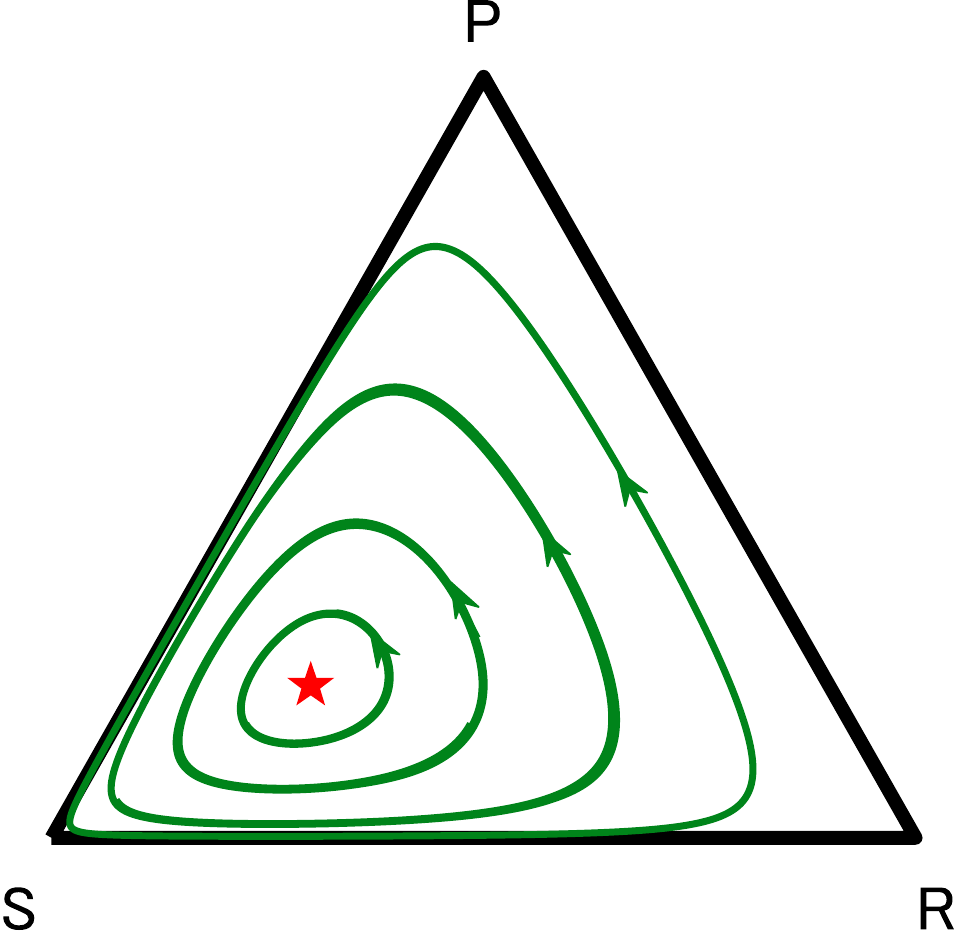}
        \subcaption{RD}\label{fig:FTRL-trajectory}
    \end{minipage}
    \begin{minipage}[t]{0.24\textwidth}
        \centering
        \includegraphics[width=\linewidth]{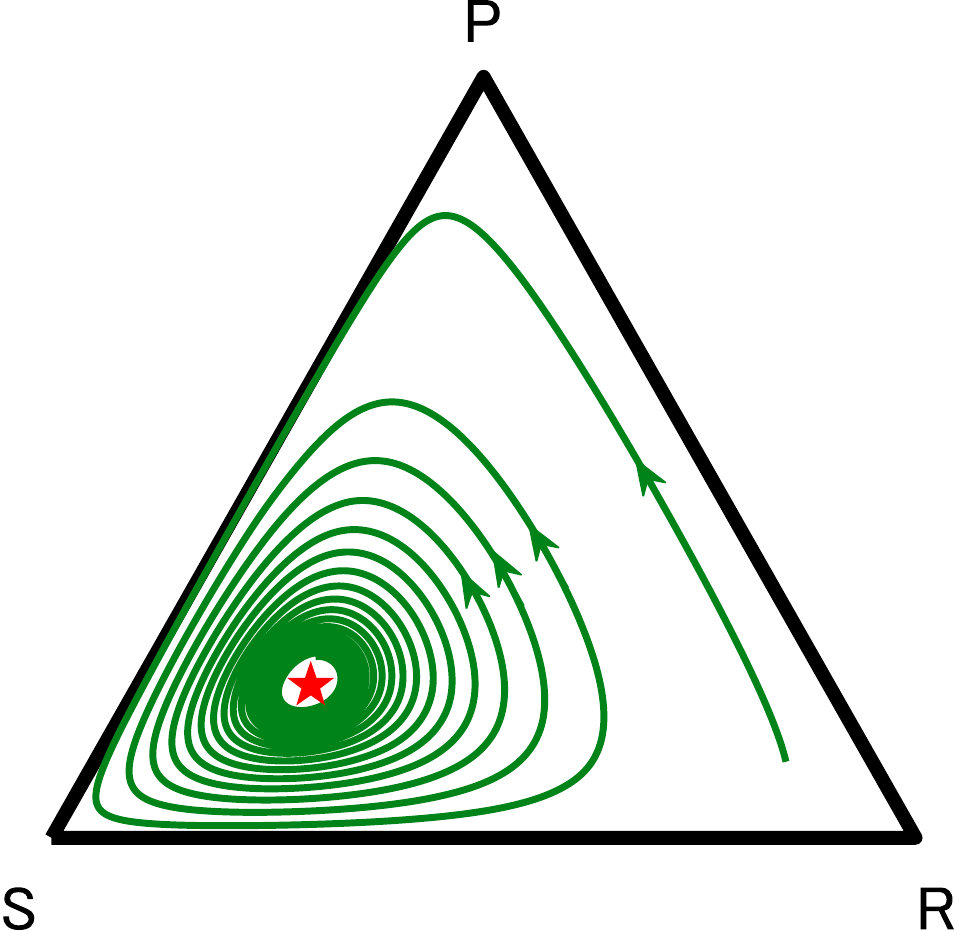}
        \subcaption{RMD ($\mu=0.01$)}\label{fig:M-FTRL-trajectory-0.01}
    \end{minipage}
    \begin{minipage}[t]{0.24\textwidth}
        \centering
        \includegraphics[width=\linewidth]{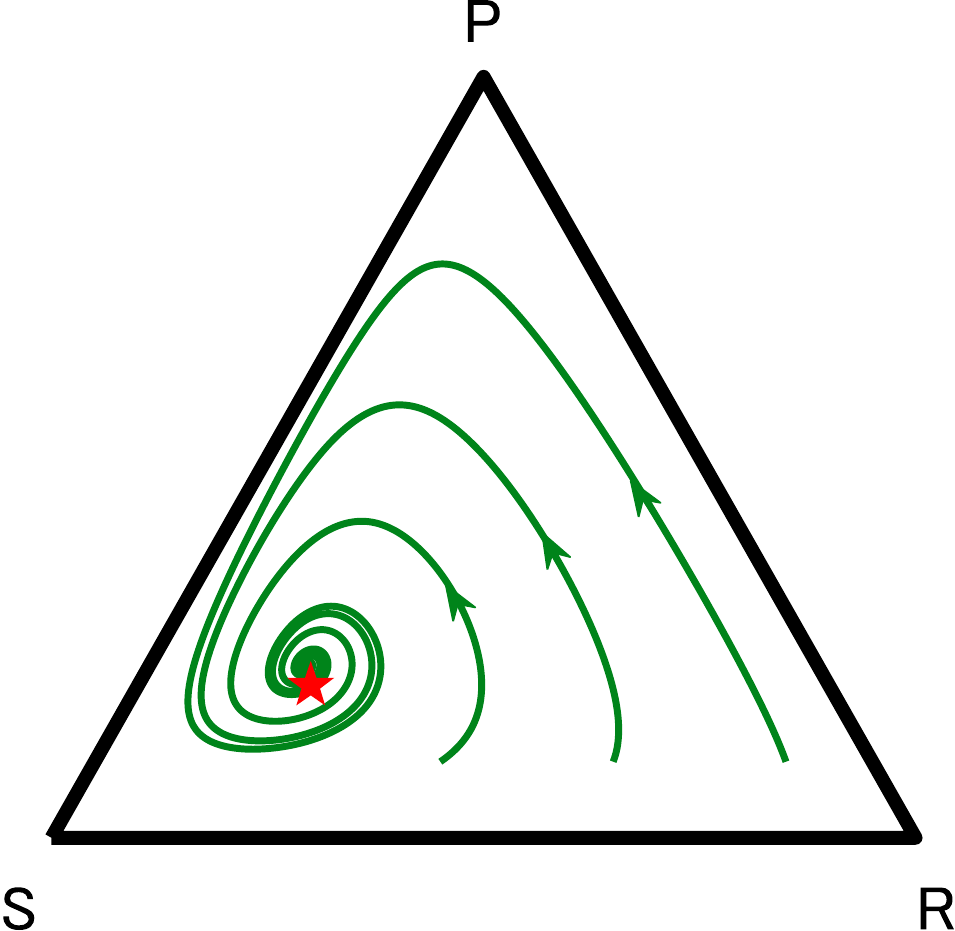}
        \subcaption{RMD ($\mu=0.1$)}\label{fig:M-FTRL-trajectory-0.1}
    \end{minipage}
    \begin{minipage}[t]{0.24\textwidth}
        \centering
        \includegraphics[width=\linewidth]{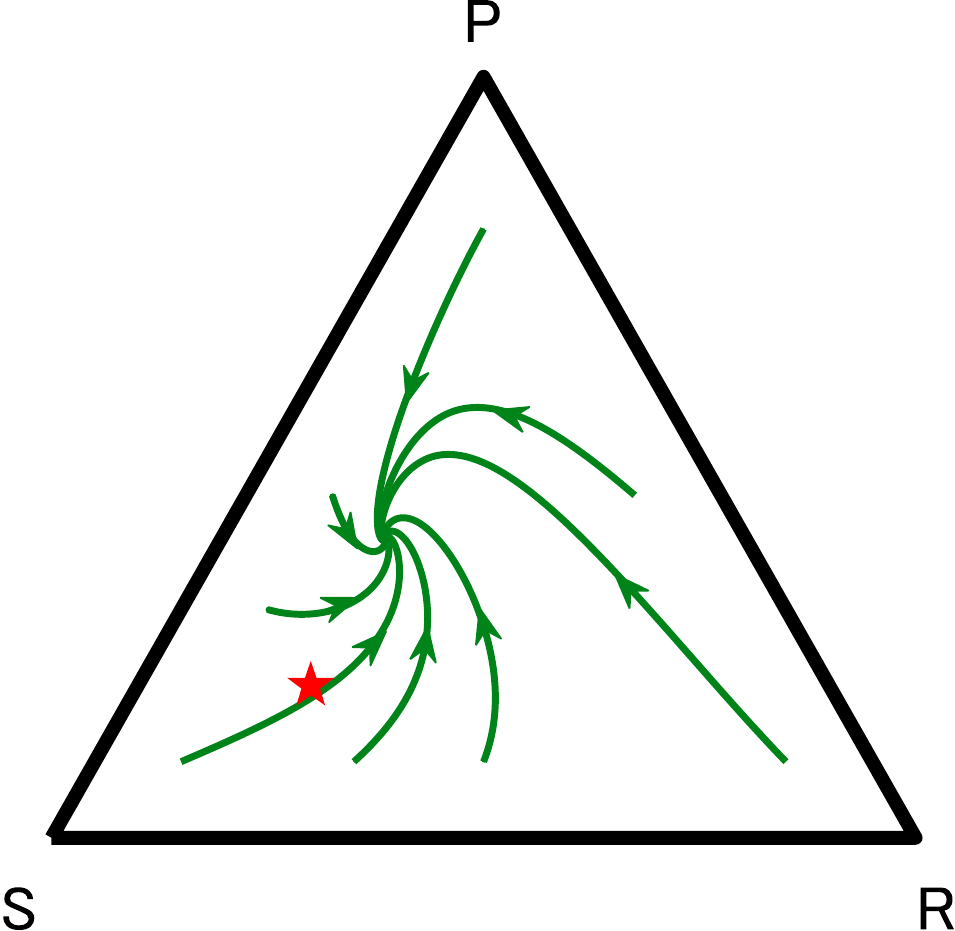}
        \subcaption{RMD ($\mu=1.0$)}\label{fig:M-FTRL-trajectory-1.0}
    \end{minipage}
    \caption{
    Learning dynamics of RD and RMD in biased Rock-Paper-Scissors (the game matrix is given by $[[0,-3,1],~[3,0,-1],~[-1,1,0]]$.
    The red star represents the Nash equilibrium point of the game.
    }
    \label{fig:learningdynamics}
\end{figure*}
\fi

\section{INTRODUCTION}
This paper considers learning algorithms for finding an (approximate) equilibrium in two-player zero-sum games.
Motivated by the training for Generative Adversarial Networks (GANs) \citep{goodfellow2014generative} and multi-agent reinforcement learning \citep{busoniu2008comprehensive}, many algorithms have been developed to find a near-optimal solution to minimax problems~\citep{blum:agt:2007,daskalakis2017training} in the form of $\min_x\max_y f(x,y)$.

In this context, no-regret learning, which minimizes regret in repeated decisions, has been extensively studied~\citep{banerjee2005efficient,zinkevich2007regret,daskalakis2011near}.
These algorithms, including the well-known Multiplicative Weights Update (MWU), exhibit the \textit{average-iterate convergence} by minimizing the regret of each player; that is, the averaged strategies over iterations converge to the minimax solution (the Nash equilibrium).
Still, it has been shown that the actual trajectory of updated strategies diverges or cycles~\citep{mertikopoulos2018cycles,bailey2018multiplicative}. 
This feature is unsatisfactory because averaging may require non-negligible amounts of memory and computation for large games, or averaging introduces additional error when the nonlinear function approximation is used, as in the case of training GANs.

This paper focuses on whether the actual sequence of updated strategies converges to an equilibrium, i.e., the \textit{last-iterate convergence}, which is inevitably a stronger notion than the average-iterate convergence. 
A series of \textit{optimistic} no-regret learning algorithms is proven to exhibit the last-iterate convergence~\citep{daskalakis2017training,mertikopoulos2018optimistic}.
In particular, the Optimistic MWU (OMWU) algorithm is guaranteed to converge to a Nash equilibrium at an exponential rate \citep{daskalakis2018last,wei2020linear}.
However, existing guarantees require that players observe the exact gradient vectors of their utility functions at each iteration, which we call \textit{full feedback}.  

We generalize the full feedback setting to the \textit{noisy feedback} setting, where players can only observe the gradient vectors with additive noise at each iteration; this setting is also called \textit{semi-bandit feedback}. 
For this setting, the celebrated OMWU is not guaranteed to have the last-iterate convergence and may diverge or enter a limit cycle, as shown in Figure~\ref{fig:cycle with OMWU}. 
It has been so far guaranteed only in some restricted games, such as those with a strict Nash equilibrium, in the noisy feedback setting~\citep{Heliou2017Learning,giannou2021convergence}. 

To this end, we propose \textit{Mutant} MWU\footnote{An implementation of our method is available at \url{https://github.com/CyberAgentAILab/m2wu}.} (M2WU) as the first learning algorithm that enjoys the last-iterate convergence with noisy feedback. 
M2WU is inspired by the fact that MWU is tantamount to replicator dynamics (RD), which is widely used in evolutionary game theory~\citep{borgers:jet:1997,Bloembergen2015}. 
Our M2WU is designed so that it corresponds to replicator-mutator dynamics (RMD)~\citep{Hofbauer1998,hofbauer:mor:2009,zagorsky:plosone:2013,Bauer:2019}, where each player may mutate his/her action. RMD has a unique stationary point that is asymptotically stable. Then, introducing mutation stabilizes the dynamics and empirically makes numerical errors in computation small~\citep{zagorsky:plosone:2013}. Figure~\ref{fig:learningdynamics} demonstrates that RMD clearly converges to a near-equilibrium in a biased Rock-Paper-Scissors game, while RD oscillates around an equilibrium. Our M2WU inherits these advantages via an additional mutation term. 

Starting with the full feedback case, we show that M2WU with a constant learning rate converges to a stationary point of RMD, which is known to be an approximate Nash equilibrium.
The amount of approximation is specified by the mutation rates.
Convergence occurs at an exponentially fast rate. 
Although OMWU achieves a similar convergence rate, it requires that the equilibrium in underlying games must be unique to establish convergence at that rate \citep{daskalakis2018last,wei2020linear}. 
The convergence guarantee of M2WU holds with noisy feedback under mild conditions for noise influencing the player's observations (zero-mean martingale noise with tame second-moment tails).
Specifically, M2WU converges to the stationary point almost surely. 
We utilize the fact that M2WU forms a continuous-time dynamics (RMD) and the existence of its Lyapunov function. 
In contrast, the existing convergence proof of OMWU depends on the path length of the observed gradient vectors~\citep{mertikopoulos2018optimistic,wei2020linear}, making such a guarantee with noise difficult. 

Surprisingly, in both feedback settings, we successfully establish convergence to an exact Nash equilibrium via iteratively adapting the mutation term according to the recently maintained strategy by M2WU.
To the best of our knowledge, the proposed M2WU with an appropriate choice of the update interval is the first to exhibit the last-iterate convergence to an exact Nash equilibrium with noisy feedback. 
We further empirically demonstrate that M2WU outperforms MWU and OMWU in several games in exploitability and convergence rate, regardless of which feedback is applied.

\ifarxiv
\begin{figure*}[t!]
    \centering
    \begin{minipage}[t]{0.24\textwidth}
        \centering
        \includegraphics[width=\linewidth]{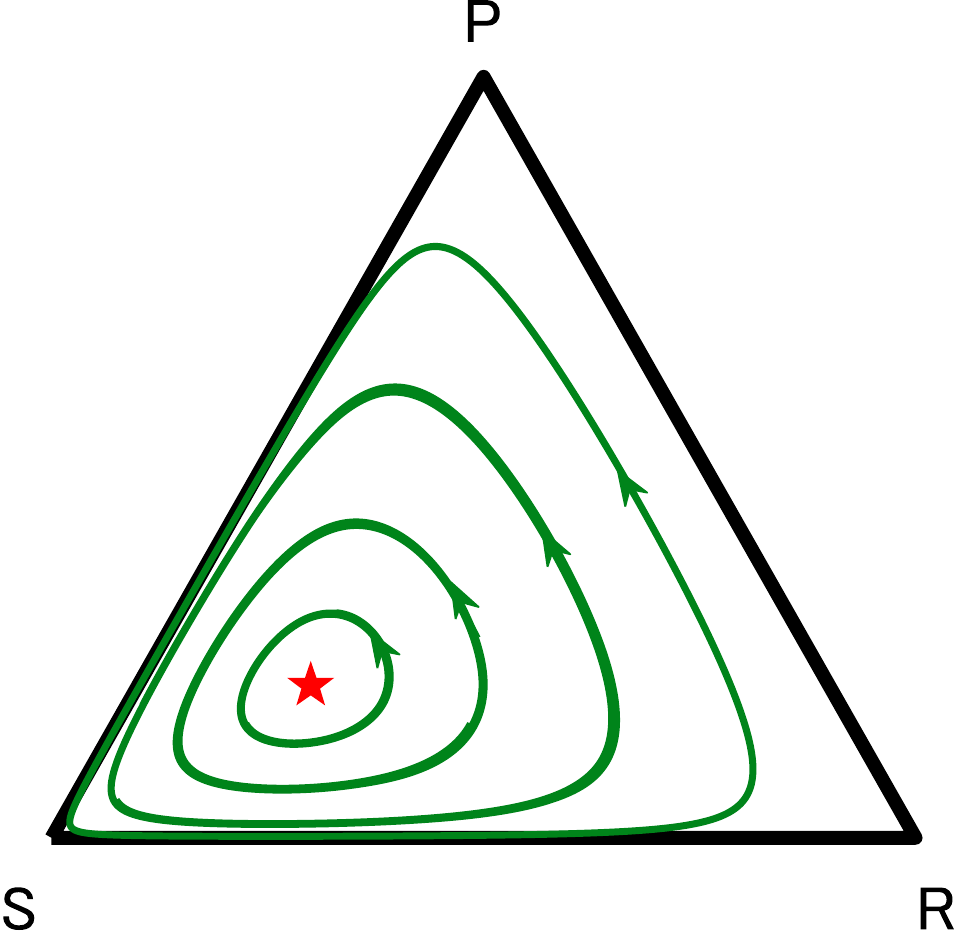}
        \subcaption{RD}\label{fig:FTRL-trajectory}
    \end{minipage}
    \begin{minipage}[t]{0.24\textwidth}
        \centering
        \includegraphics[width=\linewidth]{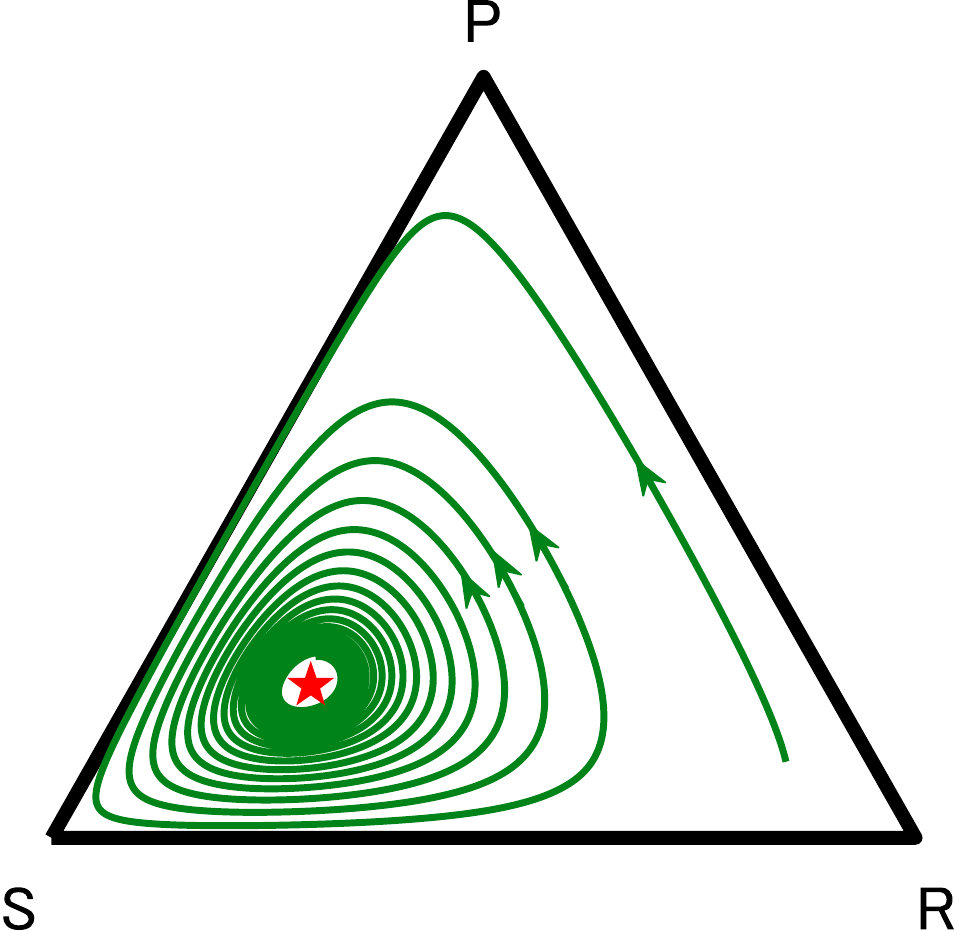}
        \subcaption{RMD ($\mu=0.01$)}\label{fig:M-FTRL-trajectory-0.01}
    \end{minipage}
    \begin{minipage}[t]{0.24\textwidth}
        \centering
        \includegraphics[width=\linewidth]{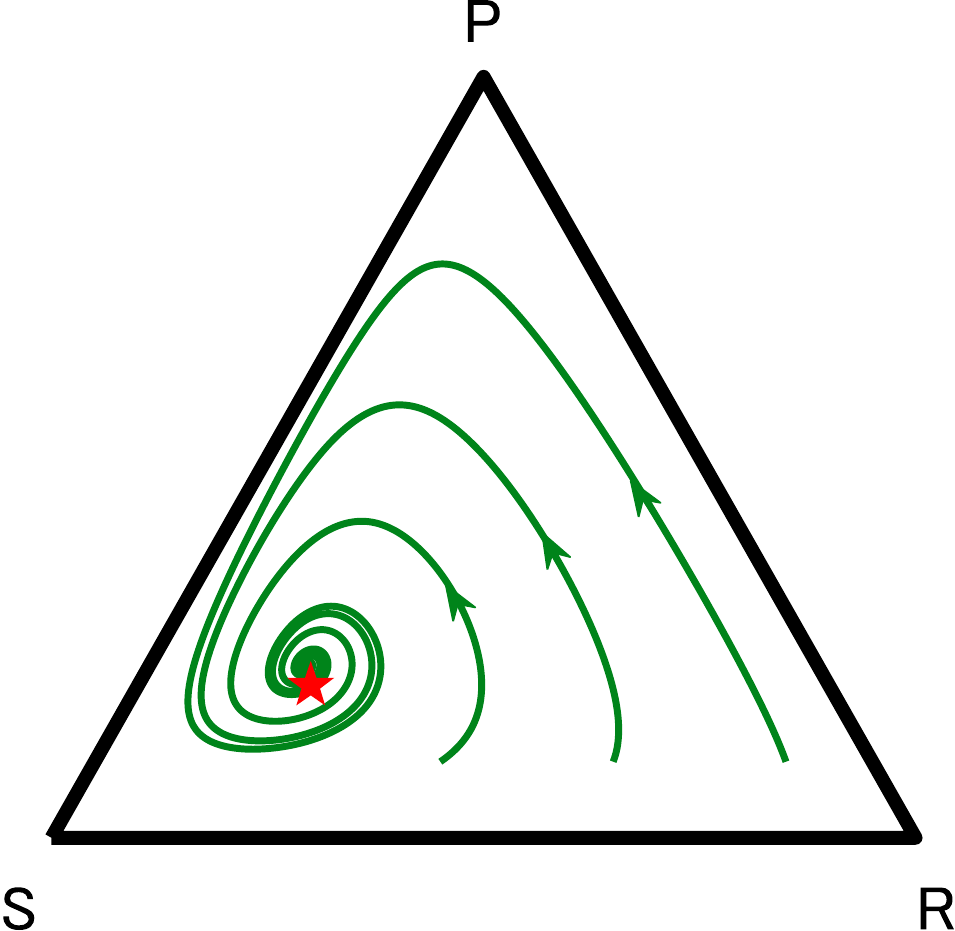}
        \subcaption{RMD ($\mu=0.1$)}\label{fig:M-FTRL-trajectory-0.1}
    \end{minipage}
    \begin{minipage}[t]{0.24\textwidth}
        \centering
        \includegraphics[width=\linewidth]{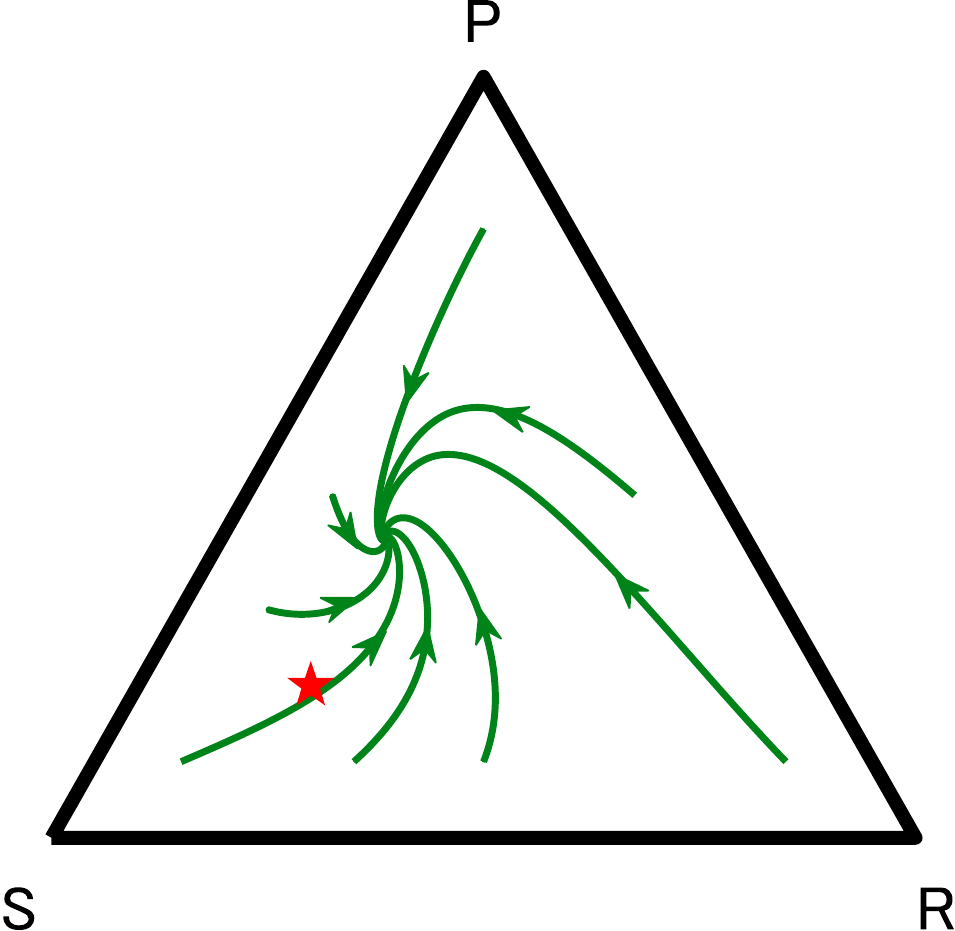}
        \subcaption{RMD ($\mu=1.0$)}\label{fig:M-FTRL-trajectory-1.0}
    \end{minipage}
    \caption{
    Learning dynamics of RD and RMD in biased Rock-Paper-Scissors (the game matrix is given by $[[0,-3,1],~[3,0,-1],~[-1,1,0]]$.
    The red star represents the Nash equilibrium point of the game.
    }
    \label{fig:learningdynamics}
\end{figure*}
\fi

\section{RELATED LITERATURE}
\paragraph{Last-Iterate Convergence with Full Feedback.}
Recently, various optimistic learning algorithms \citep{rakhlin2013online,rakhlin2013optimization} such as optimistic Follow the Regularized Leader (FTRL) and optimistic Mirror Descent have been proposed, and their last-iterate convergence guarantees are proven with full feedback.
In particular, last-iterate convergence for OMWU \citep{daskalakis2018last,wei2020linear,lei2021last,farina2022kernelized}, Optimistic Gradient Descent Ascent (OGDA) \citep{daskalakis2017training,mertikopoulos2018optimistic,daskalakis2018limit,liang2019interaction,golowich2020tight,wei2020linear,de2022convergence}, and extra-gradient algorithms \citep{golowich2020last,mokhtari2020unified,cai2022tight} have been proven in various settings such as minimax optimization and monotone games.
Some studies have proposed alternative approaches that exhibit last-iterate convergence by perturbing each player's utility function via strongly convex functions \citep{cen2021fast,perolat2021poincare,liu2022power,bernasconi2022last} or by utilizing the asymmetric information assumption \citep{nguyen2021last}.
Notably, \citet{abe2022mutationdriven} analyze a continuous-time version of M2WU. 
However, the last-iterate convergence properties 
are guaranteed only with full feedback, not with noisy feedback.

\paragraph{Last-Iterate Convergence with Noisy Feedback.}
A few studies have been done to prove last-iterate convergence with noisy feedback. 
Most existing studies discuss last-iterate convergence under the assumption that the game's equilibrium is a pure (or strict).
\citet{Heliou2017Learning} prove the convergence of an MWU-based algorithm with noise for the potential game, in which the game always has a pure Nash equilibrium, with the help of the stochastic approximation technique. There are also analyses with FTRL-based algorithms with noise \citep{giannou2021convergence, giannou2021Survival}.
Such results have been obtained under other strong assumptions, such as strict (or strong) monotonicity \citep{bravo2018bandit, hsieh2019ontheconvergence,kannan2019optimal,azizian2021LastIterate}, strict variational stability \citep{mertikopoulos2018optimistic, mertikopoulos2019learning}, and unconstrained action set \citep{hsieh2022no}.
Another approach is to use a two-time scaling, i.e., fixing the strategies of the two players to obtain sufficient samples for accurate estimates of the expected value of the utility, e.g., \citet{wei2021last}.

\section{PRELIMINARIES}
\subsection{Two-Player Zero-Sum Normal-Form Game}
A two-player normal-form game is defined as $\langle A_1, A_2, u_1, u_2\rangle$, where $A_i$ and $u_i:A_1\times A_2\to [-u_{\max}, u_{\max}]$ denote the finite action set for each player $i\in \{1,2\}$ and a utility function for player $i$, respectively.
In a two-player zero-sum normal-form game, there are conditions on the utility function: $u_1(a_1, a_2)=-u_2(a_1, a_2)$ for all $a_1\in A_1$ and $a_2\in A_2$.
We denote $\Delta(A_i)= \{p \in [0, 1]^{|A_i|} ~|~ \sum_{a\in A_i} p(a) = 1\}$ as a probability simplex on $A_i$ and $\pi_i\in \Delta(A_i)$ as a {\it mixed strategy} for player $i$.
Further, we denote by $\pi=(\pi_1, \pi_2)$ the {\it strategy profile}.
For a given strategy profile $\pi$, the expected value of the utility for player $i$ is written as follows $v_i^{\pi}=\mathbb{E}_{a\sim \pi}\left[u_i(a_1, a_2)\right]$.
We also define the conditional expected utility with action $a_i\in A_i$ as $q^{\pi}_i(a_i)=\mathbb{E}_{a_{-i}\sim \pi_{-i}}[u_i(a_i, a_{-i}) | a_i]$, where $-i$ represents the opponent to the player $i$.
We denote $q_i^{\pi}=(q_i^{\pi}(a))_{a\in A_i}$ as the conditional expected utility vector.

\subsection{Nash Equilibrium and Exploitability}
A {\it Nash equilibrium} \citep{nash1951non} is a widely used solution concept for games. 
In a Nash equilibrium, no player can improve his/her expected utility by deviating from his/her specified strategy.
In two-player zero-sum normal-form games, a strategy profile $\pi^{\ast}=(\pi_1^{\ast}, \pi_2^{\ast})$ is called a Nash equilibrium if for any $\pi_1\in \Delta(A_1)$ and $\pi_2\in \Delta(A_2)$,
\begin{align*}
    v_1^{\pi_1^{\ast}, \pi_2} \geq v_1^{\pi_1^{\ast}, \pi_2^{\ast}} \geq v_1^{\pi_1, \pi_2^{\ast}}.
\end{align*}
We denote the set of Nash equilibria by $\Pi^{\ast}$.
An {\it $\epsilon$-Nash equilibrium} $(\pi_1, \pi_2)$ is an approximation of a Nash equilibrium, which satisfies the following inequality:
\begin{align*}
    \max_{\tilde{\pi}_1\in \Delta(A_1)}v_1^{\tilde{\pi}_1, \pi_2} + \max_{\tilde{\pi}_2\in \Delta(A_2)}v_2^{\pi_1, \tilde{\pi}_2} &
    \leq \epsilon.
\end{align*}
Furthermore, we define $\mathrm{explt}(\pi)=\max_{\tilde{\pi}_1\in \Delta(A_1)}v_1^{\tilde{\pi}_1, \pi_2} + \max_{\tilde{\pi}_2\in \Delta(A_2)}v_2^{\pi_1, \tilde{\pi}_2}$ as {\it exploitability} of the strategy profile $\pi$.
Exploitability is widely used to assess how close $\pi$ is to Nash equilibrium in two-player zero-sum games and always takes a non-negative value \citep{johanson2011accelerating,johanson2012finding,lockhart2019computing,timbers2020approximate,abe2020off}. 
A strategy profile $\pi$ has exploitability of $0$ if and only if $\pi$ is a Nash equilibrium.

\subsection{Problem Setting}
In this study, we consider a setting where the following process is repeated:
1) At each iteration $t \in \mathbb{N}$, each player $i\in \{1,2\}$ determines the (mixed) strategy $\pi_i^t\in \Delta(A_i)$ based on the previously observed feedback;
2) Each player $i$ observes the new feedback $\hat{q}_i^{\pi^t}$ with respect to the gradient vector of the expected utility function $\nabla_{\pi_i^t}v_i^{\pi^t}=q_i^{\pi^t}$.
This study considers two feedback settings: {\it full feedback} and {\it noisy feedback}.
In the full feedback setting, each player $i$ observes the conditional expected utility vector $q_i^{\pi^t}$ as feedback, i.e., $\hat{q}_i^{\pi^t}=q_i^{\pi^t}$.
In the noisy feedback setting, at each iteration $t$, each player observes the noisy conditional expected utility vector
\begin{align*}
     \hat{q}_i^{\pi^t}(a) = {q}_i^{\pi^t}(a) + \xi^t(a) \quad \text{for } a \in A_i,
\end{align*}
where the sequence of the noise vectors $ (\xi^t(a))_{a \in A_i}$ is independent over $a$ and $t$. 
This type of noise-additive setting is standard in recent research \citep{Heliou2017Learning, bravo2018bandit, giannou2021convergence, giannou2021Survival}.

Multiplicative Weights Update (MWU) is a widely used algorithm for learning a Nash equilibrium.
In MWU, each player $i$ updates her strategy $\pi_i^t$ at iteration $t$ as follows:
\begin{align*}
    \pi_i^{t+1}(a) &= \frac{\pi_i^t(a)\exp\left(\eta_t \hat{q}_i^{\pi^t}(a)\right)}{\sum_{a'\in A_i}\pi_i^t(a')\exp\left(\eta_t \hat{q}_i^{\pi^t}(a')\right)},
\end{align*}
where $\eta_t>0$ is a learning rate.

\subsection{Other Notations}
We denote the interior of $\Delta(A_i)$ by $\Delta^{\circ}(A_i) = \{p \in \Delta(A_i) ~|~ \forall a\in A_i, ~p(a) > 0\}$.
The Kullback-Leibler divergence is defined by $\mathrm{KL}(x, y)=\sum_i x_i\ln \frac{x_i}{y_i}$.
Besides, with a slight abuse of notation, we denote the sum of Kullback-Leibler divergences as $\mathrm{KL}(\pi, \pi')=\sum_{i=1}^2\mathrm{KL}(\pi_i, \pi_i')$.

\section{MUTANT MWU}
\label{sec:m2wu}
This section proposes a mutant Multiplicative Weights Update (M2WU) algorithm.
M2WU is a variant of MWU, which adds a mutation (perturbation) term to the gradient vector.
Specifically, M2WU updates each player's strategy by the following update rule:
\begin{align}
    \label{eq:m2wu}
    \pi_i^{t+1}(a) &= \frac{\pi_i^t(a)\exp\left(\eta_t q_i^{\mu,t}(a)\right)}{\sum_{a'\in A_i}\pi_i^t(a')\exp\left(\eta_t q_i^{\mu,t}(a')\right)}, \\
    q_i^{\mu,t}(a) &= \hat{q}_i^{\pi^t}(a) + \frac{\mu}{\pi_i^t(a)}\left(r_i(a)-\pi_i^t(a)\right) \nonumber,
\end{align}
where $\mu \in (0,1]$ is the {\it mutation rate}, and $r_i\in \Delta^{\circ}(A_i)$ is the {\it reference strategy}.
We call $q_i^{\mu,t}=(q_i^{\mu,t}(a))_{a\in A_i}$ the {\it mutation gradient}.

The pseudo-code of M2WU is Algorithm~\ref{alg:m2wu} with $N=\infty$.

The mutation gradient $q_i^{\mu,t}$ is inspired by RMD, which is governed by the following ordinary differential equation:
\ifarxiv
\begin{align}
\label{eq:rmd}
\begin{aligned}
    \frac{d}{dt}\pi_i^t(a) = \pi_i^t(a)\left(q_i^{\pi^t}(a) - v_i^{\pi^t}\right) + \mu\left(r_i(a)-\pi_i^t(a)\right).
\end{aligned}  \tag{RMD}
\end{align}
\else
\begin{align}
\label{eq:rmd}
\begin{aligned}
    \frac{d}{dt}\pi_i^t(a) =& \pi_i^t(a)\left(q_i^{\pi^t}(a) - v_i^{\pi^t}\right) + \mu\left(r_i(a)-\pi_i^t(a)\right).
\end{aligned}  \tag{RMD}
\end{align}
\fi
RMD is the continuous-time version of M2WU and has been reported to stabilize the learning dynamics \citep{bomze:geb:1995,Bauer:2019}.
Intuitively, the mutation term $\mu\left(r_i(a)-\pi_i^t(a)\right)$ of RMD has the role of slightly shifting the trajectory of strategies from one of RD.
This allows the trajectory to escape from the cyclic orbits and allows it to converge to an approximate Nash equilibrium as a stationary point of RMD.
Figure \ref{fig:learningdynamics} illustrates the trajectories of RD and RMD with $\mu\in\{0.01, 0.1, 1.0\}$ in a biased version of the Rock-Paper-Scissors game.
From Figure \ref{fig:FTRL-trajectory}, the trajectory of RD cycles and fails to converge to a Nash equilibrium since the equilibrium is a mixed strategy with full support.
On the other hand, as depicted in Figures~\ref{fig:M-FTRL-trajectory-0.01}-\ref{fig:M-FTRL-trajectory-1.0}, RMD's trajectory converges to a unique stationary point in this game.
In fact, \citet{abe2022mutationdriven} proved this convergence property of RMD for any two-player zero-sum normal-form games.

We note that M2WU with a constant learning rate can be viewed as an instantiation of the discrete-time Mutant FTRL algorithm with entropy regularization \citep{abe2022mutationdriven}.
The squared $\ell^2$-norm regularization, as used in OGDA \citep{daskalakis2018limit,wei2020linear}, cannot be used with the mutation term since it allows strategies to run into the boundary of $\Delta(A_i)$ where the mutation term is undefined.
It is interesting future work to find appropriate mutation terms for other regularizers as well.

\section{CONVERGENCE TO AN APPROXIMATE NASH EQUILIBRIUM}
\label{sec:convergence}
This section mainly shows that the updated strategy profile $\pi^t$ converges to a stationary point of \eqref{eq:rmd}.
We denote the stationary point of (\ref{eq:rmd}) with fixed $\mu$ and $r=\left(r_i\right)_{i=1}^2$ by $\pi^{\mu,r}$.

\begin{figure}[t!]
\begin{algorithm}[H]
    \caption{M2WU for player $i$. The algorithm with $N=\infty$ corresponds to M2WU with a fixed reference strategy.}
    \label{alg:m2wu}
    \begin{algorithmic}[1]
    \REQUIRE{Learning rate sequence $\{\eta_t\}_{t \ge 0}$, mutation rate $\mu$, update frequency $N$, initial strategy $\pi_i^0$, initial reference strategy $r_i^0$}
    \STATE $k\gets 0, ~\tau \gets 0$
    \FOR{$t=0,1,2,\cdots$}
        \STATE Observe the (noisy) gradient vector $\hat{q}_i^{\pi^t}$.
        \FOR{$a\in A_i$}
            \STATE Compute next action probability $\pi_i^{t+1}(a)$ by $\pi_i^{t+1}(a) = \frac{\pi_i^t(a)\exp\left(\eta_t q_i^{\mu,t}(a)\right)}{\sum_{a'\in A_i}\pi_i^t(a')\exp\left(\eta_t q_i^{\mu,t}(a')\right)},$
            \STATE where $q_i^{\mu,t}(a) = \hat{q}_i^{\pi^t}(a) + \frac{\mu}{\pi_i^t(a)}\left(r_i^k(a)-\pi_i^t(a)\right)$
        \ENDFOR
        \STATE $\tau \gets \tau + 1$
        \IF{$\tau=N$}
            \STATE $k\gets k+1, ~\tau\gets 0$
            \STATE $r_i^k\gets \pi_i^t$
        \ENDIF
    \ENDFOR
    \end{algorithmic}
\end{algorithm}
\end{figure}

\subsection{Full Feedback Setting}
First, we establish the last-iterate convergence rate of M2WU with {\it full feedback}.
Recall that in the full feedback setting, each player $i$ observes the conditional expected utility vector $\hat{q}_i^{\pi^t}=q_i^{\pi^t}$ as feedback.
The following convergence result for M2WU with a constant learning rate $\eta_t=\eta$ is obtained in the full feedback setting:
\begin{theorem}
\label{thm:kl_convergence}
Let $\pi^{\mu,r}\in \prod_{i=1}^2 \Delta(A_i)$ be a stationary point of (\ref{eq:rmd}).
If we use the constant learning rate sequence in M2WU, $\forall t \ge 0 : \eta_t = \eta  \in (0,\min(\frac{\mu\alpha}{\mu^2\beta+\gamma},\zeta))$, the strategy $\pi^t$ updated by M2WU satisfies that for any initial strategy profile $\pi^{0}\in \prod_{i=1}^2\Delta^{\circ}(A_i)$ and $t\geq 0$:
\begin{align*}
    \mathrm{KL}(\pi^{\mu,r}, \pi^t) \leq \mathrm{KL}(\pi^{\mu,r}, \pi^0) (1- \eta(\mu\alpha - \eta(\mu^2\beta + \gamma)))^t,
\end{align*}
where $\alpha, \beta, \gamma$, and $\zeta$ are constants that depend only on $\pi_0$, $\pi^{\mu,r}$, and $r$.
\end{theorem}
This result means that for a fixed $\mu$ and $r$, $\pi^t$ converges to $\pi^{\mu,r}$ exponentially fast.
From this theorem, $\pi^t$ converges to a $2\mu$-Nash equilibrium because $\pi^{\mu,r}$ is a $2\mu$-Nash equilibrium \citep{Bauer:2019}:
\begin{corollary}
\label{cor:exploitability_convergence}
For any constant learning rate $\eta_t = \eta \in (0,\min(\frac{\mu\alpha}{\mu^2\beta+\gamma},\zeta))$, the exploitability for M2WU is bounded as:
\ifarxiv
\begin{align*}
    \mathrm{explt}(\pi^t) &\leq \mathrm{explt}(\pi^{\mu,r}) + 2u_{\max}\sqrt{\mathrm{KL}(\pi^{\mu,r}, \pi^0)}(1 - C)^{\frac{t}{2}} \\
    &\leq 2\mu + 2u_{\max}\sqrt{\mathrm{KL}(\pi^{\mu,r}, \pi^0)}(1 - C)^{\frac{t}{2}},
\end{align*}
\else
\begin{align*}
    \mathrm{explt}(\pi^t) &\leq \mathrm{explt}(\pi^{\mu,r}) + 2u_{\max}\sqrt{\mathrm{KL}(\pi^{\mu,r}, \pi^0)}(1 - C)^{\frac{t}{2}} \\
    &\leq 2\mu + 2u_{\max}\sqrt{\mathrm{KL}(\pi^{\mu,r}, \pi^0)}(1 - C)^{\frac{t}{2}},
\end{align*}
\fi
where $C=\eta(\mu\alpha - \eta(\mu^2\beta + \gamma))$, and $\alpha, \beta, \gamma$ and $\zeta$ are the same constants used in Theorem \ref{thm:kl_convergence}.
\end{corollary}
The proof of this corollary is shown in Appendix \ref{sec:appendix_proof_exploitability_convergence}.
We note that from the upper bound on the learning rate $\frac{\mu\alpha}{\mu^2\beta + \gamma}$ in Theorem \ref{thm:kl_convergence} and Corollary \ref{cor:exploitability_convergence}, $\eta$ should decrease in proportion to the decrease of $\mu$.
We will demonstrate this fact empirically in Figure \ref{fig:compare_mu_eta} in Section \ref{sec:experiments}.

\subsubsection{Proof Sketch of Theorem \ref{thm:kl_convergence}}
We sketch below the proof of Theorem \ref{thm:kl_convergence}.
Complete proofs for the theorem and associated lemmas are presented in Appendix \ref{sec:appendix_proof_kl_convergence}.

\paragraph{(1) Decomposing Single-Step Variation of $\mathrm{KL}(\pi^{\mu,r},\cdot)$.}
First, we derive the following difference equation for the Kullback-Leibler divergence between $\pi^{\mu,r}$ and $\pi^t$:
\ifarxiv
\begin{align}
\label{eq:kl_diff_decomposition}
    \mathrm{KL}(\pi^{\mu,r}, \pi^{t+1}) - \mathrm{KL}(\pi^{\mu,r}, \pi^t) = \eta \underbrace{\sum_{i=1}^2\left(v_i^{\pi_i^t, \pi_{-i}^{\mu,r}} + \mu - \mu\sum_{a\in A_i}r_i(a)\frac{\pi_i^{\mu,r}(a)}{\pi_i^t(a)}\right)}_{(\mathrm{A})} + \underbrace{\mathrm{KL}(\pi^t, \pi^{t+1})}_{(\mathrm{B})}.
\end{align}
\else
\begin{align}
    &\mathrm{KL}(\pi^{\mu,r}, \pi^{t+1}) - \mathrm{KL}(\pi^{\mu,r}, \pi^t) =  \label{eq:kl_diff_decomposition}
    \\
    & \eta \underbrace{\sum_{i=1}^2\!\left(\!v_i^{\pi_i^t, \pi_{-i}^{\mu,r}} \!+\! \mu \!-\! \mu\sum_{a\in A_i}r_i(a)\frac{\pi_i^{\mu,r}(a)}{\pi_i^t(a)}\!\right)\!}_{(\mathrm{A})} \!+ \underbrace{\mathrm{KL}(\pi^t, \pi^{t+1})}_{(\mathrm{B})}.\nonumber
\end{align}
\fi
Equation \eqref{eq:kl_diff_decomposition} stems from the fact that for any $\pi\in \prod_{i=1}^2\Delta(A_i)$, $ \mathrm{KL}(\pi, \pi^t) = \sum_{i=1}^2(\eta\langle \sum_{s=1}^{t-1} q_i^{\mu,s}, \pi_i^t - \pi_i\rangle - \psi_i(\pi_i^t) + \psi_i(\pi_i))$, where $\psi_i(p)=\sum_{a\in A_i}p(a)\ln p(a)$.
For the details of the proof, see Appendix~\ref{sec:appendix_proof_kl_convergence}.
Hereafter, we quantify the terms (A) and (B), respectively.

\paragraph{(2) Equivalence Notation of (A) in Quasi-Metric Form.}
First, we prove that the term (A) can be rewritten by the (pseudo) metric between $\pi^{\mu,r}$ and $\pi^t$.
\begin{lemma}
\label{lem:rmd_divergence}
Let $\pi^{\mu,r}\in \prod_{i=1}^2 \Delta(A_i)$ be a stationary point of (\ref{eq:rmd}).
Then, $\pi^t$ updated by M2WU satisfies that:
\ifarxiv
\begin{align*}
    \sum_{i=1}^2\left(v_i^{\pi_i^t, \pi_{-i}^{\mu,r}} + \mu - \mu\sum_{a\in A_i}r_i(a)\frac{\pi_i^{\mu,r}(a)}{\pi_i^t(a)}\right) = -\mu \sum_{i=1}^2\sum_{a\in A_i}r_i(a)\left(\sqrt{\frac{\pi_i^t(a)}{\pi_i^{\mu,r}(a)}} - \sqrt{\frac{\pi_i^{\mu,r}(a)}{\pi_i^t(a)}}\right)^2.
\end{align*}
\else
\begin{align*}
    &\sum_{i=1}^2\left(v_i^{\pi_i^t, \pi_{-i}^{\mu,r}} + \mu - \mu\sum_{a\in A_i}r_i(a)\frac{\pi_i^{\mu,r}(a)}{\pi_i^t(a)}\right) \\
    &= -\mu \sum_{i=1}^2\sum_{a\in A_i}r_i(a)\left(\sqrt{\frac{\pi_i^t(a)}{\pi_i^{\mu,r}(a)}} - \sqrt{\frac{\pi_i^{\mu,r}(a)}{\pi_i^t(a)}}\right)^2.
\end{align*}
\fi
\end{lemma}
This result can be shown by using Lemma 5.6 in \citet{abe2022mutationdriven}.

\paragraph{(3) Quasi-Metric Upper Bound on the Term (B).}
Next, we upper bound the Kullback-Leibler divergence between $\pi^t$ and $\pi^{t+1}$ by the (pseudo) metric between $\pi^{\mu,r}$ and $\pi^t$:
\begin{lemma}
\label{lem:kl_path_ub}
For any fixed learning rate $\eta_t = \eta \in (0, \zeta)$, M2WU ensures for any $t\geq 0$:
\ifarxiv
\begin{align*}
    \mathrm{KL}(\pi^t, \pi^{t+1}) \leq 8 \eta^2u_{\max}^2\sum_{i=1}^2 \|\pi_i^t - \pi_i^{\mu,r}\|_1^2 + 8 \eta^2 \mu^2\sum_{i=1}^2 \sum_{a\in A_i}r_i(a)^2\left(\frac{1}{\pi_i^{\mu,r}(a)} - \frac{1}{\pi_i^t(a)}\right)^2,
\end{align*}
\else
\begin{align*}
    &\mathrm{KL}(\pi^t, \pi^{t+1}) \leq 8 \eta^2u_{\max}^2\sum_{i=1}^2 \|\pi_i^t - \pi_i^{\mu,r}\|_1^2 \\
    & + 8 \eta^2 \mu^2\sum_{i=1}^2 \sum_{a\in A_i}r_i(a)^2\left(\frac{1}{\pi_i^{\mu,r}(a)} - \frac{1}{\pi_i^t(a)}\right)^2,
\end{align*}
\fi
where $\zeta$ is the same constant in Theorem \ref{thm:kl_convergence}.
\end{lemma}

\paragraph{(4) Putting It All Together.}
By combining (\ref{eq:kl_diff_decomposition}), Lemma \ref{lem:rmd_divergence}, and Lemma \ref{lem:kl_path_ub}, we get:
\ifarxiv
\begin{align*}
    \mathrm{KL}(\pi^{\mu,r}, \pi^{t+1}) - \mathrm{KL}(\pi^{\mu,r}, \pi^t) \leq& - \eta \mu \sum_{i=1}^2\sum_{a\in A_i}r_i(a)\left(\sqrt{\frac{\pi_i^t(a)}{\pi_i^{\mu,r}(a)}} - \sqrt{\frac{\pi_i^{\mu,r}(a)}{\pi_i^t(a)}}\right)^2 \\
    &+ 8 \eta^2 \mu^2\sum_{i=1}^2 \sum_{a\in A_i}r_i(a)^2\left(\frac{1}{\pi_i^{\mu,r}(a)} - \frac{1}{\pi_i^t(a)}\right)^2 \\
    &+ 8 \eta^2 u_{\max}^2\sum_{i=1}^2 \|\pi_i^t - \pi_i^{\mu,r}\|_1^2.
\end{align*}
\else
\begin{align*}
    \mathrm{KL}&(\pi^{\mu,r}, \pi^{t+1}) - \mathrm{KL}(\pi^{\mu,r}, \pi^t) \\
    \leq& - \eta \mu \sum_{i=1}^2\sum_{a\in A_i}r_i(a)\left(\sqrt{\frac{\pi_i^t(a)}{\pi_i^{\mu,r}(a)}} - \sqrt{\frac{\pi_i^{\mu,r}(a)}{\pi_i^t(a)}}\right)^2 \\
    &+ 8 \eta^2 \mu^2\sum_{i=1}^2 \sum_{a\in A_i}r_i(a)^2\left(\frac{1}{\pi_i^{\mu,r}(a)} - \frac{1}{\pi_i^t(a)}\right)^2 \\
    &+ 8 \eta^2 u_{\max}^2\sum_{i=1}^2 \|\pi_i^t - \pi_i^{\mu,r}\|_1^2.
\end{align*}
\fi

From Pinsker's inequality \citep{Tsybakov:1315296}, we can upper bound $\sum_{i=1}^2\sum_{a\in A_i}r_i(a)^2\left(\frac{1}{\pi_i^{\mu,r}(a)} - \frac{1}{\pi_i^t(a)}\right)^2$ and $\sum_{i=1}^2\|\pi_i^t - \pi_i^{\mu,r}\|_1^2$ by $\mathrm{KL}(\pi^{\mu,r}, \pi^t)$, respectively.
Furthermore, from Jensen's inequality, we can lower bound $\sum_{i=1}^2\sum_{a\in A_i}r_i(a)\left(\sqrt{\frac{\pi_i^t(a)}{\pi_i^{\mu,r}(a)}} - \sqrt{\frac{\pi_i^{\mu,r}(a)}{\pi_i^t(a)}}\right)^2$ by $\mathrm{KL}(\pi^{\mu,r}, \pi^t)$.
Therefore, for $\eta\in (0, \frac{\mu\alpha}{\mu^2\beta + \gamma})$, we have:
\begin{align*}
    \!\mathrm{KL}(\pi^{\mu,r}, \pi^{t+1})\!\leq\! \left(1 \!-\!\eta(\mu\alpha \!-\! \eta(\mu^2\beta + \gamma))\right) \!\mathrm{KL}(\pi^{\mu,r}, \pi^t).
\end{align*}
Thus, by mathematical induction, the statement of the theorem is concluded. \qed


\subsection{Noisy Feedback Setting}
Next, we consider a {\it noisy feedback} setting, where each player's observation is affected by noise. We assume the following mild condition on the noise distribution.
Let $\mathcal{F}_t$ be the $\sigma$-algebra generated by the random sequence $(\pi_i^s, (\hat{q}_i^{\pi^s}(a))_{a \in A_i})_{s = 1, \ldots, t}$.
\begin{assumption}\label{asm:noise_semibandit}
For all player $i \in \{1,2\}$, the noise process $ (\xi^t(a))_{a \in A_i}$ satisfies the following two conditions:
\begin{itemize}[nolistsep, topsep=-4pt]
\setlength{\itemsep}{0pt} \setlength{\parskip}{0pt}%
    \item[(i)] Zero-mean: $ \mathbb{E}[\xi^t(a) \mid \mathcal{F}_{t-1}] = 0 , \forall a \in A_i, \forall t \ge 1,$ almost surely.
    \item[(ii)] Moderate tails: For any $x>0$, $\mathbb{P}[|\xi^t(a)|^2 \ge x \mid \mathcal{F}_{t-1}] \le C/x^\kappa, \forall a \in A_i, \forall t \ge 1,$ almost surely, with some constants $C>0$ and $\kappa > 2$.
\end{itemize}
\end{assumption}
Assumption~\ref{asm:noise_semibandit} (i) means that the observation is unbiased: $ \mathbb{E}[ \hat{q}_i^{\pi^t}(a)  \mid \mathcal{F}_{t-1}] = q_i^{\pi^t}(a)$, almost surely.  Assumption~\ref{asm:noise_semibandit} (ii) is a relatively weak assumption on the noise and is satisfied by a wide range of distributions, including bounded, sub-Gaussian, and sub-exponential distributions \citep{Heliou2017Learning}. Assumption~\ref{asm:noise_semibandit} (ii) also implies that the variance of the noise is upper bounded by some constant.
The following convergence results are obtained in the noisy feedback setting. 

\begin{theorem}\label{thm:conv_semibandit}
Suppose that there exists a constant $D>0$ such that $\pi_i^t(a_i) > D$ for all $i\in \{1,2\}$, $a_i\in A_i$, and $t\geq 0$.
Under Assumption~\ref{asm:noise_semibandit}, the strategy $\pi^t$ updated by M2WU with the step size $\eta_t \propto t^{-\lambda}$ for some constant $\lambda \in (1/\kappa, 1]$ converges to the stationary point $\pi^{\mu,r}$ almost surely. 
\end{theorem}

For the noise process that can take an arbitrarily large value of $\kappa$, the value of $\lambda$ can be arbitrarily close to $0$. 
This suggests that learning is possible with a nearly-constant learning rate sequence for sub-Gaussian, sub-exponential, and bounded distributions. The proof of Theorem~\ref{thm:conv_semibandit} is based on the method of stochastic approximation \citep{benaim1999dynamics,borkar2009stochastic}.
Here, we only present a sketch of the proof of Theorem~\ref{thm:conv_semibandit}.
A complete proof is presented in Appendix \ref{sec:appendix_proof_conv_semibandit}.

\subsubsection{Proof Sketch of Theorem~\ref{thm:conv_semibandit}}

First, we show that the update rule of the strategy $\pi_i^{t}$ by M2WU is an approximate Robbins-Monro algorithm \citep{robbins1951stochastic,benaim1999dynamics}. 
For this purpose, we use Taylor's theorem to rewrite the strategy update formula as the following equation.
$$
\pi^{t+1}_i(a_i)  =  \pi^{t}_i(a_i) + \eta_t (F(\pi_i^t) + U_t + \hat{\phi}_t).
$$
It is relatively easy to check that $ F(\pi_i^t)$ is a continuous function and $U_t$ is a martingale difference sequence. 
The fact that the Hessian of the logit function is bounded by a constant indicates that $\hat{\phi}_t$ is on the order of $\eta_t \|\hat{q}_i^{\mu, t} \|_2^2$. 
Using Assumption~\ref{asm:noise_semibandit}, we can show that $\phi_t = \mathcal{O}(\eta_t t^p) \to 0$, almost surely, where $p \in (0, 1/\kappa) $. Thus, we can conclude that $ \{ \pi^t_i\}_{t\ge 1}$ is an approximate Robbins-Monro algorithm and an asymptotic pseudo-trajectory of the replicator mutator dynamics \eqref{eq:rmd}. 

From Theorem~5.2 in \citet{abe2022mutationdriven}, there exists a strict Lyapunov function of \eqref{eq:rmd}, and the stationary point of \eqref{eq:rmd} is unique. 
These conditions allow us to apply the results of \citet{benaim1999dynamics} and conclude that  $\pi_i^{t}$ converges to the stationary point almost surely. \qed

\begin{remark}[Comparison to optimistic algorithms]
\label{rm:comparison_to_omwu}
\normalfont
In the previous work on optimistic algorithms such as OMWU and OGDA \citep{mertikopoulos2018optimistic,wei2020linear}, the proofs for the last-iterate convergence depend heavily on the path length of the gradient vectors $\sum_{t=1}^T\sum_{i=1}^2\|q_i^{\pi^t} - q_i^{\pi^{t-1}}\|^2$.
In the full feedback setting, this term can be canceled out by the path length of strategy profiles $\sum_{t=1}^T\sum_{i=1}^2(-1/\eta_t^2)\|\pi_i^t - \pi_i^{t-1}\|^2$ with the universal constant learning rate.
However, in the noisy feedback setting, the term $\sum_{t=1}^T\sum_{i=1}^2\|\hat{q}_i^{\pi^t} - \hat{q}_i^{\pi^{t-1}}\|^2$ appears instead of $\sum_{t=1}^T\sum_{i=1}^2\|q_i^{\pi^t} - q_i^{\pi^{t-1}}\|^2$, and it would grow linearly in $T$ even if $\pi^t$ is fixed.
Therefore, providing the last-iterate convergence results for optimistic algorithms with noisy feedback is challenging.
In contrast, the proof of last-iterate convergence with M2WU does not rely on the path length of the gradient vectors.
Specifically, it exploits the existence of continuous-time dynamics \eqref{eq:rmd} for M2WU and its Lyapunov function.
\end{remark}

\ifarxiv
\else
\begin{figure*}[t!]
    \centering
    \begin{minipage}[t]{0.24\textwidth}
        \centering
        \includegraphics[width=1.0\linewidth]{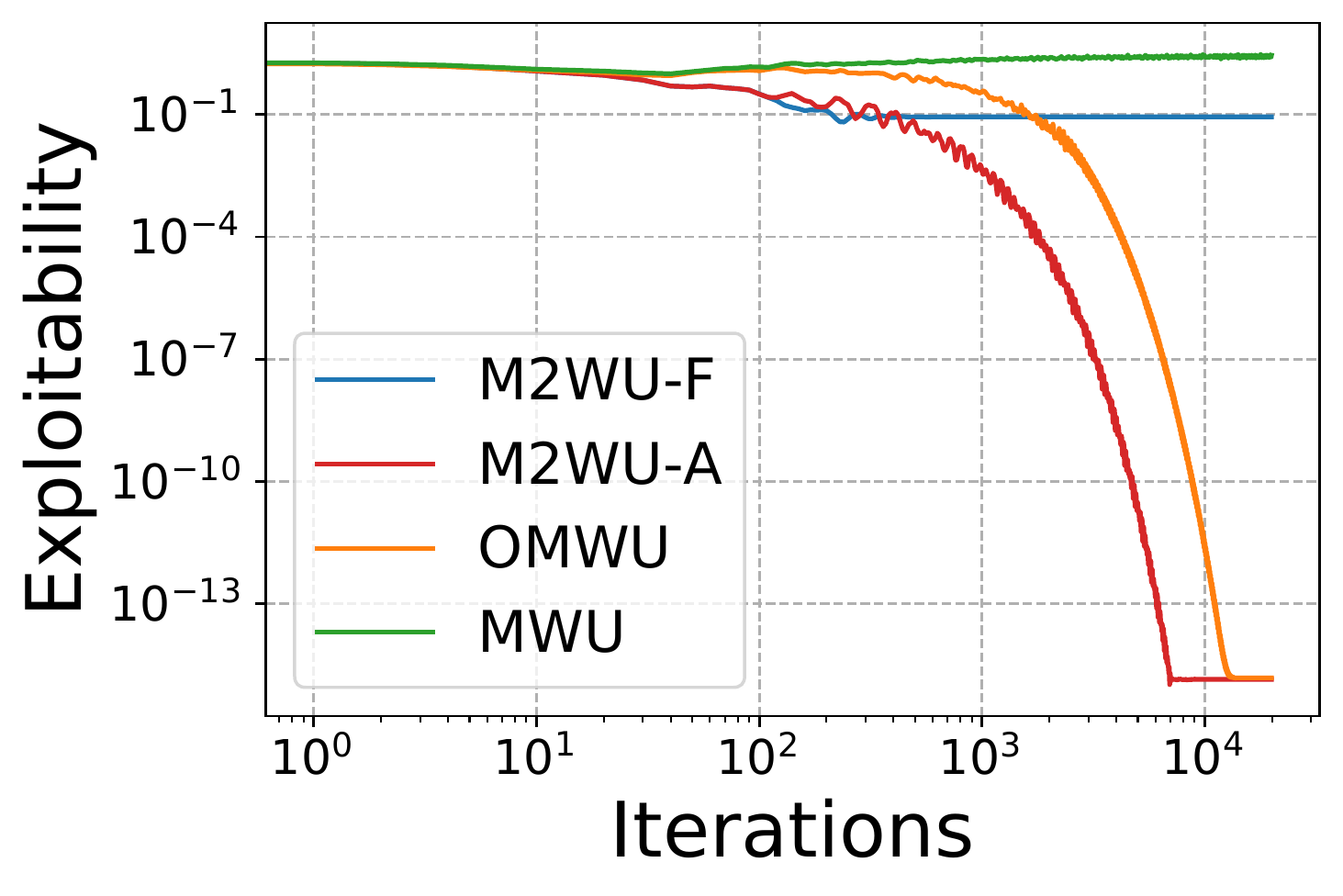}
        \subcaption{BRPS}
    \end{minipage}
    \begin{minipage}[t]{0.24\textwidth}
        \centering
        \includegraphics[width=1.0\linewidth]{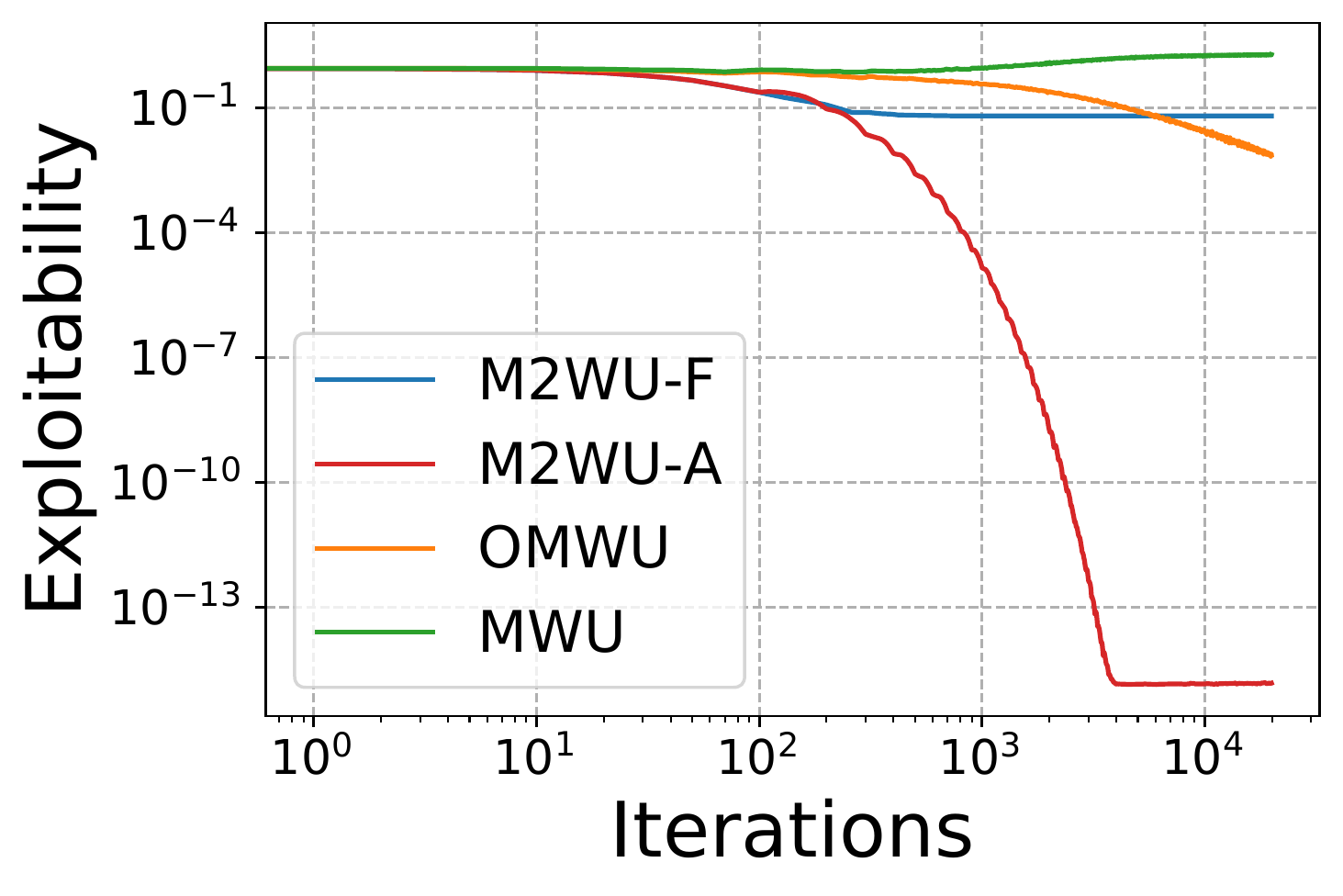}
        \subcaption{M-Ne}
    \end{minipage}
    \begin{minipage}[t]{0.24\textwidth}
        \centering
        \includegraphics[width=1.0\linewidth]{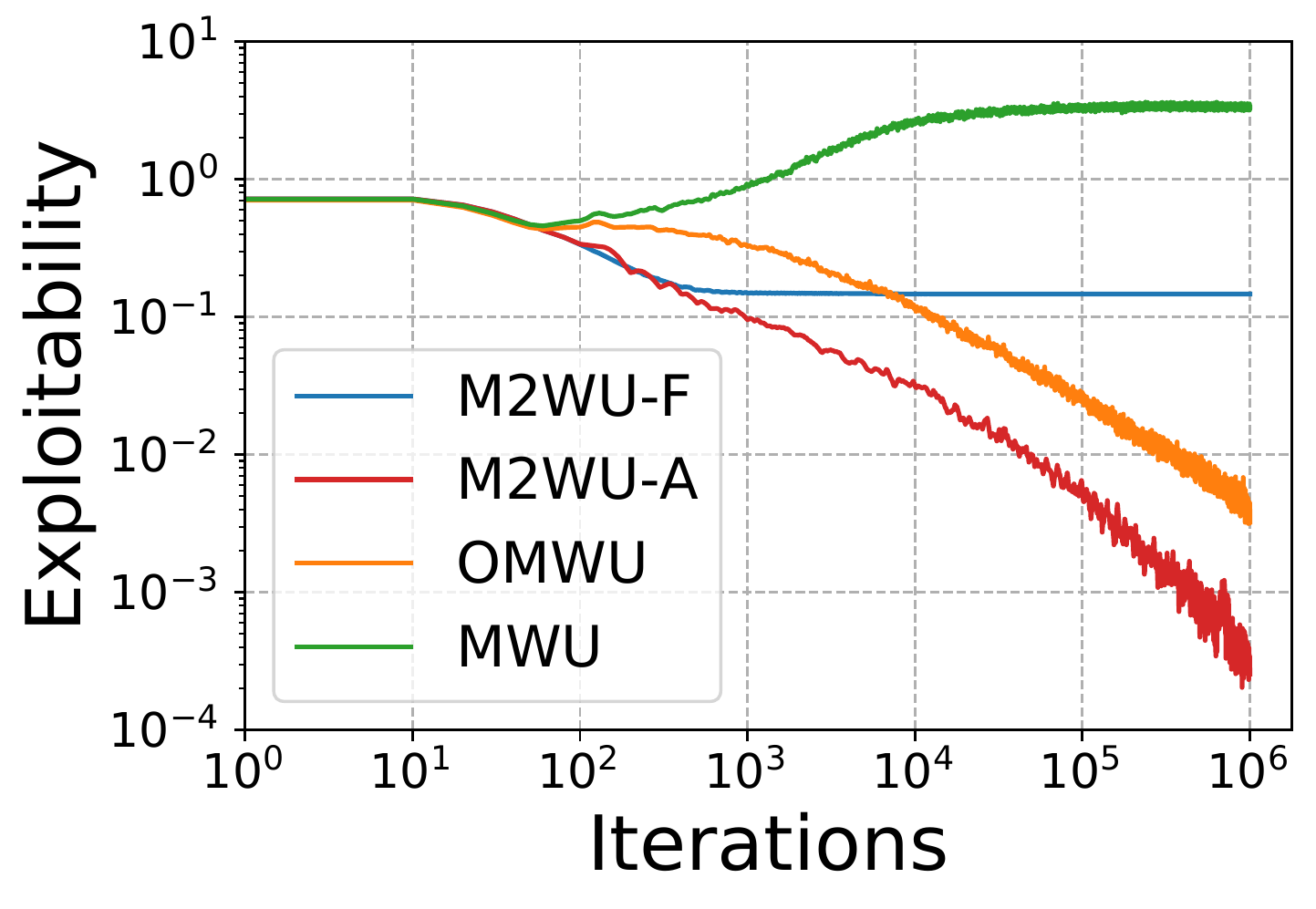}
        \subcaption{Random utility ($25\times 25$)}
    \end{minipage}
    \begin{minipage}[t]{0.24\textwidth}
        \centering
        \includegraphics[width=1.0\linewidth]{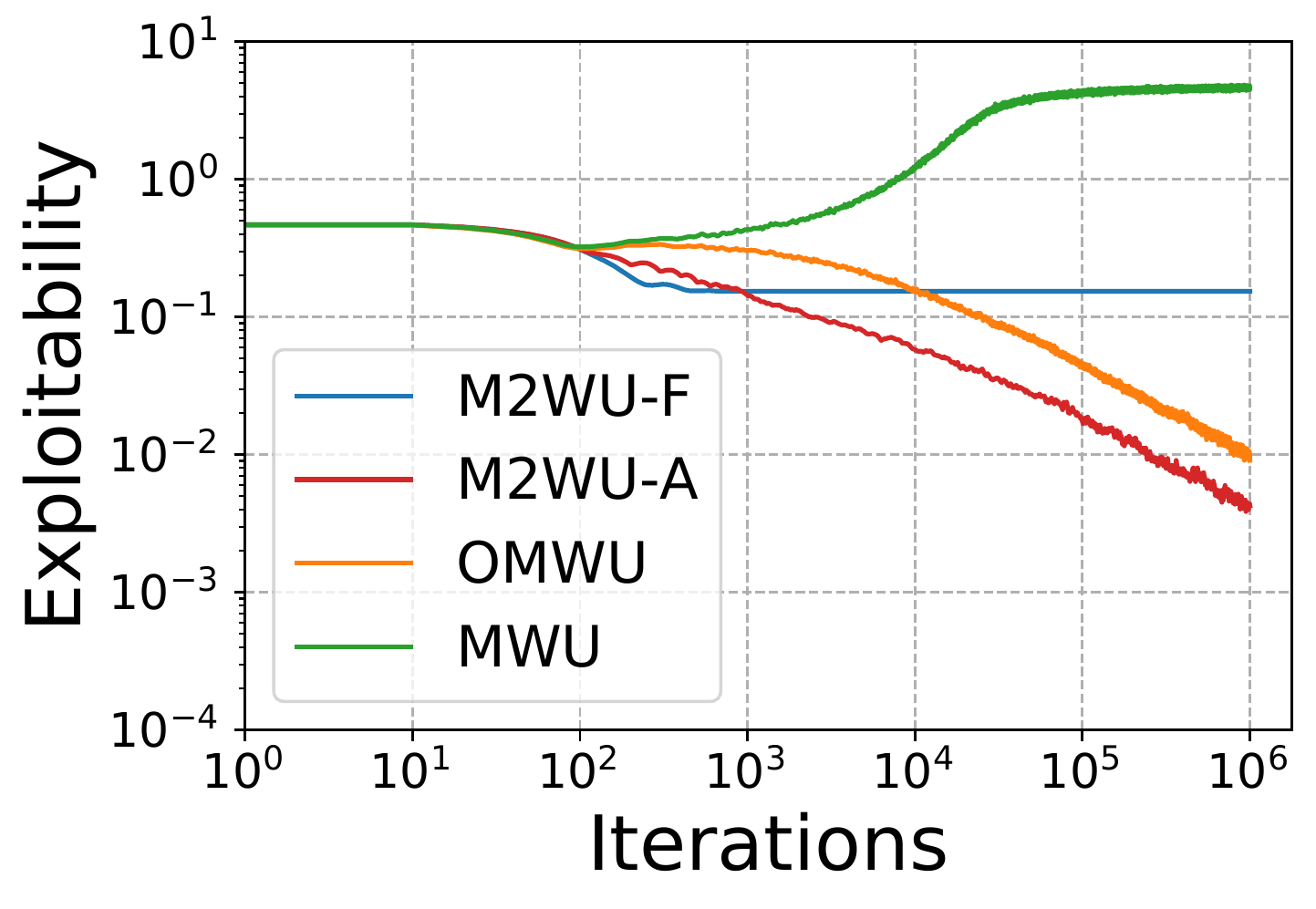}
        \subcaption{Random utility ($100\times 100$)}
    \end{minipage}
    \caption{
    Exploitability of $\pi^t$ for M2WU, MWU, and OMWU with full feedback.
    }
    \label{fig:exploitability_full}
\end{figure*}
\fi

\section{CONVERGENCE TO AN EXACT NASH EQUILIBRIUM}
Sections \ref{sec:m2wu} and \ref{sec:convergence} presented the M2WU with a fixed reference strategy profile $r$ and its convergence results.
As shown in Corollary \ref{cor:exploitability_convergence}, $\pi^t$ updated by M2WU converges to the stationary point $\pi^{\mu, r}$. 
Therefore, if the exploitability of the stationary point $\pi^{\mu, r}$ goes to zero, the exact Nash equilibrium of the original game can be obtained.
To this end, we control the exploitability of $\pi^{\mu,r}$ by adapting the reference strategy $r$.
That is, we copy the updated strategy profile $\pi^t$ to the reference profile $r$ every $N$ iterations.
This technique is similar to the direct convergence method by \citet{perolat2021poincare}.
The pseudo-code of M2WU with this technique corresponds to Algorithm~\ref{alg:m2wu} with finite~$N$.

Let us define $r^k$ as the $k$-th reference strategy profile.
From Theorem \ref{thm:kl_convergence}, $\pi^t$ converges to $\pi^{\mu, r}$ when $N$ is set to a sufficiently large value.
In this case, the following reference strategy $r^{k+1}$ is set to the stationary point $\pi^{\mu, r^k}$ of the (\ref{eq:rmd}) dynamics with reference strategy $r^k$.
In the remaining part of this section, we show that the sequence of stationary points $\{\pi^{\mu,r^k}\}_{k\geq 0}=\{r^k\}_{k\geq 1}$ converges to $\Pi^{\ast}$ of the original game.
\begin{theorem}
\label{thm:direct_convergence}
For any start point $r^0\in \prod_{i=1}^2\Delta^{\circ}(A_i)$, the sequence of stationary points $\{\pi^{\mu,r^k}\}_{k\geq 0}=\{r^k\}_{k\geq 1}$ converges to the set of equilibria $\Pi^{\ast}$ of the original game.
\end{theorem}
\ifarxiv
This result means that the exploitability of $\pi^{\mu,r^k}$ converges to $0$ because $\mathrm{explt}(\pi^{\mu,r^k})\leq \sum\limits_{i\in \{1,2\}}\max\limits_{\tilde{\pi}_i\in \Delta(A_i)}( v_i^{\tilde{\pi}_i, \pi_{-i}} - v_i^{\tilde{\pi}_i, \pi_{-i}^{\star}})\leq \mathcal{O}\left(\sqrt{\mathrm{KL}(\pi^{\star},\pi^{\mu,r^k}})\right)$, where $\pi^{\star}=\argmin_{\pi\in \Pi^{\ast}}\mathrm{KL}(\pi, \pi^{\mu,r^k})$.
\else
This result means that the exploitability of $\pi^{\mu,r^k}$ converges to $0$ because $\mathrm{explt}(\pi^{\mu,r^k})\leq \!\mathcal{O}\!\left(\!\sqrt{\mathrm{KL}(\pi^{\star},\pi^{\mu,r^k}})\!\right)$, where $\pi^{\star}=\argmin_{\pi\in \Pi^{\ast}}\mathrm{KL}(\pi, \pi^{\mu,r^k})$.
\fi
We explain the proof sketch here; the complete proof of Theorem~\ref{thm:direct_convergence} is in Appendix \ref{sec:appendix_proof_direct_convergence}.
\paragraph{Proof Sketch of Theorem~\ref{thm:direct_convergence}}
Let $F: \prod_{i=1}^2\Delta^{\circ}(A_i)\to \prod_{i=1}^2\Delta^{\circ}(A_i)$ be a function which maps the reference strategies $r$ to the associated stationary point $\pi^{\mu,r}$.
We also define $r^{k+1} = F(r^k)$, and $r^0\in \prod_{i=1}^2\Delta^{\circ}(A_i)$ as the starting reference strategy profile.
 Note that the function $F$ is well-defined because the stationary point $\pi^{\mu,r}$ of \eqref{eq:rmd} is unique for each $r \in \prod_{i=1}^2\Delta^{\circ}(A_i)$ \citep{abe2022mutationdriven}.
First, we prove that the distance between $\Pi^{\ast}$ and $r^k$ decreases monotonically as $k$ increases:
\begin{lemma}
\label{lem:min_kl_diff}
For any $k\geq 0$, if $r^k\in \prod_{i=1}^2\Delta^{\circ}(A_i) \setminus \Pi^{\ast}$, then:
\begin{align*}
    \min_{\pi^{\ast}\in \Pi^{\ast}}\mathrm{KL}(\pi^{\ast}, r^{k+1}) < \min_{\pi^{\ast}\in \Pi^{\ast}}\mathrm{KL}(\pi^{\ast}, r^k).
\end{align*}
Otherwise, if $r^k\in \Pi^{\ast}$, then $r^{k+1}=r^k\in \Pi^{\ast}$.
\end{lemma}
We also show that $F(\cdot)$ is a continuous function:
\begin{lemma}
\label{lem:rmd_continuity}
Let $F(r): \prod_{i=1}^2\Delta^{\circ}(A_i)\to \prod_{i=1}^2\Delta^{\circ}(A_i)$ be a function which maps the reference strategies $r$ to the associated stationary point $\pi^{\mu,r}$ of (\ref{eq:rmd}).
Then, $F(\cdot)$ is a continuous function on $\prod_{i=1}^2\Delta^{\circ}(A_i)$.
\end{lemma}
For these lemmas, we can use Lyapunov arguments to obtain a convergence result for $\{\pi^{\mu,r^k}\}_{k\geq 0}$. \qed

\begin{remark}
\normalfont
One might think that a simple annealing approach that gradually decrease mutation parameter $\mu$, leads dynamics to reach an exact Nash equilibrium. While this is indeed the case, Theorem~\ref{thm:kl_convergence} indicates that the learning rate $\eta$ must also be significantly reduced.  
The annealing approach would make the convergence speed very slow. 
%
\end{remark}

\section{EXPERIMENTS}
\label{sec:experiments}
We here abbreviate M2WU with a fixed reference strategy profile (Algorithm~\ref{alg:m2wu} with $N=\infty$) as M2WU-F, while M2WU \ with adaptive reference strategy profiles as M2WU-A. 
This section conducts a series of experiments to demonstrate how the four algorithms, i.e., MWU, OMWU, M2WU-F, and M2WU-A, including ours, perform.

We focus on four games: Biased Rock-Paper-Scissors (BRPS), Multiple Nash Equilibria (M-Ne), and two random utility games with 25 and 100 actions.
Note that we borrow the M-Ne game from \citet{wei2020linear}.
Tables~\ref{tab:biased-rps} and \ref{tab:m-eq} provide the payoff matrices for BRPS and M-Ne, respectively. 
\ifarxiv
\begin{table}[h!]
    \centering
    \begin{minipage}[t]{0.49\textwidth}
    \centering
    \caption{Biased RPS game matrix}
    \label{tab:biased-rps}
    \begin{tabular}{cccc}
    \hline
      & R  & P  & S  \\ \hline
    R & $0$  & $-1$  & $3$ \\
    P & $1$ & $0$  & $-1$ \\
    S & $-3$  & $1$ & $0$  \\ \hline
    \end{tabular}
    \end{minipage}
    \begin{minipage}[t]{0.49\textwidth}
    \centering
    \caption{M-Ne game matrix}
    \label{tab:m-eq}
    \begin{tabular}{cccccc}
    \hline
      & $y_1$  & $y_2$ & $y_3$ & $y_4$ & $y_5$  \\ \hline
    $x_1$ & $0$ & $1$ & $-1$ & $0$ & $0$ \\ 
    $x_2$ & $-1$ & $0$ & $1$ & $0$ & $0$ \\ 
    $x_3$ & $1$  & $-1$ & $0$ & $0$ & $0$  \\ 
    $x_4$ & $1$  & $-1$ & $0$ & $-2$ & $1$  \\ 
    $x_5$ & $1$  & $-1$ & $0$ & $1$ & $-2$  \\ \hline
    \end{tabular}
    \end{minipage}
\end{table}
\else
\begin{table}[h!]
    \centering
    \caption{Biased RPS game matrix}
    \label{tab:biased-rps}
    \begin{tabular}{cccc}
    \hline
      & R  & P  & S  \\ \hline
    R & $0$  & $-1$  & $3$ \\
    P & $1$ & $0$  & $-1$ \\
    S & $-3$  & $1$ & $0$  \\ \hline
    \end{tabular}
    \caption{M-Ne game matrix}
    \label{tab:m-eq}
    \begin{tabular}{cccccc}
    \hline
      & $y_1$  & $y_2$ & $y_3$ & $y_4$ & $y_5$  \\ \hline
    $x_1$ & $0$ & $1$ & $-1$ & $0$ & $0$ \\ 
    $x_2$ & $-1$ & $0$ & $1$ & $0$ & $0$ \\ 
    $x_3$ & $1$  & $-1$ & $0$ & $0$ & $0$  \\ 
    $x_4$ & $1$  & $-1$ & $0$ & $-2$ & $1$  \\ 
    $x_5$ & $1$  & $-1$ & $0$ & $1$ & $-2$  \\ \hline
    \end{tabular}
\end{table}
\fi

BRPS has the unique Nash equilibrium $\Pi^{\ast}_i = \{(1/5, 3/5, 1/5)\}$. M-Ne has the following set of Nash equilibria:
\begin{align*}
&\Pi^{\ast}_1 = \left\{\left(1/3, 1/3, 1/3, 0, 0\right) \right\},  \\
&\Pi^{\ast}_2 \!=\! \left\{ y\in \Delta^5 ~|~y_1=y_2=y_3;~ y_5/2\leq y_4\leq 2y_5\right\}.
\end{align*}

Let us proceed to random utility games to consider how our algorithms perform in relatively large games whose numbers of actions are 25 or 100. We draw each utility (or payoff) component from the standard Gaussian distribution $\mathcal{N}(0, 1)$ in an i.i.d. manner. 

For each game, we average exploitability over $100$ instances with different random seeds. We also set the initial strategy profile $\pi^0$ uniformly at random in $\prod_{i=1}^2\Delta^{\circ}(A_i)$ in each instance for BRPS and M-Ne with full feedback.
In other instances, the initial strategy is set to $(1/|A_i|)_{a\in A_i}$ for $i\in \{1,2\}$.

\ifarxiv
\begin{figure*}[t!]
    \centering
    \begin{minipage}[t]{0.24\textwidth}
        \centering
        \includegraphics[width=1.1\linewidth]{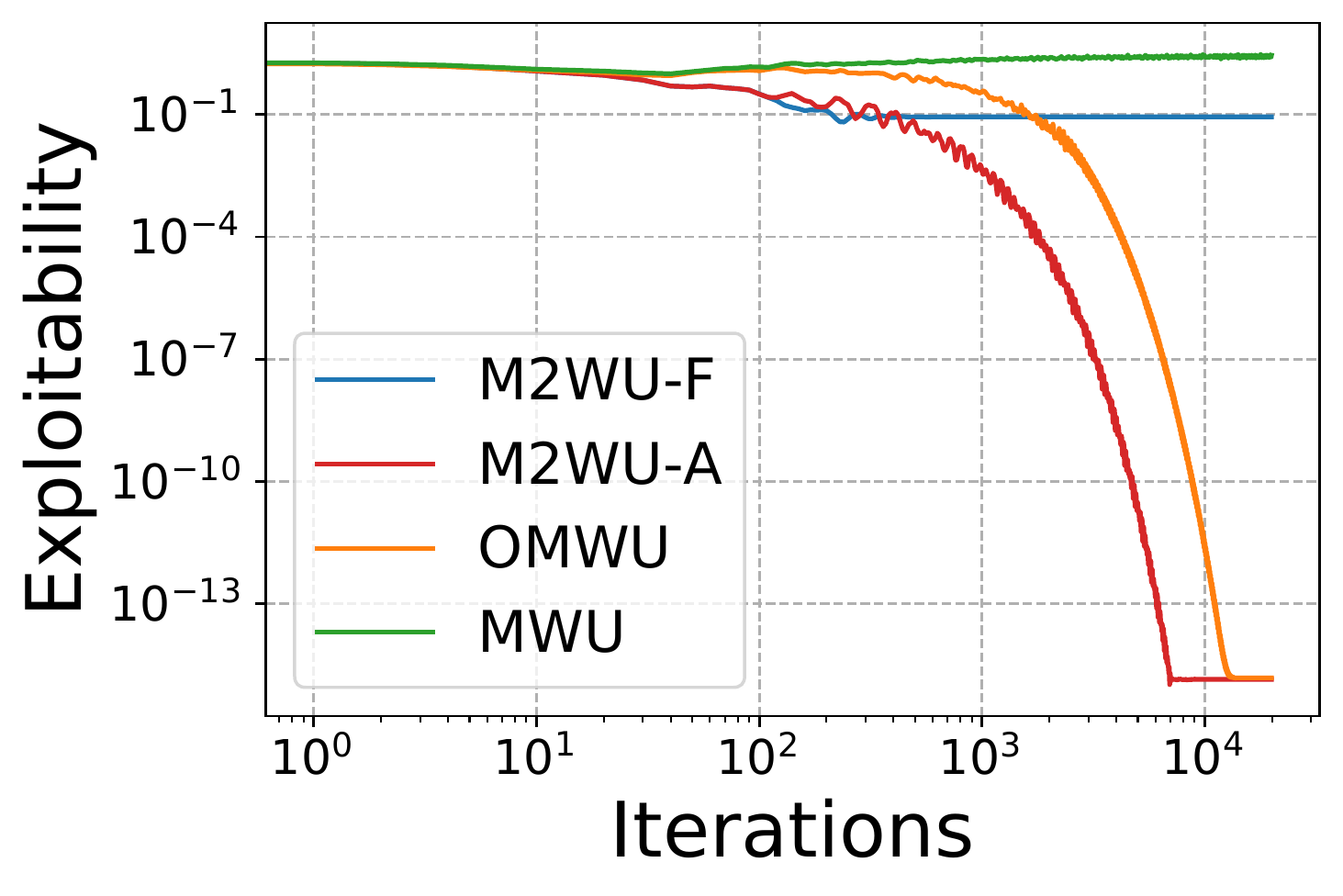}
        \subcaption{BRPS}
    \end{minipage}
    \begin{minipage}[t]{0.24\textwidth}
        \centering
        \includegraphics[width=1.1\linewidth]{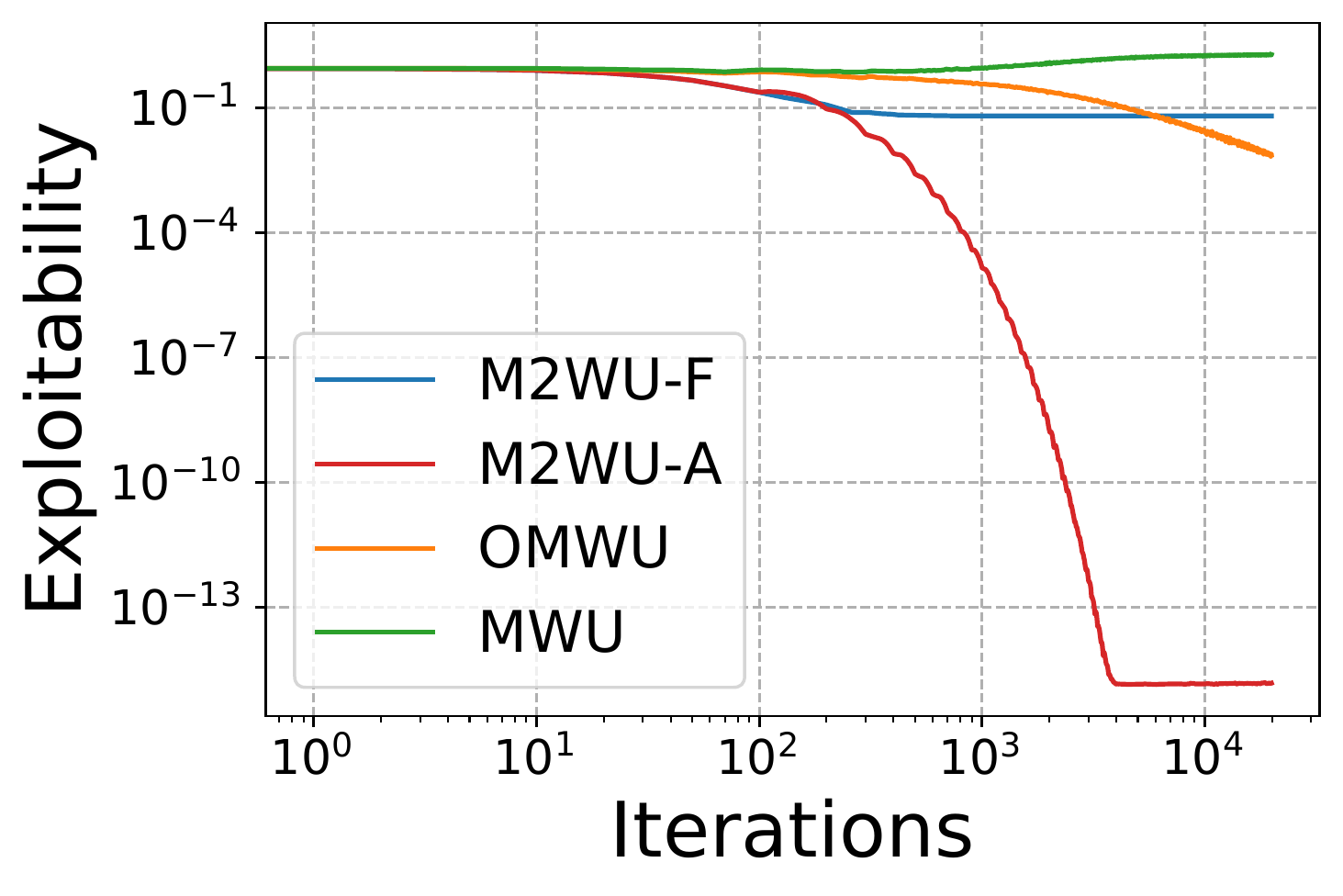}
        \subcaption{M-Ne}
    \end{minipage}
    \begin{minipage}[t]{0.24\textwidth}
        \centering
        \includegraphics[width=1.1\linewidth]{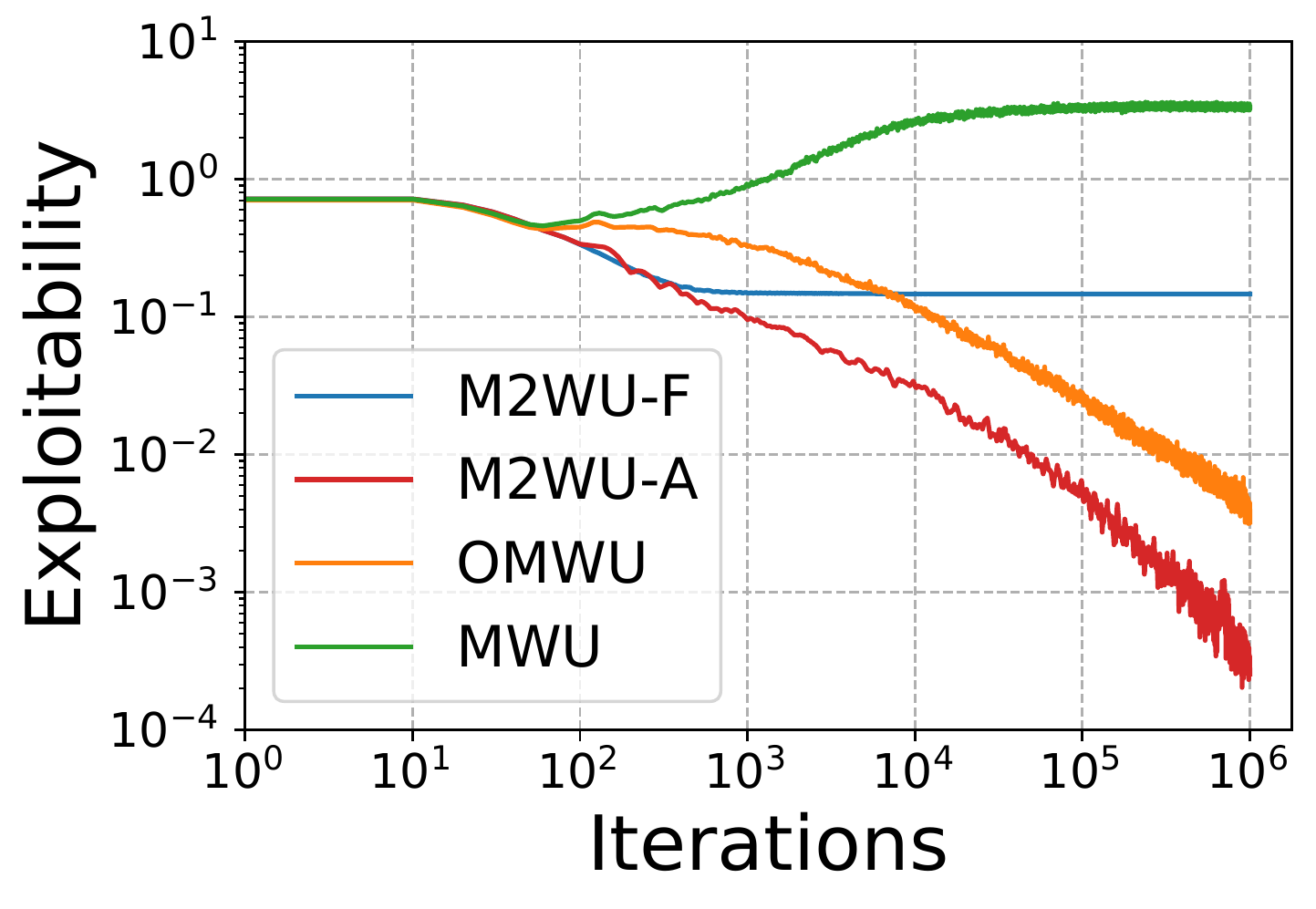}
        \subcaption{Random utility ($25\times 25$)}
    \end{minipage}
    \begin{minipage}[t]{0.24\textwidth}
        \centering
        \includegraphics[width=1.1\linewidth]{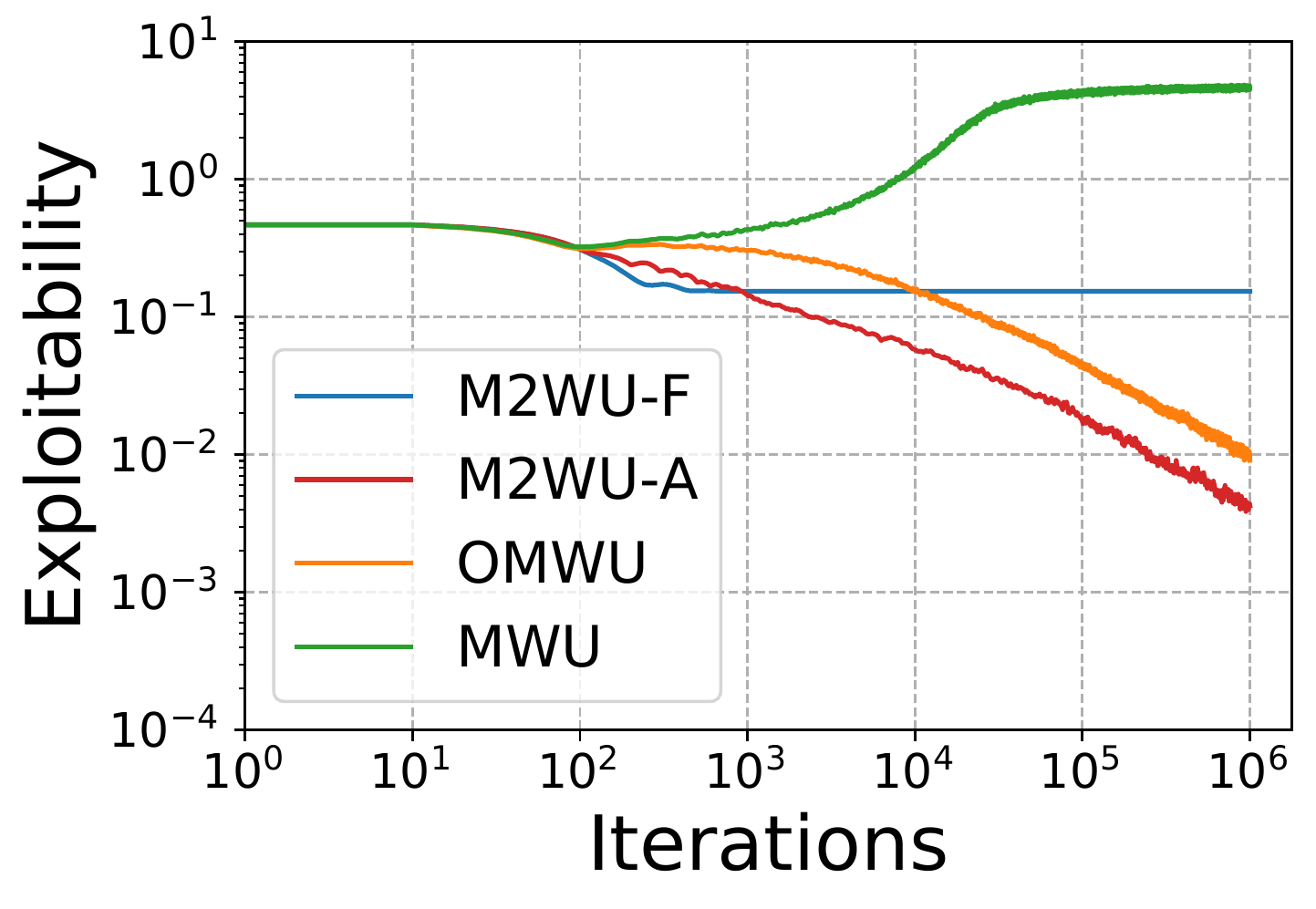}
        \subcaption{Random utility ($100\times 100$)}
    \end{minipage}
    \caption{
    Exploitability of $\pi^t$ for M2WU, MWU, and OMWU with full feedback.
    }
    \label{fig:exploitability_full}
\end{figure*}
\begin{figure}[t!]
    \centering
    \includegraphics[width=0.375\linewidth]{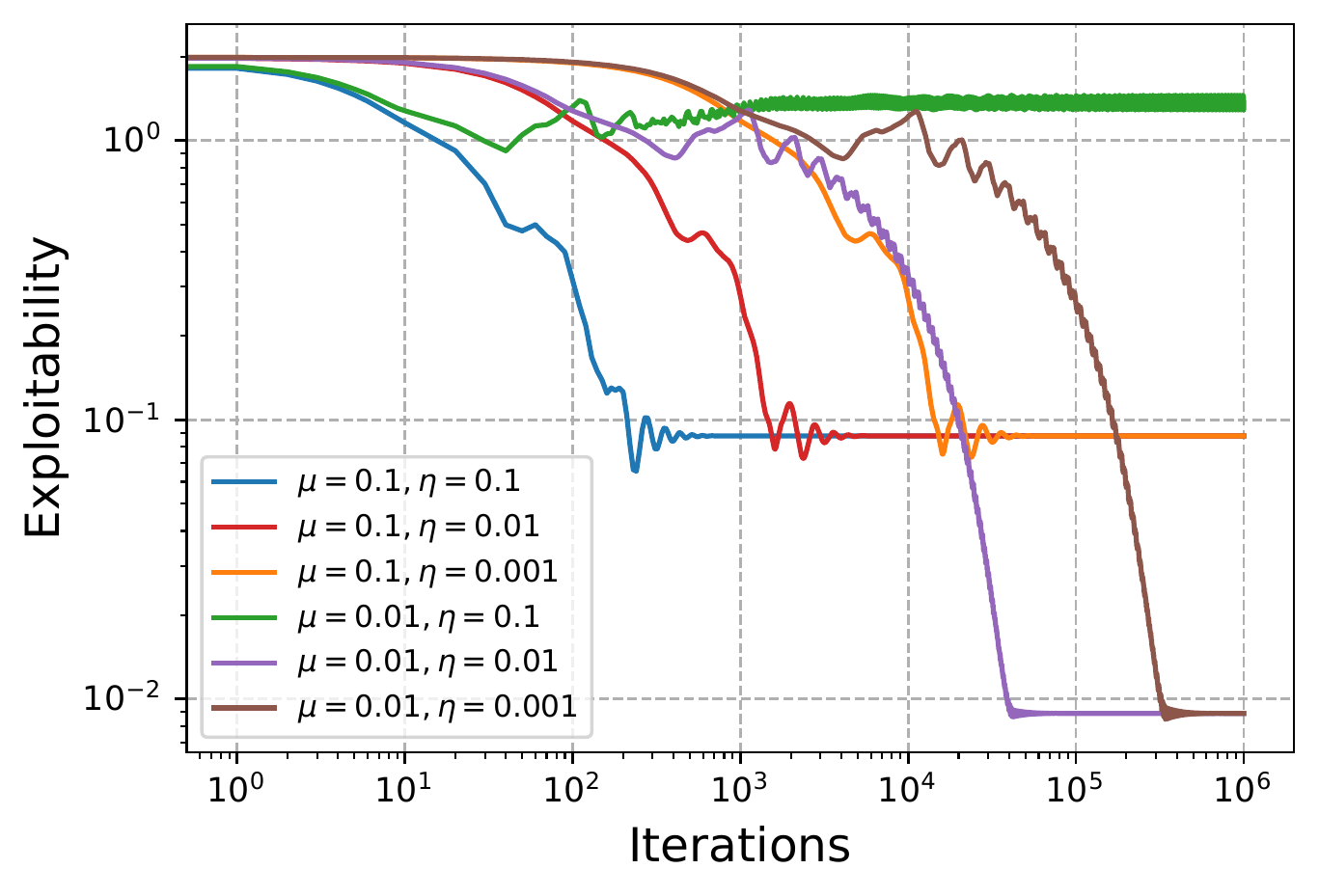}
    \caption{Exploitability of $\pi^t$ for M2WU-F with varying $\mu\in\{0.1, 0.01\}$ and $\eta\in \{0.1, 0.01, 0.001\}$ in BRPS with full feedback.}
    \label{fig:compare_mu_eta}
\end{figure}
\fi

\subsection{Full Feedback}
This section examines the algorithms with full feedback.
Unless noticed, we use the learning rate $\eta$ of $0.1$ and the mutation rate $\mu$ of $0.1$.
For M2WU-F, we also assume the reference strategy profile $r_i^0$ is fixed at $(1/|A_i|)_{a\in A_i}$.
For M2WU-A, we update it every $N=100$ iterations.

\ifarxiv
\else
\begin{figure}
    \centering
    \includegraphics[width=1.0\linewidth]{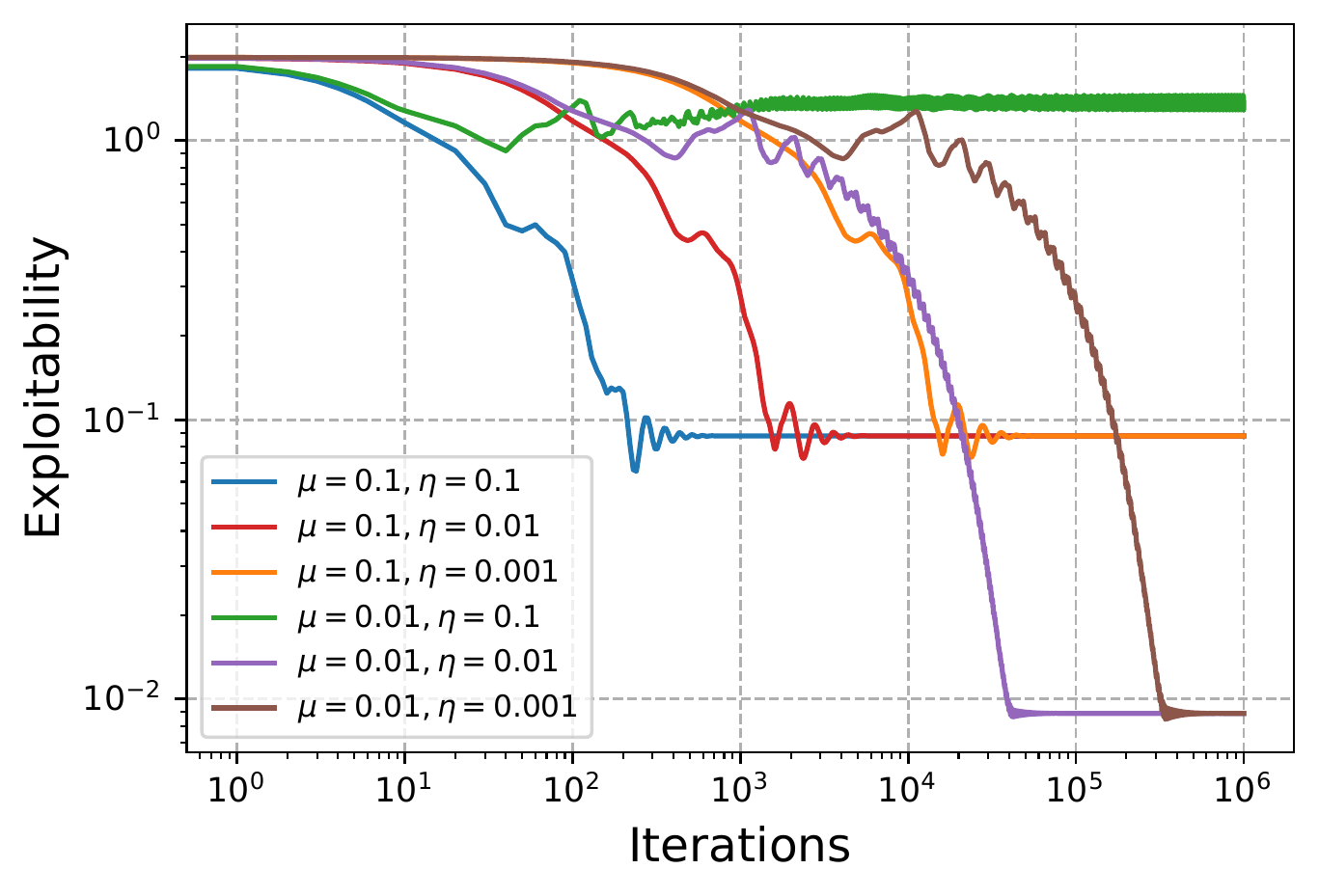}
    \caption{Exploitability of $\pi^t$ for M2WU-F with varying $\mu\in\{0.1, 0.01\}$ and $\eta\in \{0.1, 0.01, 0.001\}$ in BRPS with full feedback.}
    \label{fig:compare_mu_eta}
\end{figure}
\begin{figure*}[t!]
    \centering
    \begin{minipage}[t]{0.24\textwidth}
        \centering
        \includegraphics[width=1.0\linewidth]{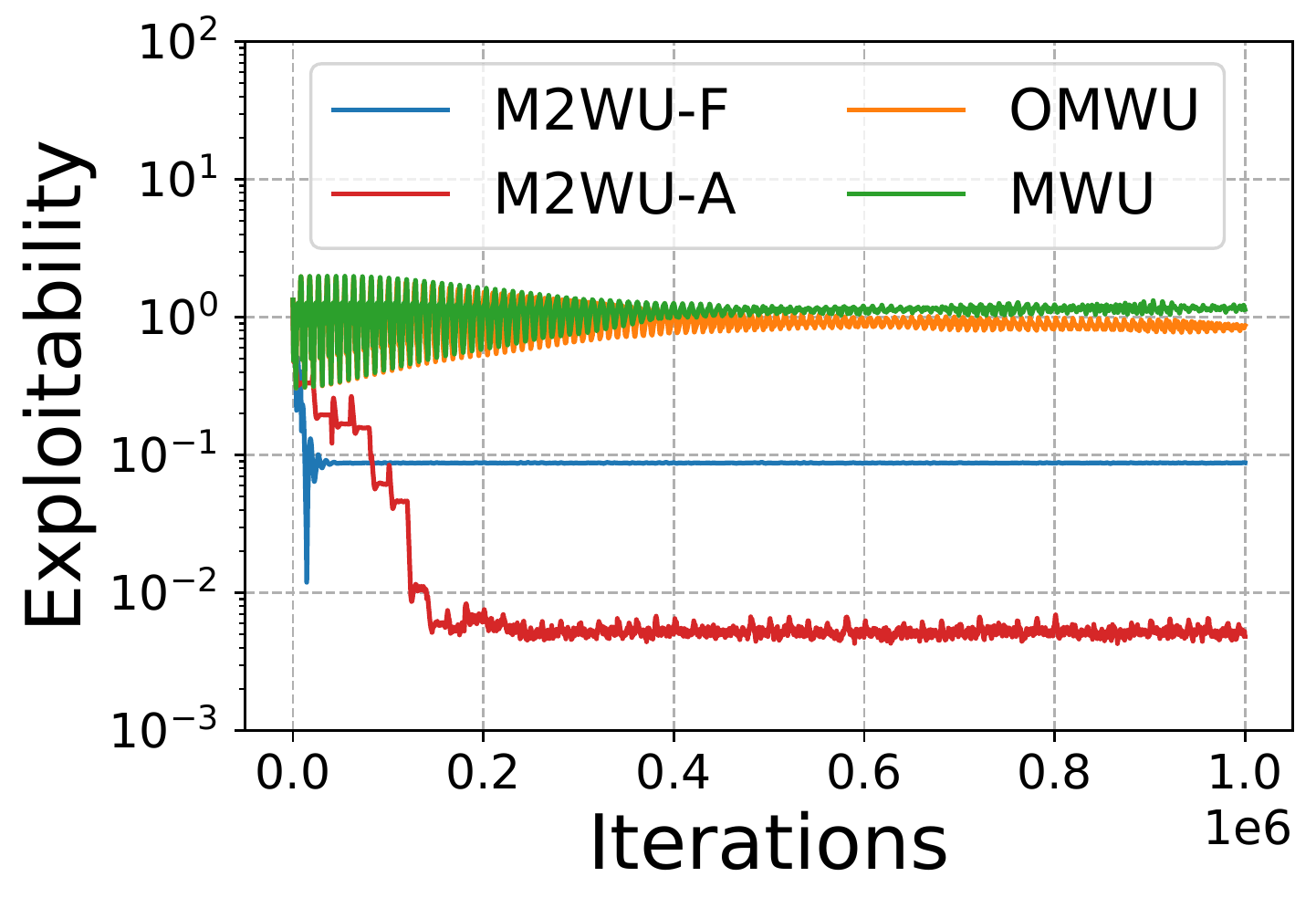}
        \subcaption{BRPS}
    \end{minipage}
    \begin{minipage}[t]{0.24\textwidth}
        \centering
        \includegraphics[width=1.0\linewidth]{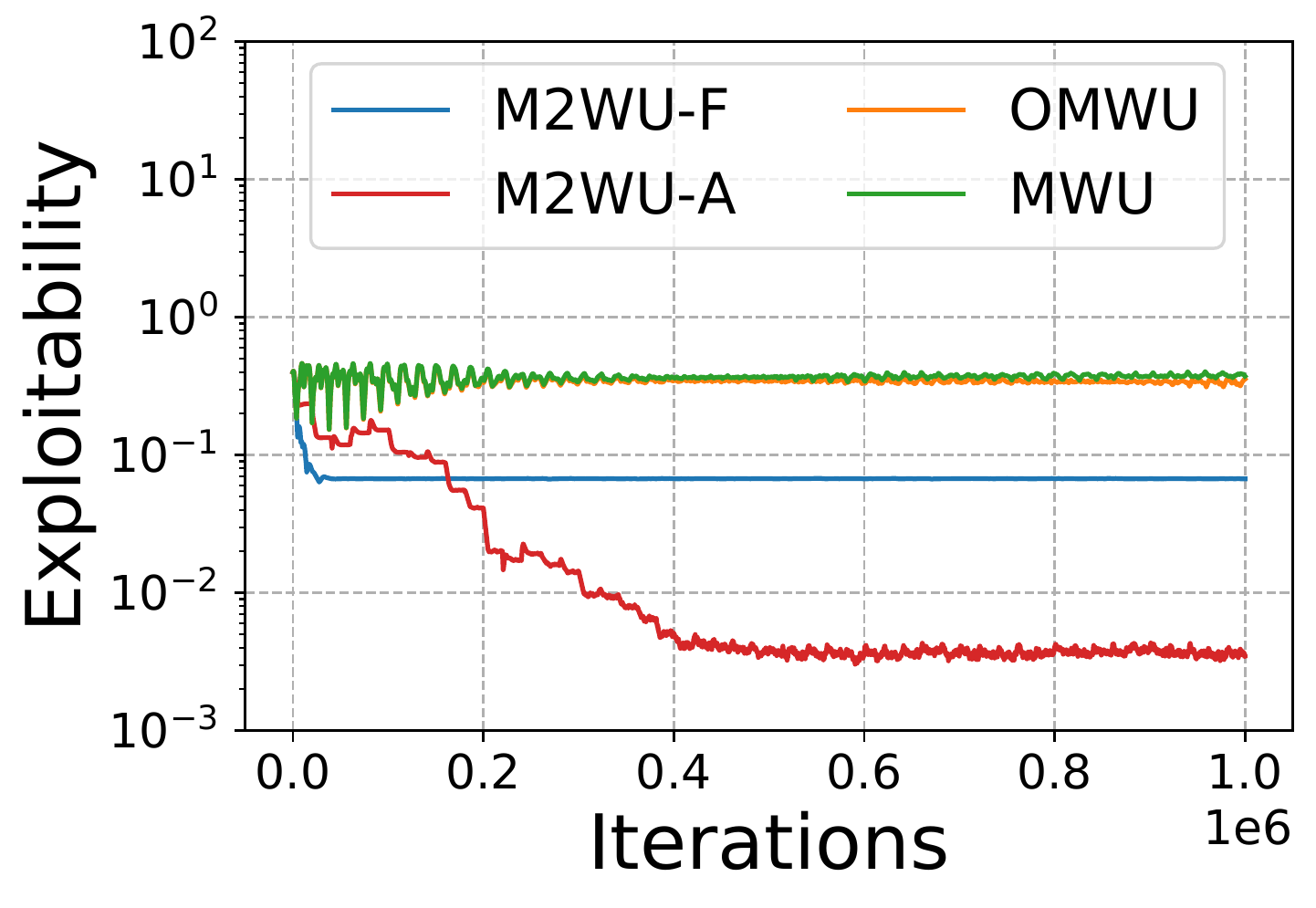}
        \subcaption{M-Ne}
    \end{minipage}
    \begin{minipage}[t]{0.24\textwidth}
        \centering
        \includegraphics[width=1.0\linewidth]{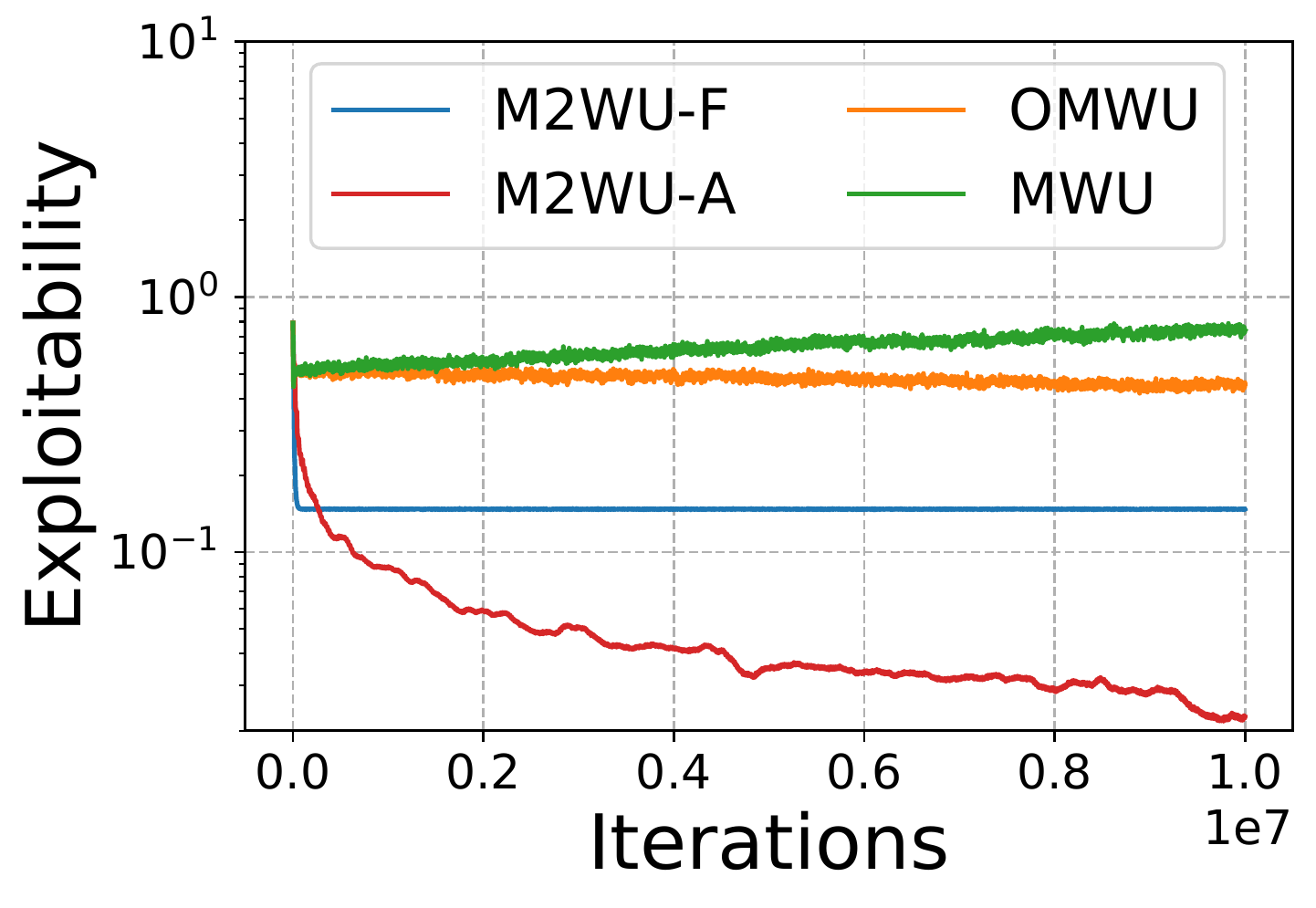}
        \subcaption{Random utility ($25\times 25$)}
    \end{minipage}
    \begin{minipage}[t]{0.24\textwidth}
        \centering
        \includegraphics[width=1.0\linewidth]{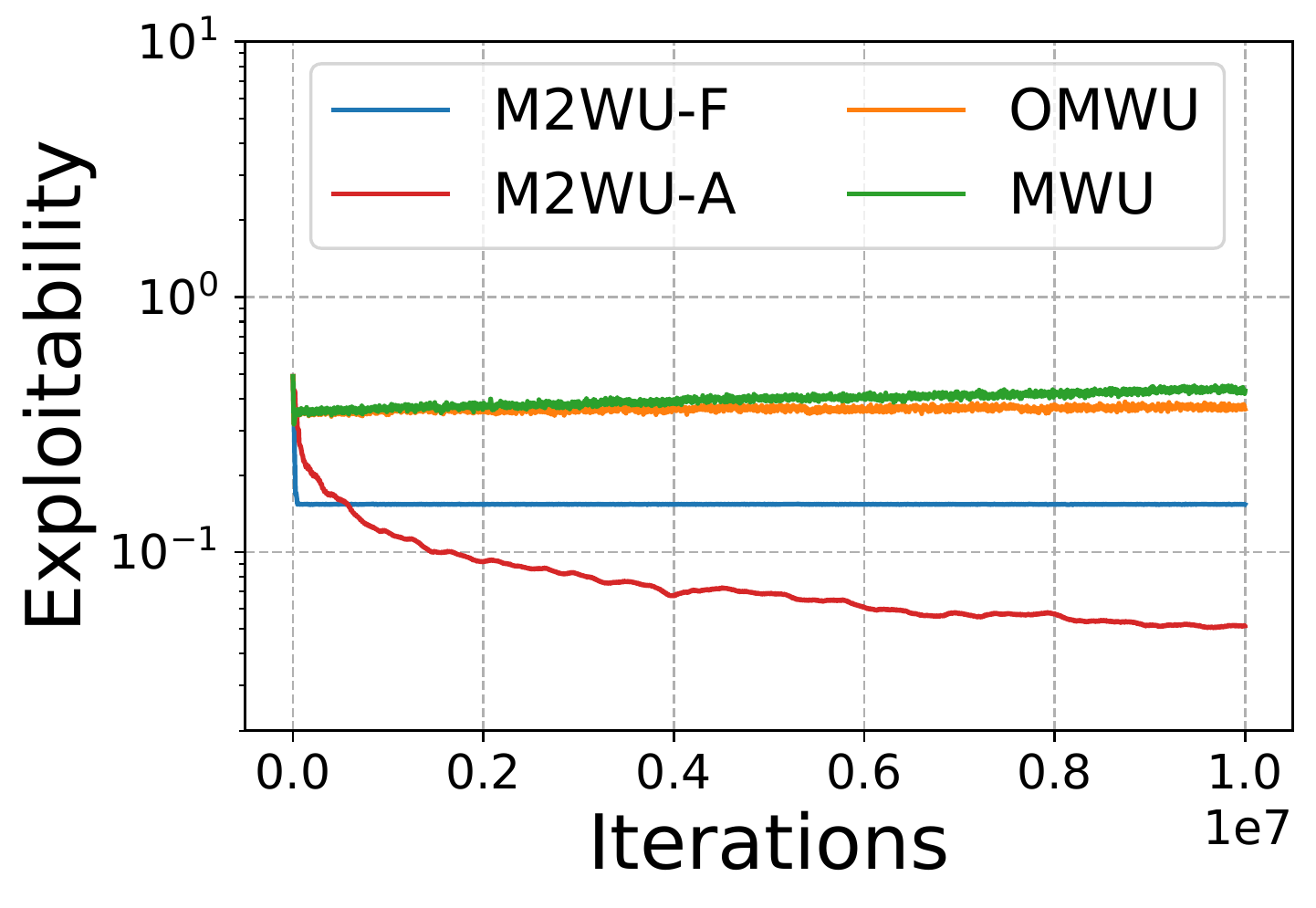}
        \subcaption{Random utility ($100\times 100$)}
    \end{minipage}
    \caption{
    Exploitability of $\pi^t$ for M2WU, MWU, and OMWU with noisy feedback.
    }
    \label{fig:exploitability_noisy}
\end{figure*}
\begin{figure*}[t!]
    \centering
    \begin{minipage}[t]{0.24\textwidth}
        \centering
        \includegraphics[width=1.0\linewidth]{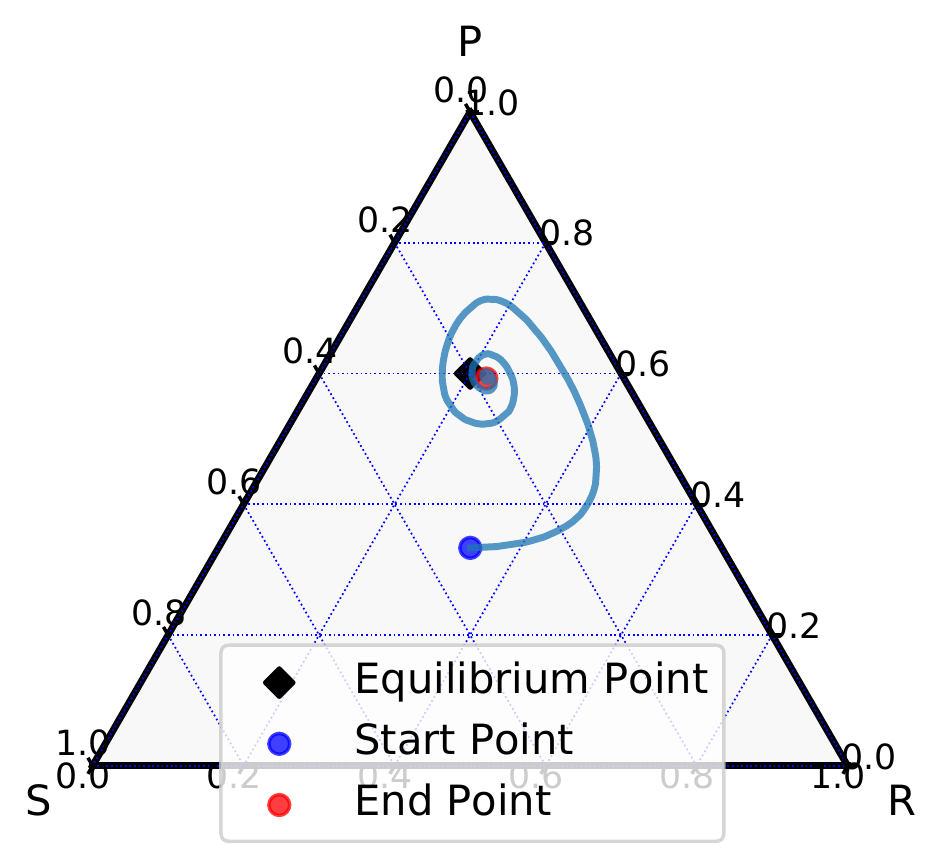}
        \subcaption{M2WU-F}
    \end{minipage}
    \begin{minipage}[t]{0.24\textwidth}
        \centering
        \includegraphics[width=1.0\linewidth]{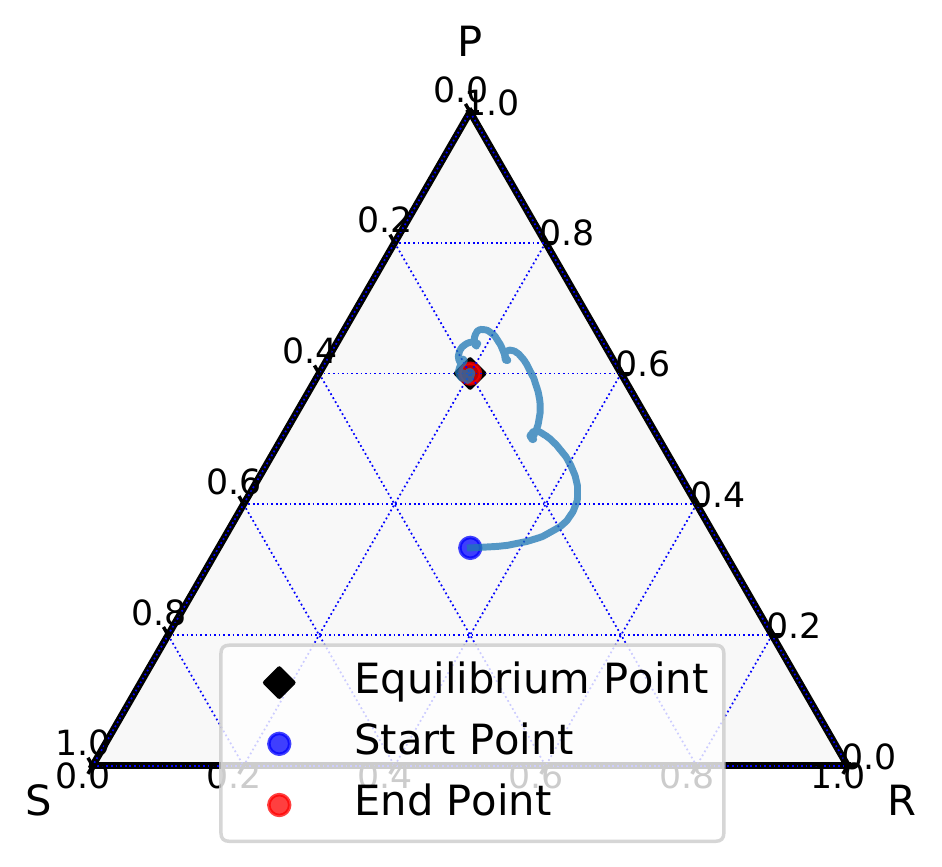}
        \subcaption{M2WU-A}
    \end{minipage}
    \begin{minipage}[t]{0.24\textwidth}
        \centering
        \includegraphics[width=1.0\linewidth]{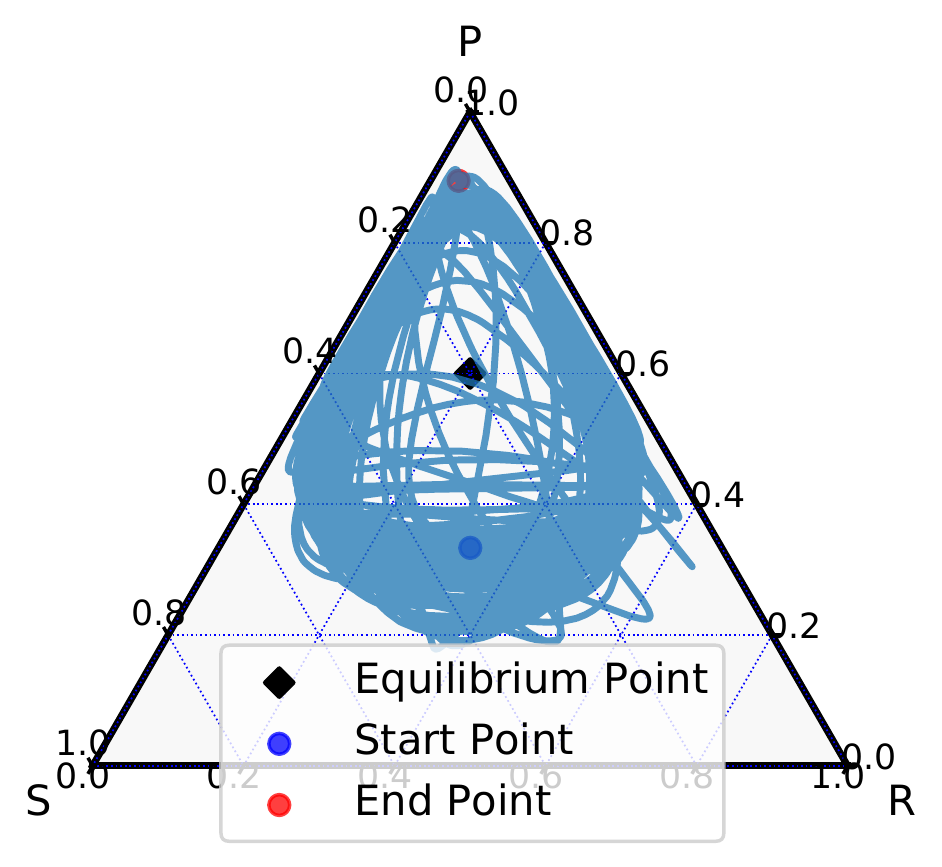}
        \subcaption{MWU}
    \end{minipage}
    \begin{minipage}[t]{0.24\textwidth}
        \centering
        \includegraphics[width=1.0\linewidth]{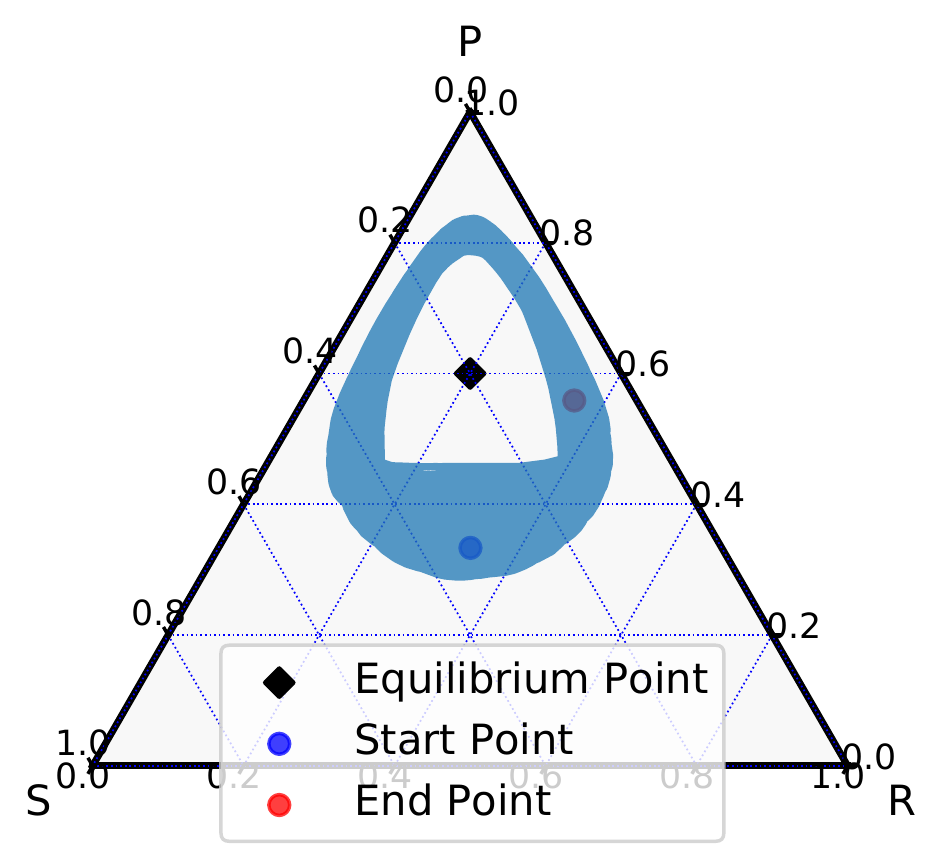}
        \subcaption{OMWU}\label{fig:cycle with OMWU}
    \end{minipage}
    \caption{
    Trajectories of $\pi^t$ for M2WU, MWU and OMWU in BRPS with noisy feedback.
    We set the initial strategy to $\pi_i^0=(1/|A_i|)_{a\in A_i}$ for $i\in \{1, 2\}$.
    The black point represents the equilibrium strategy.
    The blue/red points represent the initial/final points, respectively.
    }
    \label{fig:trajectory_brps_noisy}
\end{figure*}
\fi

Figure~\ref{fig:exploitability_full} averages the exploitability of the last-iterate strategy $\pi^t$ for the four algorithms in the four games. We observe that, in any game, MWU and M2WU-F do not converge to an equilibrium and the exploitability reaches around $0.1$ at best. Both OMWU and M2WU-A exhibit clear convergence. M2WU-A converges faster than OMWU, which may converge to the same exploitability in the long run. 

Note that the quick convergence of M2WU-F to a constant value lower than $2\mu=0.2$ supports Corollary~\ref{cor:exploitability_convergence} in the upper bound on the exploitability of $\pi^t$. The best performance of M2WU-A supports Theorem~\ref{thm:direct_convergence} and Lemma~\ref{lem:min_kl_diff}, which imply that the sequence of stationary points of (\ref{eq:rmd}) converges to the Nash equilibrium. 

We have so far fixed the learning rate at $0.1$ and the mutation rate at $0.1$.
Figure~\ref{fig:compare_mu_eta} exhibits the exploitability of M2WU-F in BRPS with varying $\eta\in\{0.1, 0.01, 0.001\}$ and $\mu\in \{0.1, 0.01\}$.
The tendencies in the other games are qualitatively similar.
We observe that the lower mutation rate makes exploitability low, while the larger mutation rate makes convergence fast.
When the mutation rate is low (e.g., $0.01$), the learning rate must be lower to enjoy the last-iterate convergence.
Otherwise, agents do not properly learn equilibrium (See the green line with $\eta=0.1$ and $\mu=0.01$ in Figure~\ref{fig:compare_mu_eta}).
This is why we chose $\eta=0.1$ and $\mu=0.1$ as a baseline.
With them, M2WU-F quickly converges, though the exploitability can be tuned.
This result highlights the relationship between $\mu$ and the upper bound on $\eta$ in Theorem \ref{thm:kl_convergence}.

\subsection{Noisy Feedback}
\ifarxiv
\begin{figure*}[t!]
    \centering
    \begin{minipage}[t]{0.24\textwidth}
        \centering
        \includegraphics[width=1.1\linewidth]{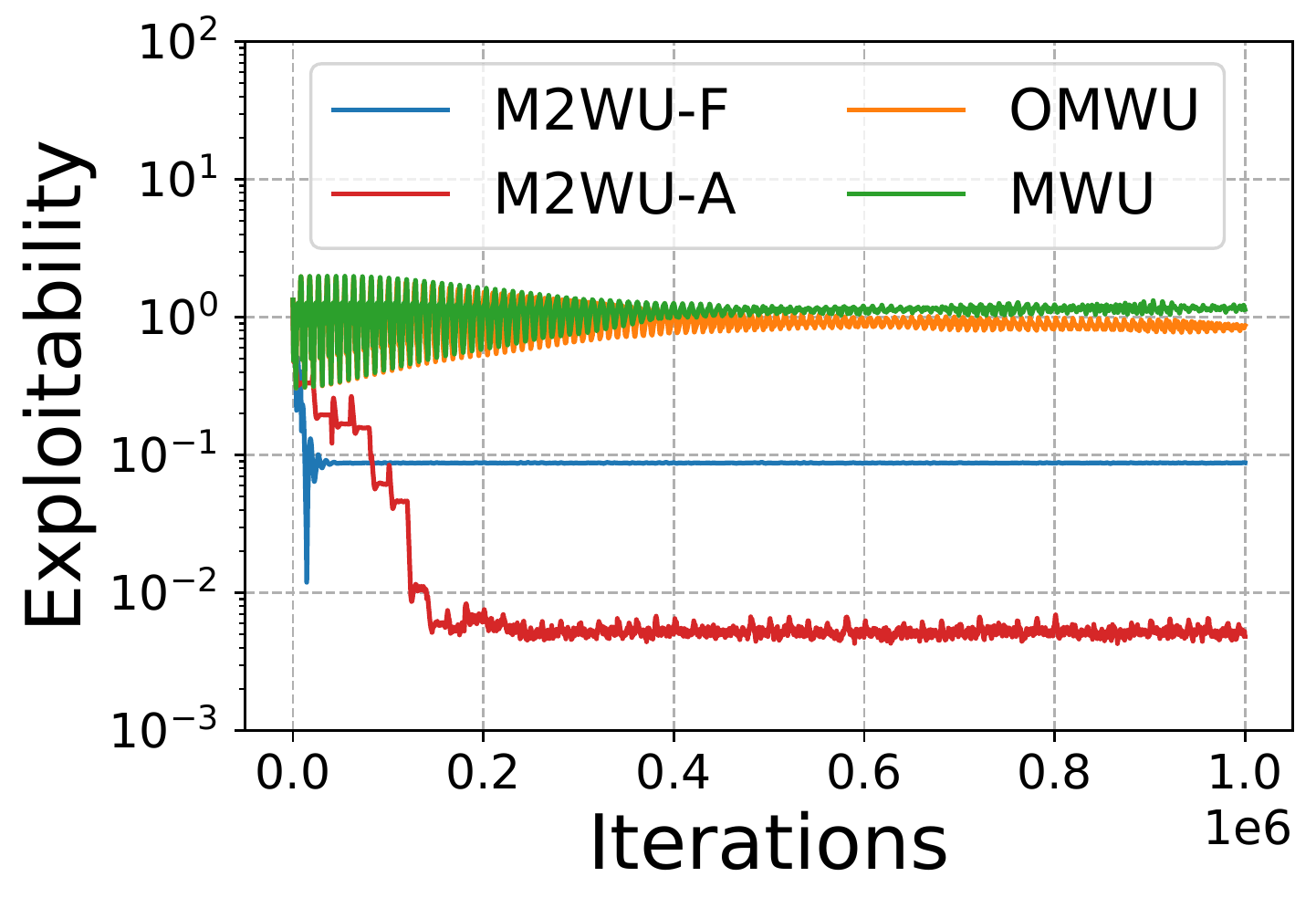}
        \subcaption{BRPS}
    \end{minipage}
    \begin{minipage}[t]{0.24\textwidth}
        \centering
        \includegraphics[width=1.1\linewidth]{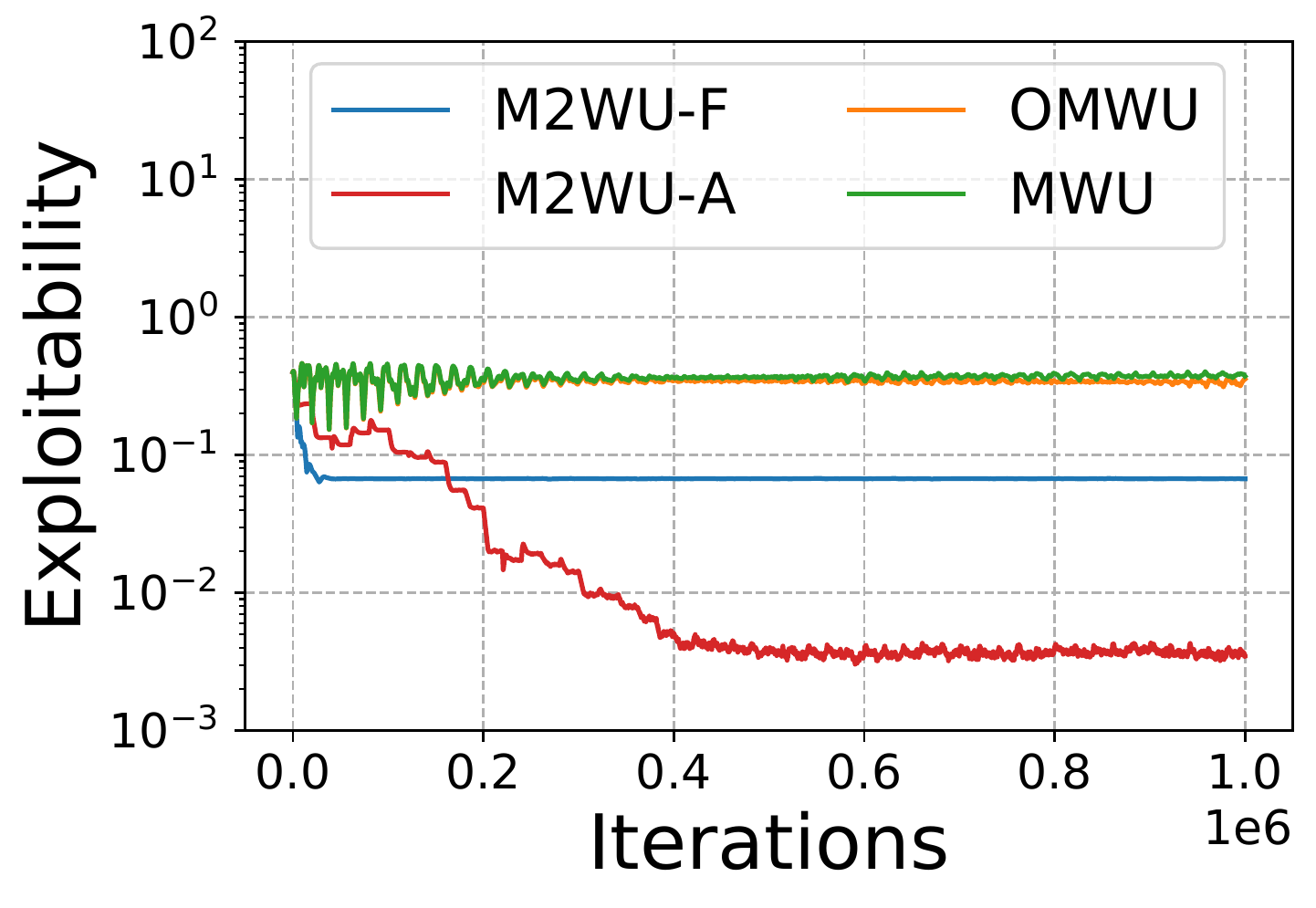}
        \subcaption{M-Ne}
    \end{minipage}
    \begin{minipage}[t]{0.24\textwidth}
        \centering
        \includegraphics[width=1.1\linewidth]{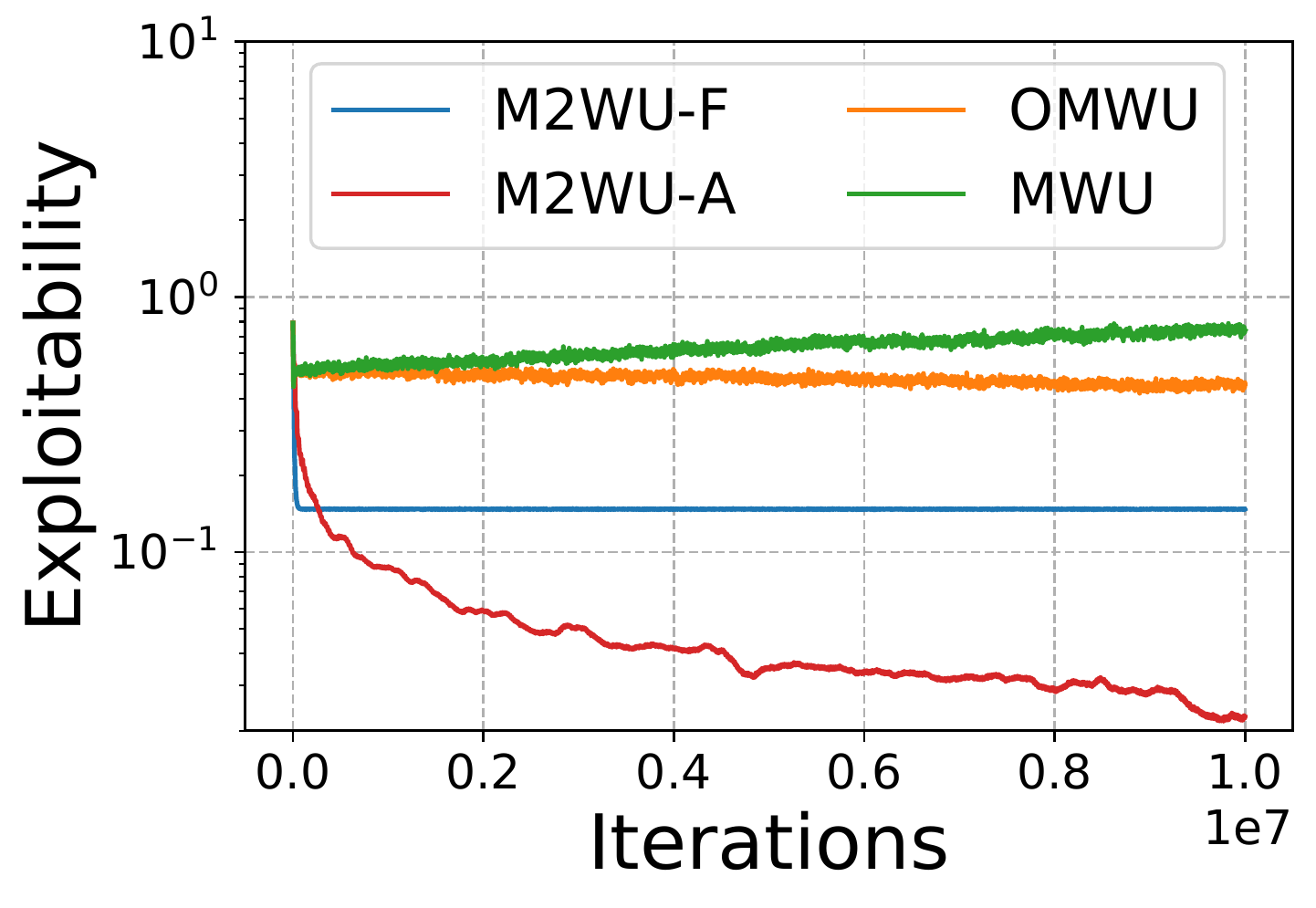}
        \subcaption{Random utility ($25\times 25$)}
    \end{minipage}
    \begin{minipage}[t]{0.24\textwidth}
        \centering
        \includegraphics[width=1.1\linewidth]{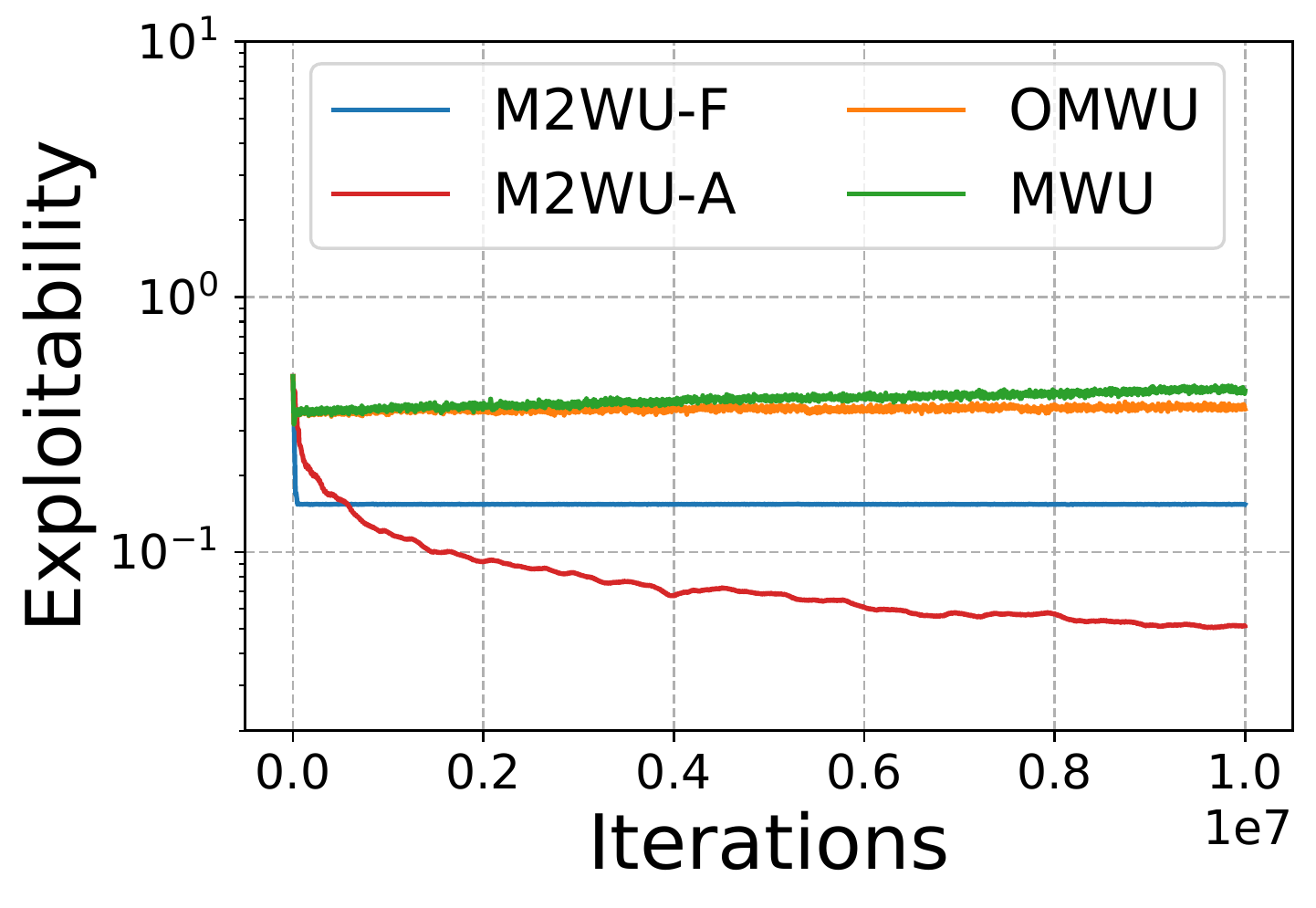}
        \subcaption{Random utility ($100\times 100$)}
    \end{minipage}
    \caption{
    Exploitability of $\pi^t$ for M2WU, MWU, and OMWU with noisy feedback.
    }
    \label{fig:exploitability_noisy}
\end{figure*}
\begin{figure*}[t!]
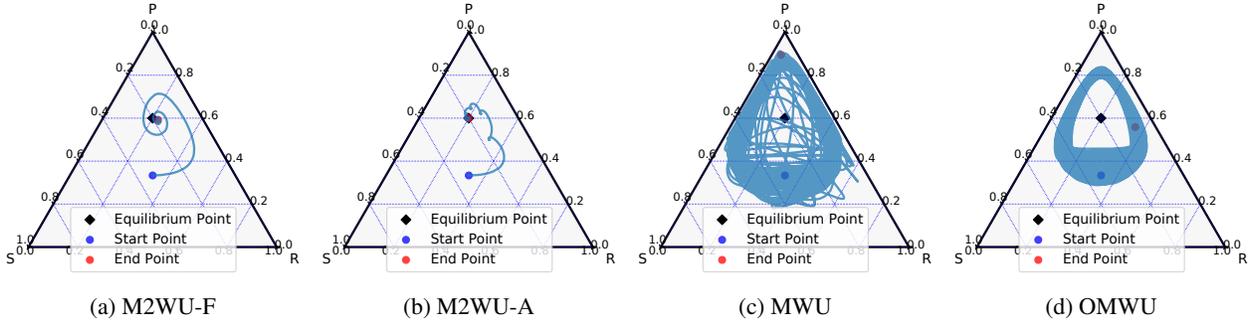

    \centering
    \begin{minipage}[t]{0.24\textwidth}
        \centering
        \includegraphics[width=1.0\linewidth]{figs/noisy_feedback/brps_noisy_feedback_M2WU_trajectories_triangle_0.pdf}
        \subcaption{M2WU-F}
    \end{minipage}
    \begin{minipage}[t]{0.24\textwidth}
        \centering
        \includegraphics[width=1.0\linewidth]{figs/noisy_feedback/brps_noisy_feedback_M2WU_adaptive_trajectories_triangle_0.pdf}
        \subcaption{M2WU-A}
    \end{minipage}
    \begin{minipage}[t]{0.24\textwidth}
        \centering
        \includegraphics[width=1.0\linewidth]{figs/noisy_feedback/brps_noisy_feedback_MWU_trajectories_triangle_0.pdf}
        \subcaption{MWU}
    \end{minipage}
    \begin{minipage}[t]{0.24\textwidth}
        \centering
        \includegraphics[width=1.0\linewidth]{figs/noisy_feedback/brps_noisy_feedback_OMWU_trajectories_triangle_0.pdf}
        \subcaption{OMWU}\label{fig:cycle with OMWU}
    \end{minipage}
    \caption{
    Trajectories of $\pi^t$ for M2WU, MWU and OMWU in BRPS with noisy feedback.
    We set the initial strategy to $\pi_i^0=(1/|A_i|)_{a\in A_i}$ for $i\in \{1, 2\}$.
    The black point represents the equilibrium strategy.
    The blue/red points represent the initial/final points, respectively.
    }
    \label{fig:trajectory_brps_noisy}
\end{figure*}
\fi
\label{sec:exp_noisy_feedback}
This section turns to the noisy feedback case, where payoffs observed in each period are perturbed. 
Players observe noisy estimates $\hat{q}_i^{\pi^t}(a_i) = {q}_i^{\pi^t}(a_i) + \xi^t(a_i) $ of the gradient vectors at each period, where the noise $\xi^t(a_i)$ is drawn from the Gaussian distribution $\mathcal{N}(0, 0.1^2)$ for $i\in \{1,2\},~a_i\in A_i$ and $t \ge 0$ in an i.i.d. manner. 
We here use the learning rate $\eta$ of $0.001$ and the mutation rate $\mu$ of $0.1$ for M2WU-F and $0.5$ for M2WU-A.  We also update it every $N=20,000$ period in Algorithm~\ref{alg:m2wu}, which is much longer than the full feedback case to handle noisy estimates. Note that, for OMWU, we use the noisy gradient vector $\hat{q}_i^{\pi^{t-1}}$ at the previous step $t-1$ as the prediction vector. 

Figure~\ref{fig:exploitability_noisy} averages the exploitability for the four algorithms in the four games. In all of the games, M2WU-F and \mbox{-A} always outperform MWU and OMWU. The two existing algorithms first oscillate, do not improve much, and remain far from equilibrium. As in the full feedback case, although M2WU-A converges less rapidly than M2WU-F, it admits significantly low exploitability. This tendency does not change even when the action space is larger and even when the learning rate is decayed.
We provide additional experimental results when using different learning rates $\eta\in \{0.1, 0.05, 0.01, 0.005\}$ in Appendix \ref{sec:appendix_experiments_noisy_varying_eta}.
Also, the case with the decayed learning rate, formally, $\eta_t=t^{-\frac{3}{4}}$, is placed in Appendix~\ref{sec:appendix_experiments_noisy_decay}. 

Figure~\ref{fig:trajectory_brps_noisy} demonstrates the trajectories of $\pi^t$ of each algorithm in BRPS. We fix the initial strategy to be uniform. Both M2WU-F and M2WU-A exhibit a clear convergence near the equilibrium (red) point, while MWU and OMWU do not at all. This result strongly supports Figure~\ref{fig:exploitability_noisy}(a). Note that M2WU-A obtains the strategy with lower exploitability than M2WU-F as Theorem~\ref{thm:direct_convergence} suggests.

\section{CONCLUSION}
In this paper, we proposed M2WU, an algorithm that utilizes a simple idea of stabilizing learning dynamics through mutation with a reference strategy. We proved that in both full and noisy feedback settings, the last-iterate strategy converges to the stationary point of RMD. In particular, we showed that such convergence occurs exponentially fast with a constant learning rate in the full feedback setting. Furthermore, last-iterate convergence to an exact Nash equilibrium was also proven by iteratively reusing the converged stationary point as a subsequent reference strategy. The numerical experiments showed that, even with the presence of noise, the strategy updated by M2WU exhibits a lower exploitability than MWU and OMWU. Future research could examine the convergence rate with noisy feedback and extend M2WU and its analyses to extensive-form games.

%% file: appendix.tex
\section{NOTATIONS}
In this section, we summarize the notations we use in Table~\ref{tb:notation}.
\begin{table}[h!]
    \centering
    \caption{Notations}
    \label{tb:notation}
    \begin{tabular}{cc} \hline
        Symbol & Description \\ \hline
        $A_i$ & Action set for player $i$ \\ 
        $u_i$ & Utility function for player $i$ \\
        $\pi_i$ & Strategy for player $i$ \\
        $\pi$ & Strategy profile \\
        $v_i^{\pi}$ & Player $i$'s expected utility for a given strategy profile $\pi$ \\
        $q_i^{\pi}$ & Player $i$'s conditional expected utility vector for a given strategy profile $\pi$ \\
        $\hat{q}_i^{\pi^t}$ & Player $i$'s noisy conditional expected utility vector at iteration $t$ \\
        $\xi^t$ & Noise vector at iteration $t$ \\
        $\pi^{\ast}$ & Nash equilibrium \\
        $\Pi_i^{\ast}$ & Set of Nash equilibria for player $i$ \\
        $\Delta(A_i)$ & Probability simplex on $A_i$ \\
        $\Delta^{\circ}(A_i)$ & Interior of $\Delta(A_i)$ \\
        $\mathrm{KL}(\cdot, \cdot)$ & Kullback-Leibler divergence \\
        $\eta_t$ & Learning rate at iteration $t$ \\
        $\mu$ & Mutation rate \\
        $r_i$ & Reference strategy \\
        $q_i^{\mu,t}$ & $\left(\hat{q}_i^{\pi^t}(a) + \frac{\mu}{\pi_i^t(a)}(r_i(a) - \pi_i^t(a))\right)_{a\in A_i}$ \\
        $\pi^{\mu,r}$ & Stationary point of (\ref{eq:rmd}) for given $\mu$ and $r$ \\ \hline
    \end{tabular}
\end{table}
\vfill

\section{PROOFS FOR THEOREM \ref{thm:kl_convergence}}
\label{sec:appendix_proof_kl_convergence}

\subsection{Proof of Theorem \ref{thm:kl_convergence}}
\label{sec:appendix_proof_kl_convergence_thm}
\begin{proof}[Proof of Theorem \ref{thm:kl_convergence}]
Let us define the following notation:
\begin{align*}
\Omega^{\mu,r} &= \left\{\pi\in \prod_{i=1}^2\Delta(A_i) ~|~ \mathrm{KL}(\pi^{\mu,r}, \pi) \leq \mathrm{KL}(\pi^{\mu,r}, \pi^0)\right\}, \\
\rho &= \min_{\pi\in \Omega^{\mu,r}}\min_{i\in \{1, 2\}, a\in A_i} \pi_i(a_i)>0, \\
\zeta&=\frac{1}{2u_{\max} + \frac{1}{\rho}\max_{i\in \{1,2\}, a\in A_i}r_i(a)}>0, \\
\alpha&=\min_{i\in \{1, 2\}, a\in A_i}\frac{r_i(a)}{\pi_i^{\mu,r}(a)} > 0, \\
\beta&=\frac{16}{\rho^2}\left(\max_{i\in \{1,2\}, a\in A_i}\frac{r_i(a)}{\pi_i^{\mu,r}(a)}\right)^2 > 0, \\
\gamma&= 16u_{\max}^2 > 0. 
\end{align*}

We prove the statement by mathematical induction.
Clearly, for $t=0$, we have $\mathrm{KL}(\pi^{\mu,r}, \pi^t)\leq \mathrm{KL}(\pi^{\mu,r},\pi^0)$ and $\pi^0\in \Omega^{\mu,r}$.
Let us assume that $\pi^t\in \Omega^{\mu,r}$, i.e., $\mathrm{KL}(\pi^{\mu,r},\pi^t)\leq \mathrm{KL}(\pi^{\mu,r},\pi^0)$.
Under the assumption that $\mathrm{KL}(\pi^{\mu,r},\pi^t)\leq \mathrm{KL}(\pi^{\mu,r},\pi^0)$, we have $\pi_i^t(a)\geq \rho$ for all $i\in \{1,2\}$ and $a\in A_i$.

We first derive the difference equation for $\mathrm{KL}(\pi^{\mu,r},\pi^t)$:
\begin{lemma}
\label{lem:bregman_div}
Let $\pi^{\mu,r}\in \prod_{i=1}^2 \Delta(A_i)$ be a stationary point of (\ref{eq:rmd}).
Then, $\pi^t$ updated by M2WU satisfies that:
\begin{align*}
    \mathrm{KL}(\pi^{\mu,r}, \pi^{t+1}) - \mathrm{KL}(\pi^{\mu,r}, \pi^t) =& \eta_t \sum_{i=1}^2\left(v_i^{\pi_i^t, \pi_{-i}^{\mu,r}} + \mu - \mu\sum_{a_i\in A_i}r_i(a_i)\frac{\pi_i^{\mu,r}(a_i)}{\pi_i^t(a_i)}\right) + \mathrm{KL}(\pi^t, \pi^{t+1}).
\end{align*}
\end{lemma}
Moreover, under the assumption that $\eta_t = \eta\in  (0,\min(\frac{\mu\alpha}{\mu^2\beta + \gamma},\zeta))$, the statement of Lemma \ref{lem:kl_path_ub} holds.
By combining Lemmas \ref{lem:rmd_divergence}, \ref{lem:kl_path_ub}, and \ref{lem:bregman_div}, we get:
\begin{align}
\label{eq:kl_div_ub1}
    \mathrm{KL}(\pi^{\mu,r}, \pi^{t+1}) - \mathrm{KL}(\pi^{\mu,r}, \pi^t) \leq& -\eta\mu \sum_{i=1}^2\sum_{a\in A_i}r_i(a)\left(\sqrt{\frac{\pi_i^t(a)}{\pi_i^{\mu,r}(a)}} - \sqrt{\frac{\pi_i^{\mu,r}(a)}{\pi_i^t(a)}}\right)^2 \nonumber\\
    &+ 8 \eta^2 \mu^2\sum_{i=1}^2 \sum_{a\in A_i}r_i(a)^2\left(\frac{1}{\pi_i^{\mu,r}(a)} - \frac{1}{\pi_i^t(a)}\right)^2 \nonumber\\
    &+ 8 \eta^2u_{\max}^2\sum_{i=1}^2 \|\pi_i^t - \pi_i^{\mu,r}\|_1^2.
\end{align}
We prove the lower bound on $\sum_{i=1}^2\sum_{a\in A_i}r_i(a)\left(\sqrt{\frac{\pi_i^t(a)}{\pi_i^{\mu,r}(a)}} - \sqrt{\frac{\pi_i^{\mu,r}(a)}{\pi_i^t(a)}}\right)^2$ as follows:
\begin{align}
\label{eq:rmd_divergence_lb}
    \sum_{i=1}^2\sum_{a\in A_i}r_i(a)\left(\sqrt{\frac{\pi_i^t(a)}{\pi_i^{\mu,r}(a)}} - \sqrt{\frac{\pi_i^{\mu,r}(a)}{\pi_i^t(a)}}\right)^2 &= \sum_{i=1}^2\sum_{a\in A_i}\frac{r_i(a)}{\pi_i^{\mu,r}(a)}\frac{\left(\pi_i^t(a) - \pi_i^{\mu,r}(a)\right)^2}{\pi_i^t(a)} \nonumber\\
    & \geq \left(\min_{i\in \{1, 2\}, a\in A_i}\frac{r_i(a)}{\pi_i^{\mu,r}(a)}\right)\sum_{i=1}^2\sum_{a\in A_i}\frac{\left(\pi_i^t(a) - \pi_i^{\mu,r}(a)\right)^2}{\pi_i^t(a)} \nonumber\\
    & \geq \left(\min_{i\in \{1, 2\}, a\in A_i}\frac{r_i(a)}{\pi_i^{\mu,r}(a)}\right)\sum_{i=1}^2\ln \left(1 + \sum_{a\in A_i}\frac{\left(\pi_i^t(a) - \pi_i^{\mu,r}(a)\right)^2}{\pi_i^t(a)}\right) \nonumber\\
    & = \left(\min_{i\in \{1, 2\}, a\in A_i}\frac{r_i(a)}{\pi_i^{\mu,r}(a)}\right)\sum_{i=1}^2\ln \left(\sum_{a\in A_i}\pi_i^{\mu,r}(a)\frac{\pi_i^{\mu,r}(a)}{\pi_i^t(a)}\right) \nonumber\\
    & \geq \left(\min_{i\in \{1, 2\}, a\in A_i}\frac{r_i(a)}{\pi_i^{\mu,r}(a)}\right)\sum_{i=1}^2\sum_{a\in A_i}\pi_i^{\mu,r}(a)\ln\left(\frac{\pi_i^{\mu,r}(a)}{\pi_i^t(a)}\right) \nonumber\\
    & = \left(\min_{i\in \{1, 2\}, a\in A_i}\frac{r_i(a)}{\pi_i^{\mu,r}(a)}\right)\mathrm{KL}(\pi^{\mu,r}, \pi^t),
\end{align}
where the second inequality follows from $x\geq \ln(1+x)$ for all $x>0$, and the third inequality follows from the concavity of the $\ln(\cdot)$ function and Jensen's inequality for concave functions.
Next, $\sum_{i=1}^2 \sum_{a\in A_i}r_i(a)^2\left(\frac{1}{\pi_i^{\mu,r}(a)} - \frac{1}{\pi_i^t(a)}\right)^2$ is upper bounded as follows:
\begin{align}
\label{eq:inverse_prob_ub}
    \sum_{i=1}^2 \sum_{a\in A_i}r_i(a)^2\left(\frac{1}{\pi_i^{\mu,r}(a)} - \frac{1}{\pi_i^t(a)}\right)^2 &= \sum_{i=1}^2 \sum_{a\in A_i}\left(\frac{r_i(a)}{\pi_i^{\mu,r}(a)\pi_i^t(a)}\right)^2\left(\pi_i^{\mu,r}(a) - \pi_i^t(a)\right)^2 \nonumber\\
    &\leq \frac{1}{\rho^2}\left(\max_{i\in \{1,2\}, a\in A_i}\frac{r_i(a)}{\pi_i^{\mu,r}(a)}\right)^2\sum_{i=1}^2\|\pi_i^{\mu,r} - \pi_i^t\|_2^2 \nonumber\\
    &\leq \frac{1}{\rho^2}\left(\max_{i\in \{1,2\}, a\in A_i}\frac{r_i(a)}{\pi_i^{\mu,r}(a)}\right)^2\sum_{i=1}^2\|\pi_i^{\mu,r} - \pi_i^t\|_1^2 \nonumber\\
    &\leq \frac{2}{\rho^2}\left(\max_{i\in \{1,2\}, a\in A_i}\frac{r_i(a)}{\pi_i^{\mu,r}(a)}\right)^2\mathrm{KL}(\pi^{\mu,r}, \pi^t),
\end{align}
where the last inequality follows from Pinsker's inequality \citep{Tsybakov:1315296}.
Similarly, $\sum_{i=1}^2 \|\pi_i^t - \pi_i^{\mu,r}\|_1^2$ is upper bounded as
\begin{align}
\label{eq:pinsker_ub}
    \sum_{i=1}^2 \|\pi_i^t - \pi_i^{\mu,r}\|_1^2 &\leq 2\mathrm{KL}(\pi^{\mu,r}, \pi^t).
\end{align}

By combining (\ref{eq:kl_div_ub1}), (\ref{eq:rmd_divergence_lb}), (\ref{eq:inverse_prob_ub}), and (\ref{eq:pinsker_ub}), we have:
\begin{align*}
    &\mathrm{KL}(\pi^{\mu,r}, \pi^{t+1}) - \mathrm{KL}(\pi^{\mu,r}, \pi^t) \nonumber\\
    &\leq -\eta\mu \left(\min_{i\in \{1, 2\}, a\in A_i}\frac{r_i(a)}{\pi_i^{\mu,r}(a)}\right)\mathrm{KL}(\pi^{\mu,r}, \pi^t) + 16 \eta^2 \left(\frac{\mu^2}{\rho^2}\left(\max_{i\in \{1,2\}, a\in A_i}\frac{r_i(a)}{\pi_i^{\mu,r}(a)}\right)^2 + u_{\max}^2\right)\mathrm{KL}(\pi^{\mu,r}, \pi^t) \nonumber\\
    &= \left(- \eta \mu\alpha + \eta^2 (\mu^2\beta + \gamma)\right)\mathrm{KL}(\pi^{\mu,r}, \pi^t).
\end{align*}
Thus, we get:
\begin{align*}
    \mathrm{KL}(\pi^{\mu,r}, \pi^{t+1}) \leq (1 - (\eta\mu\alpha - \eta^2 (\mu^2\beta + \gamma))) \mathrm{KL}(\pi^{\mu,r}, \pi^t),
\end{align*}
and then, for $\eta \in (0,\min(\frac{\mu\alpha}{\mu^2\beta + \gamma},\zeta))$:
\begin{align*}
    &\mathrm{KL}(\pi^{\mu,r}, \pi^{t+1}) - \mathrm{KL}(\pi^{\mu,r}, \pi^t) \leq 0.
\end{align*}
Thus, if $\eta \in (0,\min(\frac{\mu\alpha}{\mu^2\beta + \gamma},\zeta))$, then $\mathrm{KL}(\pi^{\mu,r}, \pi^{t+1}) \leq \mathrm{KL}(\pi^{\mu,r}, \pi^t) \leq \mathrm{KL}(\pi^{\mu,r}, \pi^0)$ and $\pi^{t+1}\in \Omega^{\mu,r}$ also hold.
By mathematical induction, if $\eta \in (0,\min(\frac{\mu\alpha}{\mu^2\beta + \gamma},\zeta))$, for all $t\geq 0$:
\begin{align*}
    \mathrm{KL}(\pi^{\mu,r}, \pi^{t+1}) \leq (1 - (\eta\mu\alpha - \eta^2 (\mu^2\beta + \gamma))) \mathrm{KL}(\pi^{\mu,r}, \pi^t) \leq \cdots \leq (1 - (\eta\mu\alpha - \eta^2 (\mu^2\beta + \gamma)))^{t+1} \mathrm{KL}(\pi^{\mu,r}, \pi^0).
\end{align*}
\end{proof}

\subsection{Proof of Lemma \ref{lem:rmd_divergence}}
\begin{proof}[Proof of Lemma \ref{lem:rmd_divergence}]
First, we introduce the following lemma from \citet{abe2022mutationdriven}:
\begin{lemma}[Lemma 5.6 of \citet{abe2022mutationdriven}]
\label{lem:rmd_property}
Let $\pi_i^{\mu,r}\in \Delta(A_i)$ be a stationary point of (\ref{eq:rmd}) for $i\in \{1, 2\}$.
Then, for any $\pi_i' \in \Delta(A_i)$:
\begin{align*}
    &v_i^{\pi'_i, \pi_{-i}^{\mu,r}} = v_i^{\pi^{\mu,r}} + \mu - \mu\sum_{a_i\in A_i}r_i(a_i)\frac{\pi_i'(a_i)}{\pi_i^{\mu,r}(a_i)}.
\end{align*}
\end{lemma}
From this lemma, we have:
\begin{align*}
    \sum_{i=1}^2v_i^{\pi_i^t, \pi_{-i}^{\mu,r}} + 2\mu - \mu\sum_{i=1}^2\sum_{a_i\in A_i}r_i(a_i)\frac{\pi_i^{\mu,r}(a_i)}{\pi_i^t(a_i)} &= \sum_{i=1}^2v_i^{\pi^{\mu,r}} + 4\mu - \mu\sum_{i=1}^2\sum_{a_i\in A_i}r_i(a_i)\left(\frac{\pi_i^t(a_i)}{\pi_i^{\mu,r}(a_i)} + \frac{\pi_i^{\mu,r}(a_i)}{\pi_i^t(a_i)}\right) \\
    &= 4\mu - \mu\sum_{i=1}^2\sum_{a_i\in A_i}r_i(a_i)\left(\frac{\pi_i^t(a_i)}{\pi_i^{\mu,r}(a_i)} + \frac{\pi_i^{\mu,r}(a_i)}{\pi_i^t(a_i)}\right) \\
    &= -\mu \sum_{i=1}^2\sum_{a_i\in A_i}r_i(a_i)\left(\sqrt{\frac{\pi_i^t(a_i)}{\pi_i^{\mu,r}(a_i)}} - \sqrt{\frac{\pi_i^{\mu,r}(a_i)}{\pi_i^t(a_i)}}\right)^2,
\end{align*}
where the second equality follows from $\sum_{i=1}^2 v_i^{\pi^{\mu,r}} = 0$ by the definition of zero-sum games.
\end{proof}

\subsection{Proof of Lemma \ref{lem:kl_path_ub}}
\begin{proof}[Proof of Lemma \ref{lem:kl_path_ub}]
Let us assume that $\eta\in  (0,\zeta)$, where $\alpha$, $\beta$, $\gamma$, and $\zeta$ are defined in Appendix \ref{sec:appendix_proof_kl_convergence_thm}.

First, we have:
\begin{align*}
    \mathrm{KL}(\pi^t, \pi^{t+1}) &= \sum_{i=1}^2 \sum_{a\in A_i}\pi_i^t(a)\ln \frac{\pi_i^t(a)}{\pi_i^{t+1}(a)} \\
    &= 2\sum_{i=1}^2 \sum_{a\in A_i}\frac{1}{2}\pi_i^t(a)\ln \frac{\pi_i^t(a)}{\pi_i^{t+1}(a)} \leq 2 \ln \left(\frac{1}{2}\sum_{i=1}^2 \sum_{a\in A_i}\pi_i^t(a)\frac{\pi_i^t(a)}{\pi_i^{t+1}(a)}\right),
\end{align*}
where the inequality follows from the concavity of the $\ln(\cdot)$ function and Jensen's inequality for concave functions.
Here, from the update rule (\ref{eq:m2wu}):
\begin{align*}
    \frac{\pi_i^t(a)}{\pi_i^{t+1}(a)} = \frac{\sum_{a'\in A_i}\pi_i^t(a')\exp\left(\eta\left(q_i^{\pi^t}(a') + \mu \frac{r_i(a')}{\pi_i^t(a')}\right)\right)}{\exp\left(\eta\left(q_i^{\pi^t}(a) + \mu \frac{r_i(a)}{\pi_i^t(a)}\right)\right)},
\end{align*}
and then we get:
\begin{align}
\label{eq:kl_diff_upper}
    \mathrm{KL}(\pi^t, \pi^{t+1}) &\leq 2\ln \left(\frac{1}{2}\sum_{i=1}^2 \sum_{a\in A_i}\pi_i^t(a)\frac{\sum_{a'\in A_i}\pi_i^t(a')\exp\left(\eta\left(q_i^{\pi^t}(a') + \mu \frac{r_i(a')}{\pi_i^t(a')}\right)\right)}{\exp\left(\eta\left(q_i^{\pi^t}(a) + \mu \frac{r_i(a)}{\pi_i^t(a)}\right)\right)}\right) \nonumber \\
    &= 2\ln \left(\frac{1}{2}\sum_{i=1}^2 \sum_{a\in A_i}\sum_{a'\in A_i}\pi_i^t(a)\pi_i^t(a')\exp\left(\eta\left(q_i^{\pi^t}(a') + \mu \frac{r_i(a')}{\pi_i^t(a')} - q_i^{\pi^t}(a) - \mu \frac{r_i(a)}{\pi_i^t(a)}\right)\right)\right).
\end{align}

Furthermore, from the assumption for the learning rate $\eta < \frac{1}{2u_{\max} + \frac{\mu}{\rho}\max_{i\in \{1,2\}, a\in A_i}r_i(a)} \leq \zeta$, we have $\eta\left(q_i^{\pi^t}(a') + \mu \frac{r_i(a')}{\pi_i^t(a')} - q_i^{\pi^t}(a) - \mu \frac{r_i(a)}{\pi_i^t(a)}\right) \leq 1$.
Thus, we can use the fact that $\exp(x) \leq 1 + x + x^2$ for $x\leq 1$, and then:
\begin{align}
\label{eq:exp_upper}
    \frac{1}{2}\sum_{i=1}^2&\sum_{a\in A_i}\sum_{a'\in A_i}\pi_i^t(a)\pi_i^t(a')\exp\left(\eta\left(q_i^{\pi^t}(a') + \mu \frac{r_i(a')}{\pi_i^t(a')} - q_i^{\pi^t}(a) - \mu \frac{r_i(a)}{\pi_i^t(a)}\right)\right) \nonumber\\
    \leq& \frac{1}{2}\sum_{i=1}^2\sum_{a\in A_i}\sum_{a'\in A_i}\pi_i^t(a)\pi_i^t(a')\left(1 + \eta\left(q_i^{\pi^t}(a') + \mu \frac{r_i(a')}{\pi_i^t(a')} - q_i^{\pi^t}(a) - \mu \frac{r_i(a)}{\pi_i^t(a)}\right)\right) \nonumber\\ 
    &+ \frac{1}{2}\sum_{i=1}^2\sum_{a\in A_i}\sum_{a'\in A_i}\pi_i^t(a)\pi_i^t(a')\left(\eta^2\left(q_i^{\pi^t}(a') + \mu \frac{r_i(a')}{\pi_i^t(a')} - q_i^{\pi^t}(a) - \mu \frac{r_i(a)}{\pi_i^t(a)}\right)^2\right) \nonumber\\
    =& 1 + \frac{\eta^2}{2}\sum_{i=1}^2\sum_{a\in A_i}\sum_{a'\in A_i}\pi_i^t(a)\pi_i^t(a')\left(q_i^{\pi^t}(a') + \mu \frac{r_i(a')}{\pi_i^t(a')} - q_i^{\pi^t}(a) - \mu \frac{r_i(a)}{\pi_i^t(a)}\right)^2 \nonumber\\
    \leq& \exp\left(\frac{\eta^2}{2}\sum_{i=1}^2\sum_{a\in A_i}\sum_{a'\in A_i}\pi_i^t(a)\pi_i^t(a')\left(q_i^{\pi^t}(a') + \mu \frac{r_i(a')}{\pi_i^t(a')} - q_i^{\pi^t}(a) - \mu \frac{r_i(a)}{\pi_i^t(a)}\right)^2\right),
\end{align}
where the first equality follows from $\sum_{a\in A_i}\sum_{a'\in A_i}\pi_i^t(a)\pi_i^t(a')\left(\eta\left(q_i^{\pi^t}(a') + \mu \frac{r_i(a')}{\pi_i^t(a')} - q_i^{\pi^t}(a) - \mu \frac{r_i(a)}{\pi_i^t(a)}\right)\right) = 0$, and the last inequality follows from $1 + x \leq \exp(x)$ for $x\in \mathbb{R}$.
By combining (\ref{eq:kl_diff_upper}) and (\ref{eq:exp_upper}), we get:
\begin{align}
    \label{eq:kl_diff_upper_by_q_diff}
    \mathrm{KL}(\pi^t, \pi^{t+1}) &\leq 2\ln \left(\exp\left(\frac{\eta^2}{2}\sum_{i=1}^2\sum_{a\in A_i}\sum_{a'\in A_i}\pi_i^t(a)\pi_i^t(a')\left(q_i^{\pi^t}(a') + \mu \frac{r_i(a')}{\pi_i^t(a')} - q_i^{\pi^t}(a) - \mu \frac{r_i(a)}{\pi_i^t(a)}\right)^2\right)\right) \nonumber\\
    &= \eta^2\sum_{i=1}^2\sum_{a\in A_i}\sum_{a'\in A_i}\pi_i^t(a)\pi_i^t(a')\left(q_i^{\pi^t}(a') + \mu \frac{r_i(a')}{\pi_i^t(a')} - q_i^{\pi^t}(a) - \mu \frac{r_i(a)}{\pi_i^t(a)}\right)^2.
\end{align}

Here, by using the ordinary differential equation (\ref{eq:rmd}), we have for all $i\in \{1,2\}$ and $a\in A$:
\begin{align*}
    q_i^{\pi^{\mu,r}}(a) = v_i^{\pi^{\mu,r}} - \frac{\mu}{\pi_i^{\mu,r}(a)}\left( r_i(a) - \pi_i^{\mu,r}(a)\right).
\end{align*}
Thus,
\begin{align*}
    &q_i^{\pi^t}(a') + \mu \frac{r_i(a')}{\pi_i^t(a')} - q_i^{\pi^t}(a) - \mu \frac{r_i(a)}{\pi_i^t(a)} \\
    &= q_i^{\pi^{\mu,r}}(a') + \mu \frac{r_i(a')}{\pi_i^t(a')} - q_i^{\pi^{\mu,r}}(a) - \mu \frac{r_i(a)}{\pi_i^t(a)} + q_i^{\pi^t}(a') - q_i^{\pi^{\mu,r}}(a') - q_i^{\pi^t}(a) + q_i^{\pi^{\mu,r}}(a) \\
    &= \mu\left(\frac{r_i(a)}{\pi_i^{\mu,r}(a)} - \frac{r_i(a)}{\pi_i^t(a)} - \frac{r_i(a')}{\pi_i^{\mu,r}(a')} + \frac{r_i(a')}{\pi_i^t(a')} \right) + q_i^{\pi^t}(a') - q_i^{\pi^{\mu,r}}(a') - q_i^{\pi^t}(a) + q_i^{\pi^{\mu,r}}(a).
\end{align*}
Then,
\begin{align}
\label{eq:expected_squared_diff}
    \sum_{i=1}^2 &\sum_{a\in A_i}\sum_{a'\in A_i}\pi_i^t(a)\pi_i^t(a') \left(q_i^{\pi^t}(a') + \mu \frac{r_i(a')}{\pi_i^t(a')} - q_i^{\pi^t}(a) - \mu \frac{r_i(a)}{\pi_i^t(a)}\right)^2 \nonumber\\
    =& \sum_{i=1}^2 \sum_{a\in A_i}\sum_{a'\in A_i}\pi_i^t(a)\pi_i^t(a')\left(\mu\left(\frac{r_i(a)}{\pi_i^{\mu,r}(a)} - \frac{r_i(a)}{\pi_i^t(a)} - \frac{r_i(a')}{\pi_i^{\mu,r}(a')} + \frac{r_i(a')}{\pi_i^t(a')}\right) + q_i^{\pi^t}(a') - q_i^{\pi^{\mu,r}}(a') - q_i^{\pi^t}(a) + q_i^{\pi^{\mu,r}}(a) \right)^2 \nonumber\\
    \leq& 2\mu^2\sum_{i=1}^2 \sum_{a\in A_i}\sum_{a'\in A_i}\pi_i^t(a)\pi_i^t(a')\left(\frac{r_i(a)}{\pi_i^{\mu,r}(a)} - \frac{r_i(a)}{\pi_i^t(a)} - \frac{r_i(a')}{\pi_i^{\mu,r}(a')} + \frac{r_i(a')}{\pi_i^t(a')} \right)^2 \nonumber\\
    &+ 2\sum_{i=1}^2 \sum_{a\in A_i}\sum_{a'\in A_i}\pi_i^t(a)\pi_i^t(a')\left(q_i^{\pi^t}(a') - q_i^{\pi^{\mu,r}}(a') - q_i^{\pi^t}(a) + q_i^{\pi^{\mu,r}}(a)\right)^2 \nonumber\\
    \leq& 4\mu^2\sum_{i=1}^2 \sum_{a\in A_i}\sum_{a'\in A_i}\pi_i^t(a)\pi_i^t(a')\left(\left(\frac{r_i(a)}{\pi_i^{\mu,r}(a)} - \frac{r_i(a)}{\pi_i^t(a)}\right)^2 + \left(\frac{r_i(a')}{\pi_i^{\mu,r}(a')} - \frac{r_i(a')}{\pi_i^t(a')}\right)^2 \right) \nonumber\\
    &+ 4\sum_{i=1}^2 \sum_{a\in A_i}\sum_{a'\in A_i}\pi_i^t(a)\pi_i^t(a')\left(\left(q_i^{\pi^t}(a') - q_i^{\pi^{\mu,r}}(a')\right)^2 + \left(q_i^{\pi^t}(a) - q_i^{\pi^{\mu,r}}(a)\right)^2\right) \nonumber\\
    =& 8\mu^2\sum_{i=1}^2 \sum_{a\in A_i}\pi_i^t(a)r_i(a)^2\left(\frac{1}{\pi_i^{\mu,r}(a)} - \frac{1}{\pi_i^t(a)}\right)^2 + 8\sum_{i=1}^2 \sum_{a\in A_i}\pi_i^t(a)\left(\sum_{b\in A_{-i}}\left(\pi_{-i}^t(b) - \pi_{-i}^{\mu,r}(b)\right)u_i(a',b)\right)^2 \nonumber\\
    \leq& 8\mu^2\sum_{i=1}^2 \sum_{a\in A_i}\pi_i^t(a)r_i(a)^2\left(\frac{1}{\pi_i^{\mu,r}(a)} - \frac{1}{\pi_i^t(a)}\right)^2 + 8\sum_{i=1}^2 \sum_{a\in A_i}\pi_i^t(a)u_{\max}^2\|\pi_{-i}^t - \pi_{-i}^{\mu,r}\|_1^2 \nonumber\\
    \leq& 8\mu^2\sum_{i=1}^2 \sum_{a\in A_i}r_i(a)^2\left(\frac{1}{\pi_i^{\mu,r}(a)} - \frac{1}{\pi_i^t(a)}\right)^2 + 8\sum_{i=1}^2 u_{\max}^2\|\pi_{-i}^t - \pi_{-i}^{\mu,r}\|_1^2 \nonumber\\
    =& 8\mu^2\sum_{i=1}^2 \sum_{a\in A_i}r_i(a)^2\left(\frac{1}{\pi_i^{\mu,r}(a)} - \frac{1}{\pi_i^t(a)}\right)^2 + 8\sum_{i=1}^2 u_{\max}^2\|\pi_i^t - \pi_i^{\mu,r}\|_1^2,
\end{align}
where the first and second inequalities follow from $(a+b)^2\leq 2(a^2 + b^2)$ for $a,b\in \mathbb{R}$, and the third inequality follows from H\"{o}lder's inequality.

By combining (\ref{eq:kl_diff_upper_by_q_diff}) and (\ref{eq:expected_squared_diff}), if $\eta\in  (0,\zeta)$, for all $t\geq 0$:
\begin{align*}
    \mathrm{KL}(\pi^t, \pi^{t+1}) \leq 8 \eta^2 \left(\mu^2\sum_{i=1}^2 \sum_{a\in A_i}r_i(a)^2\left(\frac{1}{\pi_i^{\mu,r}(a)} - \frac{1}{\pi_i^t(a)}\right)^2 + u_{\max}^2\sum_{i=1}^2 \|\pi_i^t - \pi_i^{\mu,r}\|_1^2\right).
\end{align*}
\end{proof}

\subsection{Proof of Lemma \ref{lem:bregman_div}}
\begin{proof}[Proof of Lemma \ref{lem:bregman_div}]
We introduce the following lemma:
\begin{lemma}
\label{lem:convex_conjugate}
For any $\pi\in \prod_{i=1}^2\Delta(A_i)$, $\pi^t$ updated by M2WU satisfies that:
\begin{align*}
    \mathrm{KL}(\pi, \pi^t) = \sum_{i=1}^2\left(\left\langle \sum_{s=1}^{t-1}\eta_s q_i^{\mu,s}, \pi_i^t\right\rangle - \psi_i(\pi_i^t) -\left\langle \sum_{s=1}^{t-1}\eta_s q_i^{\mu,s}, \pi_i\right\rangle + \psi_i(\pi_i)\right),
\end{align*}
where $\psi_i(p) = \sum_{a\in A_i}p(a)\ln p(a)$.
\end{lemma}
From Lemma \ref{lem:convex_conjugate}, we have:
\begin{align*}
    \mathrm{KL}&(\pi^{\mu,r}, \pi^{t+1}) - \mathrm{KL}(\pi^{\mu,r}, \pi^t) \\
    =& \sum_{i=1}^2\left(\left\langle \sum_{s=1}^t\eta_s q_i^{\mu,s}, \pi_i^{t+1}\right\rangle - \psi_i(\pi_i^{t+1}) - \left\langle \sum_{s=1}^t\eta_s q_i^{\mu,s}, \pi_i^{\mu,r}\right\rangle + \psi_i(\pi_i^{\mu,r})\right) \\
    &-\sum_{i=1}^2\left(\left\langle \sum_{s=1}^{t-1}\eta_s q_i^{\mu,s}, \pi_i^t\right\rangle - \psi_i(\pi_i^t) - \left\langle \sum_{s=1}^{t-1}\eta_s q_i^{\mu,s}, \pi_i^{\mu,r}\right\rangle + \psi_i(\pi_i^{\mu,r})\right) \\
    =& \sum_{i=1}^2\left(\left\langle \sum_{s=1}^t\eta_s q_i^{\mu,s}, \pi_i^{t+1}\right\rangle - \psi_i(\pi_i^{t+1}) - \left\langle \sum_{s=1}^t\eta_s q_i^{\mu,s}, \pi_i^t\right\rangle + \psi_i(\pi_i^t)\right) - \eta_t\sum_{i=1}^2\langle  q_i^{\mu, t}, \pi_i^{\mu,r} - \pi_i^t \rangle \\
    =& \mathrm{KL}(\pi^t, \pi^{t+1}) - \eta_t \sum_{i=1}^2\langle q_i^{\mu, t}, \pi_i^{\mu,r} - \pi_i^t \rangle \\
    =& \mathrm{KL}(\pi^t, \pi^{t+1}) + \eta_t \sum_{i=1}^2\sum_{a\in A_i}\left(q_i^{\pi^t}(a) + \frac{\mu}{\pi_i^t(a)}\left(r_i(a) - \pi_i^t(a)\right)\right)\left(\pi_i^t(a) - \pi_i^{\mu,r}(a)\right) \\
    =& \mathrm{KL}(\pi^t, \pi^{t+1}) + \eta_t \sum_{i=1}^2\sum_{a\in A_i}\left(\pi_i^t(a) - \pi_i^{\mu,r}(a)\right)\left(q_i^{\pi^t}(a) + \mu\frac{r_i(a)}{\pi_i^t(a)}\right) \\
    =& \mathrm{KL}(\pi^t, \pi^{t+1}) + \eta_t \sum_{i=1}^2\left(v_i^{\pi^t} - v_i^{\pi_i^{\mu,r}, \pi_{-i}^t} + \mu - \mu\sum_{a\in A_i}r_i(a)\frac{\pi_i^{\mu,r}(a)}{\pi_i^t(a)}\right) \\
    =& \mathrm{KL}(\pi^t, \pi^{t+1}) + \eta_t \sum_{i=1}^2\left(- v_i^{\pi_i^{\mu,r}, \pi_{-i}^t} + \mu - \mu\sum_{a\in A_i}r_i(a)\frac{\pi_i^{\mu,r}(a)}{\pi_i^t(a)}\right) \\
    =& \mathrm{KL}(\pi^t, \pi^{t+1}) + \eta_t \sum_{i=1}^2\left(v_i^{\pi_i^t, \pi_{-i}^{\mu,r}} + \mu - \mu\sum_{a\in A_i}r_i(a)\frac{\pi_i^{\mu,r}(a)}{\pi_i^t(a)}\right),
\end{align*}
where the seventh equality follows from $\sum_{i=1}^2v_i^{\pi^t}=0$, and the last equality follows from $-v_1^{\pi_1^{\mu,r},\pi_2^t}=v_2^{\pi_1^{\mu,r},\pi_2^t}$ and $-v_2^{\pi_1^t,\pi_2^{\mu,r}}=v_1^{\pi_1^t,\pi_2^{\mu,r}}$ by the definition of two-player zero-sum games.
\end{proof}

\subsection{Proof of Lemma \ref{lem:convex_conjugate}}
\begin{proof}[Proof of Lemma \ref{lem:convex_conjugate}]
From the definition of the Kullback-Leibler divergence, we have:
\begin{align}
    \label{eq:kl_div}
    \mathrm{KL}(\pi, \pi^t) &= \sum_{i=1}^2\mathrm{KL}(\pi_i, \pi_i^t) = \sum_{i=1}^2 \sum_{a\in A_i}\pi_i(a)\ln \frac{\pi_i(a)}{\pi_i^t(a)} \nonumber\\
    &= \sum_{i=1}^2 \left(\sum_{a\in A_i} \left(\pi_i^t(a) - \pi_i(a)\right)\ln \pi_i^t(a) -\sum_{a\in A_i} \pi_i^t(a)\ln \pi_i^t(a) + \sum_{a\in A_i} \pi_i(a)\ln \pi_i(a)\right)\nonumber\\
    &= \sum_{i=1}^2 \left(\sum_{a\in A_i} \left(\pi_i^t(a) - \pi_i(a)\right)\ln \pi_i^t(a) -\psi_i(\pi_i^t) + \psi_i(\pi_i)\right).
\end{align}
Here, the update rule (\ref{eq:m2wu}) is equivalent to:
\begin{align*}
    \pi_i^t(a) = \frac{\exp\left(\sum_{s=1}^{t-1}\eta_s q_i^{\mu,s}(a)\right)}{\sum_{a'\in A_i}\exp\left(\sum_{s=1}^{t-1}\eta_s q_i^{\mu,s}(a')\right)},
\end{align*}
and then we have:
\begin{align}
    \label{eq:relative_entropy}
    \sum_{a\in A_i} \left(\pi_i^t(a) - \pi_i(a)\right)\ln \pi_i^t(a) &= \sum_{a\in A_i} \left(\pi_i^t(a) - \pi_i(a)\right)\left(\sum_{s=1}^{t-1}\eta_s q_i^{\mu,s}(a) - \ln \left(\sum_{a'\in A_i}\exp\left(\sum_{s=1}^{t-1}\eta_s q_i^{\mu,s}(a')\right)\right)\right) \nonumber\\
    &= \sum_{a\in A_i} \left(\pi_i^t(a) - \pi_i(a)\right)\sum_{s=1}^{t-1}\eta_s q_i^{\mu,s}(a) \nonumber\\
    &= \left\langle \sum_{s=1}^{t-1}\eta_s q_i^{\mu,s}, \pi_i^t\right\rangle - \left\langle  \sum_{s=1}^{t-1}\eta_s q_i^{\mu,s}, \pi_i\right\rangle.
\end{align}
By combining (\ref{eq:kl_div}) and (\ref{eq:relative_entropy}), we get:
\begin{align*}
    \mathrm{KL}(\pi, \pi^t) = \sum_{i=1}^2\left(\left\langle \sum_{s=1}^{t-1}\eta_s q_i^{\mu,s}, \pi_i^t\right\rangle - \psi_i(\pi_i^t) -\left\langle \sum_{s=1}^{t-1}\eta_s q_i^{\mu,s}, \pi_i\right\rangle + \psi_i(\pi_i)\right).
\end{align*}
\end{proof}

\section{PROOF OF COROLLARY \ref{cor:exploitability_convergence}}
\label{sec:appendix_proof_exploitability_convergence}

\begin{proof}[Proof of Corollary \ref{cor:exploitability_convergence}]
From the definition of exploitability, we have:
\begin{align}
\label{eq:exploitability_bound_pi_t}
    \mathrm{explt}(\pi^t) &= \sum_{i=1}^2 \max_{\tilde{\pi}_i\in \Delta(A_i)} v_i^{\tilde{\pi}_i, \pi_{-i}^t} \nonumber\\
    &= \sum_{i=1}^2 \left(\max_{\tilde{\pi}_i\in \Delta(A_i)} v_i^{\tilde{\pi}_i, \pi_{-i}^{\mu,r}} + \max_{\tilde{\pi}_i\in \Delta(A_i)} v_i^{\tilde{\pi}_i, \pi_{-i}^t}- \max_{\tilde{\pi}_i\in \Delta(A_i)} v_i^{\tilde{\pi}_i, \pi_{-i}^{\mu,r}}\right) \nonumber\\
    &= \mathrm{explt}(\pi^{\mu,r}) + \sum_{i=1}^2 \left(\max_{\tilde{\pi}_i\in \Delta(A_i)} v_i^{\tilde{\pi}_i, \pi_{-i}^t}- \max_{\tilde{\pi}_i\in \Delta(A_i)} v_i^{\tilde{\pi}_i, \pi_{-i}^{\mu,r}}\right) \nonumber\\
    &\leq \mathrm{explt}(\pi^{\mu,r}) + \sum_{i=1}^2 \left(\max_{\tilde{\pi}_i\in \Delta(A_i)} \left(v_i^{\tilde{\pi}_i, \pi_{-i}^t}- v_i^{\tilde{\pi}_i, \pi_{-i}^{\mu,r}}\right)\right) \nonumber\\
    &\leq \mathrm{explt}(\pi^{\mu,r}) + \sum_{i=1}^2 \left(\|\pi_i^{\mu,r} - \pi_i^t\|_1\max_{\tilde{\pi}_{-i}\in \Delta(A_{-i})} \|q_i^{\pi_i^t, \tilde{\pi}_{-i}}\|_{\infty}\right) \nonumber\\
    &\leq \mathrm{explt}(\pi^{\mu,r}) + \sum_{i=1}^2 \left( u_{\max}\sqrt{2\mathrm{KL}(\pi_i^{\mu,r}, \pi_i^t)}\right) \nonumber\\
    &\leq \mathrm{explt}(\pi^{\mu,r}) + u_{\max}\sqrt{2}\sqrt{2\sum_{i=1}^2\mathrm{KL}(\pi_i^{\mu,r}, \pi_i^t)} \nonumber\\
    &= \mathrm{explt}(\pi^{\mu,r}) + 2u_{\max}\sqrt{\mathrm{KL}(\pi^{\mu,r}, \pi^t)},
\end{align}
where the second inequality follows from H\"{o}lder's inequality, the third inequality follows from Pinsker's inequality \citep{Tsybakov:1315296}, and the fourth inequality follows from $\sqrt{a} + \sqrt{b}\leq \sqrt{2(a+b)}$ for $a,b>0$.
By combining (\ref{eq:exploitability_bound_pi_t}) and Theorem \ref{thm:kl_convergence}, we have:
\begin{align}
\label{eq:exploitability_bound_pi_t_mu}
    &\mathrm{explt}(\pi^t)  \leq \mathrm{explt}(\pi^{\mu,r}) + 2u_{\max}\sqrt{\mathrm{KL}(\pi^{\mu,r}, \pi^0)}(1 - C_2)^{\frac{t}{2}}.
\end{align}

Moreover, from Lemma 3.5 of \citet{Bauer:2019}, a stationary point $\pi^{\mu,r}$ of (\ref{eq:rmd}) satisfies that for all $i\in \{1, 2\}$ and $a_i\in A_i$, $q_i^{\pi^{\mu,r}}(a_i) - v_i^{\pi^{\mu,r}} \leq \mu$.
Therefore, the term of $\mathrm{exploit}(\pi^{\mu,r})$ can be bounded as:
\begin{align}
\label{eq:exploitability_bound_pi_mu}
    \mathrm{explt}(\pi^{\mu,r}) &= \sum_{i=1}^2 \max_{\tilde{\pi}_i\in \Delta(A_i)} v_i^{\tilde{\pi}_i, \pi_{-i}^{\mu,r}} \nonumber\\
    &= \sum_{i=1}^2 \left(\max_{\tilde{\pi}_i\in \Delta(A_i)} v_i^{\tilde{\pi}_i, \pi_{-i}^{\mu,r}} - v_i^{\pi^{\mu,r}}\right) \nonumber\\
    &= \sum_{i=1}^2 \left(\max_{a_i\in A_i} q_i^{\pi^{\mu,r}}(a_i) - v_i^{\pi^{\mu,r}}\right) \leq 2\mu,
\end{align}
where the first equality follows from $\sum_{i=1}^2v_i^{\pi^{\mu,r}}=0$ by the definition of zero-sum games.
By combining (\ref{eq:exploitability_bound_pi_t_mu}) and (\ref{eq:exploitability_bound_pi_mu}), we have:
\begin{align*}
    &\mathrm{explt}(\pi^t)  \leq 2\mu  + 2u_{\max}\sqrt{\mathrm{KL}(\pi^{\mu,r}, \pi^t)},
\end{align*}
This concludes the statement.
\end{proof}

\section{PROOFS FOR THEOREM~\ref{thm:conv_semibandit}}
\label{sec:appendix_proof_conv_semibandit}
\subsection{Proof of Theorem~\ref{thm:conv_semibandit}}
In preparation for the proof, we first define the notion of {\it approximate Robbins-Monro algorithms} \citep{robbins1951stochastic,benaim1999dynamics}. 
\begin{definition}\label{def:aprx_rm_alg}
The stochastic approximation algorithm 
$$
z(t + 1) = z(t) + \eta_t (F(z(t))  + U_t + \beta_t)
$$ 
is refer to as an approximate Robbins-Monro algorithm if the following conditions are satisfied.
\begin{itemize}
    \item $F: \mathbb{R}^m \to \mathbb{R}^m$ is a continuous function
    \item $ (U_t)_{t \ge 1}  \; \text{s.t.} \; U_n \in \mathbb{R}^m, \forall n \in \mathbb{N}$ is a martingale difference noise
    \item $(\eta_t)_{t \ge 1}$ is a given sequence of numbers such that $\sum_{t=1}^\infty \eta_t = \infty$ and $\lim_{t \to \infty } \eta_t = 0$
    \item $\lim_{t \to \infty} \beta_t =0$ almost surely
\end{itemize}
\end{definition}

We provide the definition of the asymptotic pseudo-trajectory. 

\begin{definition}[\citep{benaim1996asymptotic}]
A flow $\phi$ on a metric space $(M, d)$ is a continuous mapping
\begin{align}
 \phi : \mathbb{R} \times M \mapsto M,  \qquad (t, x) \mapsto \phi_t(x)
\end{align}
such that $ \phi_0(x) = x$ and $ \phi_{t + \alpha} = \phi_{t} \circ \phi_\alpha$ for all $ t, \alpha \in \mathbb{R}$. 
For a metric space $(M, d)$, a continuous function $X : \mathbb{R} \mapsto M$ is an asymptotic pseudo trajectory for $\phi$ if 
\begin{align}
    \lim_{t \to \infty} \sup_{s \in [0, T]} d(X(t + s), \phi_s(X(t))) = 0,
\end{align}
for every $T >0$.
\end{definition}
For each $i \in \{1,2\}$, let we define the logit function $g_i: \mathbb{R}^{|A_i|} \mapsto \Delta(A_i)$ as 
\begin{align}
    g_i(z_i) = \left(\frac{\exp\left(z_i(a_i)\right)}{\sum_{a_i'\in A_i}\exp\left(z_i(a_i')\right)}\right)_{a_i \in A_i}.
\end{align}
We write  $g_i(z_i)(a_i) \in \mathbb{R}$ be the $a_i$-th element of $g_i(z_i)$ and $ \nabla g_i(z_i)(a_i) \in \mathbb{R}^{|A_i|}$ be the gradient vector of $ g_i(z_i)(a_i)$, respectively.  
As a first result, we prove that the dynamics of the strategy $\{\pi_i^t\}$ updated by M2WU is an asymptotic pseudo trajectory of a continuous dynamics. 
\begin{lemma}\label{lem:asymp_pseudo_trajec}
Suppose that the sequence $\{\eta_t\}_{t \ge 1}$ satisfy $ \eta_t \propto t^{-\lambda}$ for some $\lambda \in (1/\kappa, 1]$, where $\kappa$ is a constant defined in Assumption~\ref{asm:noise_semibandit} (ii).
Then, for each $i \in \{1, 2\}$, the sequence of strategies $\{\pi^t_i\}_{t \ge 1}$ updated by M2WU is an asymptotic pseudo trajectory for the replicator mutator dynamics:
\begin{align}
\begin{aligned}
    \frac{d}{dt}\pi_i^t(a_i) =& \pi_i^t(a_i)\left(q_i^{\pi^t}(a_i) -  v_i^{\pi^t}\right) + \mu\left(r_i(a_i)-\pi_i^t(a_i)\right).
\end{aligned}  \tag{RMD} \label{eq:rmd_app}
\end{align}
\end{lemma}
We present the proof of Lemma~\ref{lem:asymp_pseudo_trajec} in Appendix~\ref{subsec:asymp_pseudo_trajec}.
Furthermore, we have the following exponential convergence result to the stationary point $\pi^{\mu,r}_i$ in the noiseless continuous time setting in \citet{abe2022mutationdriven}. 
\begin{theorem}[Theorem 5.2. of \citet{abe2022mutationdriven}]
\label{thm:bregman_div}
Let $\pi_i^{\mu, r}\in \Delta(A_i)$ be a stationary point of \eqref{eq:rmd} for all $i\in \{1, 2\}$.
Then, for all $\mu >0$, the continuous-time dynamics $\pi^t$ updated by \eqref{eq:rmd} satisfies the following:
\begin{align*}
    \frac{d}{dt}\mathrm{KL}(\pi^{\mu,r}, \pi^t)  &= - \mu\sum_{i=1}^2\sum_{a_i\in A_i}r_i(a_i)\left(\sqrt{\frac{\pi_i^t(a_i)}{\pi_i^{\mu, r}(a_i)}}-\sqrt{\frac{\pi_i^{\mu,r}(a_i)}{\pi_i^t(a_i)}}\right)^2.
\end{align*}
Furthermore, $\pi^t$ satisfies that:
\begin{align*}
    \frac{d}{dt}\mathrm{KL}(\pi^{\mu,r}, \pi^t) \leq -\mu \xi\mathrm{KL}(\pi^{\mu,r}, \pi^t),
\end{align*}
where $\xi=\min_{i\in \{1,2\}, a_i\in A_i}\frac{c_i(a_i)}{\pi_i^{\mu,r}(a_i)}$.
\end{theorem}
It can be observed that the function $\mathrm{KL}(\pi^{\mu,r}, \cdot)$ is a strict Lyapunov function. 
Furthermore, the stationary point of \eqref{eq:rmd_app} is unique \citep{abe2022mutationdriven}. 
Therefore, from Corollary 6.6 of \cite{benaim1999dynamics}, $\pi^t_i$ updated by M2WU converges to $\pi_i^{\mu,r}$, almost surely. \qed

\subsection{Proof of Lemma~\ref{lem:asymp_pseudo_trajec}}\label{subsec:asymp_pseudo_trajec}
For each $i \in \{1,2\}$, for any $a_i, a_i', \tilde{a}_i \in A_i $,
\begin{align*}
    \frac{\partial}{\partial z_i(a_i')}( g_i(z_i)(a_i)) = g_i(z_i)(a_i) ( \indicator_{\{a_i = a_i'\}} - g_i(z_i)(a_i'))
\end{align*}
and
\begin{align}
     &\frac{\partial^2}{\partial z_i(\tilde{a}_i) \partial z_i(a_i')}( g_i(z_i)(a_i))  \nonumber
     \\
     &\ \ \ \ \ = g_i(z_i)(a_i) \left( \indicator_{\{a_i = a_i' = \tilde{a}_i\}} - g_i(z_i)(\tilde{a}_i) \indicator_{\{a_i = a_i'\}} - g_i(z_i)(a_i') \left(\indicator_{\{a_i = \tilde{a}_i\}} + \indicator_{\{a_i' = \tilde{a}_i\}} - 2 g_i(z_i)(\tilde{a}_i)\right)\right). \label{eq:compute_hessian}
\end{align}
We write
\begin{align*}
    \hat{q}_i^{\mu, t} & = \left(\hat{q}_i^{\pi^t}(a_i)  +  \frac{\mu}{\pi_i^t(a_i)}(r_i(a_i) - \pi_i^t(a_i)) \right)_{a_i \in A_i}
    \\
    {q}_i^{\mu, t} &= \left({q}_i^{\pi^t}(a_i)  +  \frac{\mu}{\pi_i^t(a_i)}(r_i(a_i) - \pi_i^t(a_i)) \right)_{a_i \in A_i}
    \\
    \hat{\phi}_t& =  \frac{\eta_t}{2}(\hat{q}_i^{\mu, t})^\top \Hess(g_i(\zeta)(a_i))\hat{q}_i^{\mu, t}.
\end{align*}
By Taylor's theorem, we get the following computations. 
\begin{align*}
    \pi^{t+1}_i(a_i) & = g_i(z_i^{t+1})(a_i)
    \\
    & = g_i(z_i^{t} + \eta_t \hat{q}_i^{\mu, t})(a_i)
    \\
    & = g_i(z_i^{t})(a_i)  + \eta_t \left( (\nabla g_i(z_i^t)(a_i))^\top\hat{q}_i^{\mu, t} + \frac{\eta_t}{2}(\hat{q}_i^{\mu, t})^\top \Hess(g_i(\zeta)(a_i))\hat{q}_i^{\mu, t} \right)
    \\
    &  = g_i(z_i^{t})(a_i)  + \eta_t \left( (\nabla g_i(z_i^t)(a_i))^\top({q}_i^{\mu, t} + \hat{q}_i^{\mu, t} - {q}_i^{\mu, t}) + \hat{\phi}_t \right),
\end{align*}
where $\zeta$ is a point between $z_i^t$ and $z_i^{t+1}$. 
\begin{align*}
     (\nabla g_i(z_i^t)(a_i))^\top {q}_i^{\mu, t} &  = \sum_{a_i' \in A_i}g_i(z_i^t)(a_i) ( \indicator_{\{a_i = a_i'\}} - g_i(z_i^t)(a_i')) \left({q}_i^{\pi^t}(a_i')  +  \frac{\mu}{\pi_i^t(a_i')}(r_i(a_i') - \pi_i^t(a_i')) \right)
    \\
    & = g_i(z_i^t)(a_i) \left(  {q}_i^{\mu, t}(a_i)  - \sum_{a_i' \in A_i} g_i(z_i^t)(a_i') {q}_i^{\mu, t}(a_i') \right)
    \\
    & = \pi^{t}_i(a_i) \left(  {q}_i^{\mu, t}(a_i)  - \sum_{a_i' \in A_i} \pi^{t}_i(a_i') {q}_i^{\mu, t}(a_i') \right).
\end{align*}
We can write the dynamics of $\pi^{t}_i(a_i)$ as follows
\begin{align*}
     \pi^{t+1}_i(a_i) & =  \pi^{t}_i(a_i) + \eta_t (F(\pi_i^t) + U_t + \hat{\phi}_t), 
\end{align*}
where $F(\pi_i^t ) = \pi^{t}_i(a_i) \left(  {q}_i^{\mu, t}(a_i)  - \sum_{a_i' \in A_i} \pi^{t}_i(a_i') {q}_i^{\mu, t}(a_i') \right) $ is a continuous function and $U_t =  (\nabla g_i(z_i^t)(a_i))^\top(\hat{q}_i^{\mu, t} - {q}_i^{\mu, t})$ is a martingale difference sequence from Assumption~\ref{asm:noise_semibandit} and the bound on the utility function.
Note that from the assumption on $(\eta_t)_{t \ge 1}$, $\lim_{t \to \infty} \eta_t = 0$ and $\sum_{t=1}^\infty \eta_t = \infty$. 
From the form of \eqref{eq:compute_hessian}, the elements of $\Hess(g_i(\zeta)(a_i))$ is bounded by some constant. Thus, the limit of $\hat{\phi}_t$ is determined by the term $\|\hat{q}_i^{\mu, t} \|_2^2$.
Let $\mathcal{E}_t$ be an event such that $ \|\hat{q}_i^{\mu, t} \|_2^2 \ge t^p$ with $ p < 1/\kappa$, where $\kappa$ is in Assumption~\ref{asm:noise_semibandit} (ii) and $p$ is a value satisfies $ \eta_t = o(t^{-p})$ (note that $\eta_t \propto t^{-\kappa}$ and $1/\kappa < \lambda \le 1$). Using Assumption~\ref{asm:noise_semibandit}, the assumption that $\pi_i^t(a) > D>0$, and the boundedness of the utility function,
\begin{align*}
    \sum_{t=1}^\infty \mathbb{P}(\mathcal{E}_t) =  \sum_{t=1}^\infty \mathbb{P}( \|\hat{q}_i^{\mu, t} \|_2^2 \ge t^p \; \mid \;\mathcal{F}_{t-1})  = \sum_{t=1}^\infty \mathcal{O}\left(t^{-\kappa p}\right) < \infty.
\end{align*}
From the Borel–Cantelli lemma, $ \mathbb{P}(\cap_{t=1}^\infty \cup_{s \ge t}^\infty \mathcal{E}_s ) = 0$. 
Therefore, the event $\mathcal{E}_t$ occurs only for a finite number of $t$, almost surely. 
Thus, as $\eta_t t^p \propto t^{-\lambda + p} = o(1)$, 
\begin{align*}
    \hat{\phi}_t = \mathcal{O} (\eta_t \|\hat{q}_i^{\mu, t} \|_2^2) = \mathcal{O} (\eta_t t^p) = o(1),
\end{align*}
almost surely. 
Therefore, from Definition~\ref{def:aprx_rm_alg}, the update of $\{\pi^{t}_i\}_{t \ge 1}$ is an approximate Robbins-Monro algorithm. From Proposition~4.2 of \citet{benaim1999dynamics}, $\{\pi^{t}_i\}_{t \ge 1}$ is an asymptotic pseudo-trajectory of the replicator mutator dynamics.
\qed

\section{PROOFS FOR THEOREM \ref{thm:direct_convergence}}
\label{sec:appendix_proof_direct_convergence}
\subsection{Proof of Theorem \ref{thm:direct_convergence}}
\begin{proof}[Proof of Theorem \ref{thm:direct_convergence}]
From Lemma \ref{lem:min_kl_diff}, sequence $\{\min_{\pi^{\ast}\in \Pi^{\ast}}\mathrm{KL}(\pi^{\ast}, r^k)\}_{k \ge 0}$ is a monotonically decreasing sequence and is bounded from below by zero. Hence, $\{\min_{\pi^{\ast}\in \Pi^{\ast}}\mathrm{KL}(\pi^{\ast}, r^k)\}_{k \ge 0}$ converges to some $b\geq 0$.
We show that $b=0$ by a contradiction argument.



Suppose $b>0$ and let us define $B=\min_{\pi^{\ast}\in \Pi^{\ast}}\mathrm{KL}(\pi^{\ast}, r^0)$.
Since $\min_{\pi^{\ast}\in \Pi^{\ast}}\mathrm{KL}(\pi^{\ast}, r^k)$ monotonically decreases, $r^k$ is in the set $\Omega_{b, B}=\{r\in \prod_{i=1}^2\Delta^{\circ}(A_i) ~|~ b\leq \min_{\pi^{\ast}\in \Pi^{\ast}}\mathrm{KL}(\pi^{\ast}, r) \leq B\}$ for all $k\geq 0$.
Since $\min_{\pi^{\ast}\in \Pi^{\ast}}\mathrm{KL}(\pi^{\ast}, r)$ is a continuous function on $\prod_{i=1}^2\Delta^{\circ}(A_i)$, the preimage $\Omega_{b, B}$ of the closed set $[b,B]$ is also closed.
Furthermore, since $\prod_{i=1}^2\Delta^{\circ}(A_i)$ is a bounded set, $\Omega_{b,B}$ is a bounded set.
Thus, $\Omega_{b, B}$ is a compact set.

From Lemma \ref{lem:rmd_continuity}, $\min_{\pi^{\ast}\in \Pi^{\ast}}\mathrm{KL}(\pi^{\ast}, F(r)) - \min_{\pi^{\ast}\in \Pi^{\ast}}\mathrm{KL}(\pi^{\ast}, r)$ is also a continuous function.
Since a continuous function has a maximum over a compact set, the maximum $\delta = \max_{r\in \Omega_{b, B}} \left\{ \min_{\pi^{\ast}\in \Pi^{\ast}}\mathrm{KL}(\pi^{\ast}, F(r)) - \min_{\pi^{\ast}\in \Pi^{\ast}}\mathrm{KL}(\pi^{\ast}, r)\right\}$ exists.
From Lemma \ref{lem:min_kl_diff} and the assumption that $b>0$, we have $\delta < 0$.
It follows that:
\begin{align*}
    \min_{\pi^{\ast}\in \Pi^{\ast}}\mathrm{KL}(\pi^{\ast}, r^k) &= \min_{\pi^{\ast}\in \Pi^{\ast}}\mathrm{KL}(\pi^{\ast}, r^0) + \sum_{l=0}^{k-1} \left(\min_{\pi^{\ast}\in \Pi^{\ast}}\mathrm{KL}(\pi^{\ast}, r^{l+1}) - \min_{\pi^{\ast}\in \Pi^{\ast}}\mathrm{KL}(\pi^{\ast}, r^l)\right) \\
    &\leq B + \sum_{l=0}^{k-1} \delta = B + k\delta.
\end{align*}
This implies that $\min_{\pi^{\ast}\in \Pi^{\ast}}\mathrm{KL}(\pi^{\ast}, r^k) < 0$ for $k > \frac{B}{-\delta}$, which is a contradiction because $\min_{\pi^{\ast}\in \Pi^{\ast}}\mathrm{KL}(\pi^{\ast}, r) \geq 0$.
Therefore, the sequence of $\min_{\pi^{\ast}\in \Pi^{\ast}}\mathrm{KL}(\pi^{\ast}, r^k)$ converges to $0$, and $r^k$ converges to some strategy profile in $\Pi^{\ast}$.
\end{proof}

\subsection{Proof of Lemma \ref{lem:min_kl_diff}}
\begin{proof}[Proof of Lemma \ref{lem:min_kl_diff}]
First, we prove the first statement of the lemma using following two lemmas:
\begin{lemma}
\label{lem:kl_diff}
Let $\pi^{\mu,r}$ be a stationary point of (\ref{eq:rmd}) with the reference strategy profile $r$.
Assuming that $r\neq \pi^{\mu,r}$, for any Nash equilibrium $\pi^{\ast}$ of the original game, we have:
\begin{align*}
    \mathrm{KL}(\pi^{\ast}, \pi^{\mu,r}) - \mathrm{KL}(\pi^{\ast}, r) &< 0.
\end{align*}
\end{lemma}
\begin{lemma}
\label{lem:recursive_point}
Let $\pi^{\mu,r}$ be a stationary point of (\ref{eq:rmd}) with the reference strategy profile $r$.
If $r = \pi^{\mu,r}$, then $r$ is a Nash equilibrium of the original game.
\end{lemma}

From Lemma \ref{lem:recursive_point}, when $r\in \prod_{i=1}^2\Delta^{\circ}(A_i)\setminus\Pi^{\ast}$, $r\neq \pi^{\mu,r}$ always holds.
Let us define $\pi^{\star} = \argmin_{\pi^{\ast}\in \Pi^{\ast}}\mathrm{KL}(\pi^{\ast}, r)$.
From Lemma \ref{lem:kl_diff}, if $r\neq \pi^{\mu,r}$ we have:
\begin{align*}
    \min_{\pi^{\ast}\in \Pi^{\ast}}\mathrm{KL}(\pi^{\ast}, r) = \mathrm{KL}(\pi^{\star}, r) > \mathrm{KL}(\pi^{\star}, \pi^{\mu,r}) \geq \min_{\pi^{\ast}\in \Pi^{\ast}}\mathrm{KL}(\pi^{\ast}, \pi^{\mu,r}).
\end{align*}
Therefore, if $r\in \prod_{i=1}^2\Delta^{\circ}(A_i)\setminus\Pi^{\ast}$ then $\min_{\pi^{\ast}\in \Pi^{\ast}}\mathrm{KL}(\pi^{\ast}, \pi^{\mu,r}) < \min_{\pi^{\ast}\in \Pi^{\ast}}\mathrm{KL}(\pi^{\ast}, r)$.

Next, we prove the second statement of the lemma.
Assume that $r\in \Pi^{\ast}$ implies that $\pi^{\mu,r}\neq r$.
In this case, we can apply Lemma \ref{lem:kl_diff}, so we have for all $\pi^{\ast}\in \Pi^{\ast}$, $\mathrm{KL}(\pi^{\ast}, \pi^{\mu,r})  < \mathrm{KL}(\pi^{\ast}, r)$.
On the other hand, since $r \in \Pi^{\ast}$, there exists a Nash equilibrium $\pi^{\star}$ such that $\mathrm{KL}(\pi^{\star}, r)=0$.
Thus, we have $\mathrm{KL}(\pi^{\star}, \pi^{\mu,r})  < \mathrm{KL}(\pi^{\star}, r)=0$, which contradicts that $\mathrm{KL}(\pi^{\star}, \pi^{\mu,r})\geq 0$.
Therefore, if $r\in \Pi^{\ast}$ then $\pi^{\mu,r}= r$.
\end{proof}

\subsection{Proof of Lemma \ref{lem:rmd_continuity}}
\begin{proof}[Proof of Lemma \ref{lem:rmd_continuity}]
For a given $r\in \prod_{i=1}^2\Delta^{\circ}(A_i)$, let us consider that $\pi^t$ follows the following (\ref{eq:rmd}) dynamics with a reference strategy profile $r\in \prod_{i=1}^2\Delta^{\circ}(A_i)$:
\begin{align*}
    \frac{d}{dt}\pi_i^t(a_i) =& \pi_i^t(a_i)\left(q_i^{\pi^t}(a_i) - v_i^{\pi^t}\right) + \mu\left(r_i(a_i)-\pi_i^t(a_i)\right).
\end{align*}
Therefore, we have for a given $r'\in \prod_{i=1}^2 \Delta^{\circ}(A_i)$ and the associated stationary point $\pi^{\mu,r'}$:
\begin{align}
\label{eq:time_derivative_of_divergence_for_continuity}
    \frac{d}{dt}\mathrm{KL}&(\pi^{\mu,r'}, \pi^t) \nonumber\\
    =& \sum_{i=1}^2\sum_{a\in A_i}\pi_i^{\mu,r'}(a) \frac{d}{dt}\ln \left(\frac{\pi_i^{\mu,r'}(a)}{\pi_i^t(a)}\right) \nonumber\\
    =& -\sum_{i=1}^2\sum_{a\in A_i}\pi_i^{\mu,r'}(a) \frac{d}{dt}\ln \pi_i^t(a) \nonumber\\
    =& -\sum_{i=1}^2\sum_{a\in A_i}\frac{\pi_i^{\mu,r'}(a)}{\pi_i^t(a)} \frac{d}{dt} \pi_i^t(a) \nonumber\\
    =& -\sum_{i=1}^2\sum_{a\in A_i}\frac{\pi_i^{\mu,r'}(a)}{\pi_i^t(a)} \left(\pi_i^t(a_i)\left(q_i^{\pi^t}(a_i) - v_i^{\pi^t}\right) + \mu\left(r_i(a_i)-\pi_i^t(a_i)\right)\right) \nonumber\\
    =& \sum_{i=1}^2\sum_{a\in A_i} \left(\pi_i^{\mu,r'}(a)\left(v_i^{\pi^t} - q_i^{\pi^t}(a_i)\right) - \frac{\mu}{\pi_i^t(a)}\left(r_i(a_i)-\pi_i^t(a_i)\right)\pi_i^{\mu,r'}(a)\right) \nonumber\\
    =& \sum_{i=1}^2\sum_{a\in A_i} \left(\left(\pi_i^t(a) - \pi_i^{\mu,r'}(a)\right)q_i^{\pi^t}(a) - \frac{\mu}{\pi_i^t(a)}\left(r_i(a_i)-\pi_i^t(a_i)\right)\pi_i^{\mu,r'}(a) + \frac{\mu}{\pi_i^t(a)}\left(r_i(a_i)-\pi_i^t(a_i)\right)\pi_i^t(a)\right) \nonumber\\
    =& \sum_{i=1}^2\sum_{a\in A_i}\left(q_i^{\pi^t}(a) + \frac{\mu}{\pi_i^t(a)}\left(r_i(a)-\pi_i^t(a)\right)\right) \left(\pi_i^t(a) - \pi_i^{\mu,r'}(a)\right) \nonumber\\
    =& \sum_{i=1}^2\sum_{a\in A_i}\left(q_i^{\pi^t}(a) + \frac{\mu}{\pi_i^t(a)}\left(r'_i(a)-\pi_i^t(a)\right)\right) \left(\pi_i^t(a) - \pi_i^{\mu,r'}(a)\right) \nonumber\\
    &+ \sum_{i=1}^2\sum_{a\in A_i}\left(\frac{\mu}{\pi_i^t(a)}\left(r_i(a)-\pi_i^t(a)\right)- \frac{\mu}{\pi_i^t(a)}\left(r_i'(a)-\pi_i^t(a)\right)\right) \left(\pi_i^t(a) - \pi_i^{\mu,r'}(a)\right),
\end{align}
where the sixth equality follows from $v_i^{\pi^t}=\sum_{a\in A_i}\pi_i^t(a)q_i^{\pi^t}(a)$ and $\sum_{a\in A_i}\frac{\mu}{\pi_i^t(a)}\left(r_i(a_i)-\pi_i^t(a_i)\right)\pi_i^t(a)=\mu\sum_{a\in A_i}\left(r_i(a_i)-\pi_i^t(a_i)\right)=0$.

The first term of (\ref{eq:time_derivative_of_divergence_for_continuity}) can be written as:
\begin{align*}
    & \sum_{i=1}^2\sum_{a\in A_i}\left(q_i^{\pi^t}(a) + \frac{\mu}{\pi_i^t(a)}\left(r'_i(a) - \pi_i^t(a)\right)\right)\left(\pi_i^t(a) - \pi_i^{\mu,r'}(a)\right) \\
    &= \sum_{i=1}^2\sum_{a\in A_i}\left(\pi_i^t(a) - \pi_i^{\mu,r'}(a)\right)\left(q_i^{\pi^t}(a) + \mu\left(\frac{r'_i(a)}{\pi_i^t(a)}-1\right)\right) \\
    &= \sum_{i=1}^2\sum_{a\in A_i}\left(\pi_i^t(a) - \pi_i^{\mu,r'}(a)\right)\left(q_i^{\pi^t}(a) + \mu\frac{r'_i(a)}{\pi_i^t(a)}\right) \\
    &= \sum_{i=1}^2\left(v_i^{\pi^t} - v_i^{\pi_i^{\mu,r'}, \pi_{-i}^t} + \mu - \mu\sum_{a\in A_i}r'_i(a)\frac{\pi_i^{\mu,r'}(a)}{\pi_i^t(a)}\right) \\
    &= \sum_{i=1}^2\left(- v_i^{\pi_i^{\mu,r'}, \pi_{-i}^t} + \mu - \mu\sum_{a\in A_i}r'_i(a)\frac{\pi_i^{\mu,r'}(a)}{\pi_i^t(a)}\right) \\
    &= \sum_{i=1}^2\left(v_i^{\pi_i^t, \pi_{-i}^{\mu,r'}} + \mu - \mu\sum_{a\in A_i}r'_i(a)\frac{\pi_i^{\mu,r'}(a)}{\pi_i^t(a)}\right),
\end{align*}
where the fourth equality follows from $\sum_{i=1}^2v_i^{\pi^t}=0$, and the last equality follows from $-v_1^{\pi_1^{\mu,r'},\pi_2^t}=v_2^{\pi_1^{\mu,r'},\pi_2^t}$ and $-v_2^{\pi_1^t,\pi_2^{\mu,r'}}=v_1^{\pi_1^t,\pi_2^{\mu,r'}}$ by the definition of two-player zero-sum games.
Here, from Lemma \ref{lem:rmd_property}, for all $i\in \{1,2\}$:
\begin{align*}
    v_i^{\pi_i^t,\pi_{-i}^{\mu,r'}} = v_i^{\pi^{\mu,r'}} + \mu - \mu\sum_{a\in A_i} r'_i(a)\frac{\pi_i^t(a)}{\pi_i^{\mu,r'}(a)},
\end{align*}
and then we have:
\begin{align}
\label{eq:term1_bound_for_continuity}
    &\sum_{i=1}^2\sum_{a\in A_i}\left(q_i^{\pi^t}(a) + \frac{\mu}{\pi_i^t(a)}\left(r'_i(a) - \pi_i^t(a)\right)\right)\left(\pi_i^t(a) - \pi_i^{\mu,r'}(a)\right) \nonumber\\
    &= \sum_{i=1}^2v_i^{\pi^{\mu,r'}} + 4\mu - \mu\sum_{i=1}^2\sum_{a\in A_i}r'_i(a)\left(\frac{\pi_i^t(a)}{\pi_i^{\mu,r'}(a)} + \frac{\pi_i^{\mu,r'}(a)}{\pi_i^t(a)}\right) \nonumber\\
    &= 4\mu - \mu\sum_{i=1}^2\sum_{a\in A_i}r'_i(a)\left(\frac{\pi_i^t(a)}{\pi_i^{\mu,r'}(a)} + \frac{\pi_i^{\mu,r'}(a)}{\pi_i^t(a)}\right) \nonumber\\
    &= -\mu \sum_{i=1}^2\sum_{a\in A_i}r'_i(a)\left(\sqrt{\frac{\pi_i^t(a)}{\pi_i^{\mu,r'}(a)}} - \sqrt{\frac{\pi_i^{\mu,r'}(a)}{\pi_i^t(a)}}\right)^2,
\end{align}
where the second equality follows from $\sum_{i=1}^2 v_i^{\pi^{\mu,r'}} = 0$ by the definition of zero-sum games.

On the other hand, the second term of (\ref{eq:time_derivative_of_divergence_for_continuity}) is written as:
\begin{align}
\label{eq:term2_bound_for_continuity}
    &\sum_{i=1}^2\sum_{a\in A_i}\left(\frac{\mu}{\pi_i^t(a)}\left(r_i(a)-\pi_i^t(a)\right)- \frac{\mu}{\pi_i^t(a)}\left(r_i'(a)-\pi_i^t(a)\right)\right) \left(\pi_i^t(a) - \pi_i^{\mu,r'}(a)\right) \nonumber\\
    &= \mu\sum_{i=1}^2\sum_{a\in A_i}\frac{1}{\pi_i^t(a)}\left(r_i(a)-r_i'(a)\right) \left(\pi_i^t(a) - \pi_i^{\mu,r'}(a)\right) \leq \mu\sum_{i=1}^2\sum_{a\in A_i}\frac{1}{\pi_i^t(a)}\left|r_i(a)-r_i'(a)\right|.
\end{align}

Combining (\ref{eq:time_derivative_of_divergence_for_continuity}), (\ref{eq:term1_bound_for_continuity}), and (\ref{eq:term2_bound_for_continuity}), we can obtain:
\begin{align*}
    \frac{d}{dt}\mathrm{KL}(\pi^{\mu,r'}, \pi^t) \leq - \mu\sum_{i=1}^2\sum_{a\in A_i} r'_i(a)\left(\sqrt{\frac{\pi_i^t(a)}{\pi_i^{\mu,r'}(a)}} - \sqrt{\frac{\pi_i^{\mu,r'}(a)}{\pi_i^t(a)}}\right)^2 + \mu\sum_{i=1}^2\sum_{a\in A_i}\frac{1}{\pi_i^t(a)}\left|r_i(a)-r_i'(a)\right|.
\end{align*}
Setting the start point as $\pi^0 = \pi^{\mu,r}$, we have for all $t\geq 0$, $\pi^t = \pi^{\mu,r}$.
In this case, for all $t\geq 0$, we have $\frac{d}{dt}\mathrm{KL}(\pi^{\mu,r'}, \pi^t) = 0$.
Thus,
\begin{align*}
    \sum_{i=1}^2\sum_{a\in A_i} r'_i(a)\left(\sqrt{\frac{\pi_i^{\mu,r}(a)}{\pi_i^{\mu,r'}(a)}} - \sqrt{\frac{\pi_i^{\mu,r'}(a)}{\pi_i^{\mu,r}(a)}}\right)^2 \leq \sum_{i=1}^2\sum_{a\in A_i}\frac{1}{\pi_i^{\mu,r}(a)}\left|r_i(a)-r_i'(a)\right|.
\end{align*}
Here, since $r_i$ is in interior of $\Delta(A_i)$, there exists $\nu_1 > 0$ such that $\forall i, \forall a\in A_i$, $r_i(a)>\nu_1$.
Furthermore, from Lemma C.2 in \citet{abe2022mutationdriven}, $\pi_i^{\mu,r}$ is also in interior of $\Delta(A_i)$.
Thus, there exists $\nu_2 > 0$ such that $\forall i, \forall a\in A_i$, $\pi^{\mu,r}_i(a)>\nu_2$.
For a given $\varepsilon > 0$, let us define $\delta = \frac{\varepsilon^2\nu_1\nu_2}{4 + \varepsilon^2\nu_2}\sqrt{\frac{1}{\sum_{i=1}^2 |A_i|}}$.
If $\|r'-r\|_2<\delta$, then $\|r'-r\|_1\leq \|r'-r\|_2\sqrt{\sum_{i=1}^2 |A_i|}<\frac{\varepsilon^2\nu_1\nu_2}{4 + \varepsilon^2\nu_2}$.
Thus, $\forall i, \forall a\in A_i$, $r_i'(a) > \left(1 - \frac{\varepsilon^2\nu_2}{4 + \varepsilon^2\nu_2}\right)\nu_1>0$.
So, if $\|r'-r\|_2<\delta$, we have then:
\begin{align*}
    \sum_{i=1}^2 \sum_{a\in A_i} r'_i(a)\left(\sqrt{\frac{\pi_i^{\mu,r'}(a)}{\pi_i^{\mu,r}(a)}}-\sqrt{\frac{\pi_i^{\mu,r}(a)}{\pi_i^{\mu,r'}(a)}}\right)^2 &= \sum_{i=1}^2 \sum_{a\in A_i} r'_i(a)\left(\frac{\pi_i^{\mu,r'}(a)}{\pi_i^{\mu,r}(a)}+\frac{\pi_i^{\mu,r}(a_i)}{\pi_i^{\mu,r'}(a)}-2\right) \\
    &= \sum_{i=1}^2 \sum_{a\in A_i} \frac{r'_i(a)}{\pi_i^{\mu,r'}(a)}\frac{(\pi_i^{\mu,r}(a)-\pi_i^{\mu,r'}(a))^2}{\pi_i^{\mu,r}(a)} \\
    &\geq \left(1 - \frac{\varepsilon^2\nu_2}{4 + \varepsilon^2\nu_2}\right)\nu_1\sum_{i=1}^2 \sum_{a\in A_i} \frac{(\pi_i^{\mu,r}(a)-\pi_i^{\mu,r'}(a))^2}{\pi_i^{\mu,r}(a)} \\
    &\geq \left(1 - \frac{\varepsilon^2\nu_2}{4 + \varepsilon^2\nu_2}\right)\nu_1\sum_{i=1}^2  \ln\left(1+\sum_{a\in A_i}\frac{(\pi_i^{\mu,r}(a)-\pi_i^{\mu,r'}(a))^2}{\pi_i^{\mu,r}(a)}\right) \\
    &= \left(1 - \frac{\varepsilon^2\nu_2}{4 + \varepsilon^2\nu_2}\right)\nu_1\sum_{i=1}^2  \ln\left(\sum_{a\in A_i}\pi_i^{\mu,r'}(a)\frac{\pi_i^{\mu,r'}(a)}{\pi_i^{\mu,r}(a)}\right) \\
    &\geq \left(1 - \frac{\varepsilon^2\nu_2}{4 + \varepsilon^2\nu_2}\right)\nu_1\sum_{i=1}^2  \sum_{a\in A_i}\pi_i^{\mu,r'}(a)\ln\left(\frac{\pi_i^{\mu,r'}(a)}{\pi_i^{\mu,r}(a)}\right) \\
    &= \left(1 - \frac{\varepsilon^2\nu_2}{4 + \varepsilon^2\nu_2}\right)\nu_1\sum_{i=1}^2  \mathrm{KL}(\pi_i^{\mu,r'},\pi_i^{\mu,r}),
\end{align*}
where the second inequality follow from $x\geq \ln (1+x)$ for all $x>0$, and the third inequality follows from the concavity of the $\ln(\cdot)$ function and Jensen's inequality for concave functions.
Moreover,
\begin{align*}
    \sum_{i=1}^2  \mathrm{KL}(\pi_i^{\mu,r'},\pi_i^{\mu,r})&\geq \frac{1}{2}\sum_{i=1}^2\|\pi_i^{\mu,r'}-\pi_i^{\mu,r}\|_1^2 \geq \frac{1}{4}\|\pi^{\mu,r'}-\pi^{\mu,r}\|_1^2,
\end{align*}
where we use Pinsker's inequality \citep{Tsybakov:1315296}, and the fact that $\sum_{i=1}^2x_i^2 \geq \frac{1}{2}\left(\sum_{i=1}^2 x_i\right)^2$ for $x_i\in \mathbb{R}$.
Thus, we get:
\begin{align*}
    \left(1 - \frac{\varepsilon^2\nu_2}{4  + \varepsilon^2\nu_2}\right)\frac{\nu_1}{4}\|\pi^{\mu,r'}-\pi^{\mu,r}\|_1^2 \leq \sum_{i=1}^2\sum_{a\in A_i}\frac{1}{\pi_i^{\mu,r}(a)}\left|r_i(a)-r_i'(a)\right| < \frac{1}{\nu_2}\|r'-r\|_1.
\end{align*}
Therefore, if $\|r'-r\|<\delta$, we have then:
\begin{align*}
    \|\pi^{\mu,r'}-\pi^{\mu,r}\|_2 \leq \|\pi^{\mu,r'}-\pi^{\mu,r}\|_1 &< \sqrt{\frac{4}{\left(1 - \frac{\varepsilon^2\nu_2}{4  + \varepsilon^2\nu_2}\right)\nu_1\nu_2}\|r'-r\|_1} \\
    &\leq \sqrt{\frac{4}{\left(1 - \frac{\varepsilon^2\nu_2}{4  + \varepsilon^2\nu_2}\right)\nu_1\nu_2}\|r'-r\|_2\sqrt{\sum_{i=1}^2 |A_i|}} \\
    &< \sqrt{\frac{4}{\left(1 - \frac{\varepsilon^2\nu_2}{4  + \varepsilon^2\nu_2}\right)\nu_1\nu_2}\delta\sqrt{\sum_{i=1}^2 |A_i|}} \\
    &= \sqrt{\frac{4}{\left(1 - \frac{\varepsilon^2\nu_2}{4  + \varepsilon^2\nu_2}\right)\nu_1\nu_2}\frac{\varepsilon^2\nu_1\nu_2}{4  + \varepsilon^2\nu_2}} \\
    &= \sqrt{\frac{4}{1 - \frac{\varepsilon^2\nu_2}{4  + \varepsilon^2\nu_2}}\frac{\varepsilon^2}{4  + \varepsilon^2\nu_2}} = \varepsilon.
\end{align*}
Thus, for every $\varepsilon>0$, there exists $\delta > 0$ such that for all $r'\in \prod_{i=1}^2\Delta^{\circ}(A_i)$, if $\|r'-r\|_2<\delta$ then $\|\pi^{\mu,r'}-\pi^{\mu,r}\|_2<\varepsilon$.
Therefore, $F(\cdot)$ is a continuous function on $\prod_{i=1}^2\Delta^{\circ}(A_i)$.
\end{proof}

\subsection{Proof of Lemma \ref{lem:kl_diff}}
\begin{proof}[Proof of Lemma \ref{lem:kl_diff}]
First, we have:
\begin{align*}
    \mathrm{KL}(\pi^{\ast}, \pi^{\mu,r}) - \mathrm{KL}(\pi^{\ast}, r) &= \sum_{i=1}^2 \left(\sum_{a_i \in A_i}\pi_i^{\ast}(a_i)\ln \frac{\pi_i^{\ast}(a_i)}{\pi_i^{\mu,r}(a_i)} - \sum_{a_i \in A_i}\pi_i^{\ast}(a_i)\ln \frac{\pi_i^{\ast}(a_i)}{r_i(a_i)}\right) \\
    &= \sum_{i=1}^2 \sum_{a_i \in A_i}\pi_i^{\ast}(a_i)\ln \frac{r_i(a_i)}{\pi_i^{\mu,r}(a_i)} \\
    &\leq  2\ln \left(\frac{1}{2}\sum_{i=1}^2\sum_{a_i \in A_i}\pi_i^{\ast}(a_i)\frac{r_i(a_i)}{\pi_i^{\mu,r}(a_i)}\right),
\end{align*}
where the inequality follows from the concavity of the $\ln(\cdot)$ function and Jensen's inequality for concave functions.
Since $\ln(\cdot)$ is strictly concave, the equality holds if and only if $r=\pi^{\mu,r}$.
Therefore, from the assumption that $r\neq \pi^{\mu,r}$, we have:
\begin{align}
    \label{eq:diff_kl}
    \mathrm{KL}(\pi^{\ast}, \pi^{\mu,r}) - \mathrm{KL}(\pi^{\ast}, r) < 2\ln \left(\frac{1}{2}\sum_{i=1}^2 \sum_{a_i \in A_i}\pi_i^{\ast}(a_i)\frac{r_i(a_i)}{\pi_i^{\mu,r}(a_i)}\right).
\end{align}

Here, by using the ordinary differential equation (\ref{eq:rmd}), we have for all $i\in \{1,2\}$ and $a_i\in A_i$:
\begin{align*}
    &\pi_i^{\mu,r}(a_i)\left(q_i^{\pi^{\mu,r}}(a_i) - v_i^{\pi^{\mu,r}}\right) + \mu\left(r_i(a_i)-\pi_i^{\mu,r}(a_i)\right) = 0,
\end{align*}
and then:
\begin{align}
\label{eq:density_ratio}
    \frac{r_i(a_i)}{\pi_i^{\mu,r}(a_i)} = 1 - \frac{1}{\mu}\left(q_i^{\pi^{\mu,r}}(a_i) - v_i^{\pi^{\mu,r}}\right).
\end{align}
Combining (\ref{eq:diff_kl}) and (\ref{eq:density_ratio}), we have:
\begin{align*}
    \mathrm{KL}(\pi^{\ast}, \pi^{\mu,r}) - \mathrm{KL}(\pi^{\ast}, r) &< 2 \ln \left(\frac{1}{2}\sum_{i=1}^2\sum_{a_i \in A_i}\pi_i^{\ast}(a_i)\left(1 - \frac{1}{\mu}\left(q_i^{\pi^{\mu,r}}(a_i) -  v_i^{\pi^{\mu,r}}\right)\right)\right) \\
    &= 2\ln \left(\frac{1}{2}\sum_{i=1}^2 \left(1 - \frac{1}{\mu}\left(v_i^{\pi_i^{\ast}, \pi_{-i}^r} - v_i^{\pi^{\mu,r}}\right)\right)\right).
\end{align*}
Since $\sum_{a_i\in A_i}\pi_i^{\ast}(a_i)\frac{r_i(a_i)}{\pi_i^{\mu,r}(a_i)} > 0$, we have $1 - \frac{1}{\mu}\left(v_i^{\pi_i^{\ast}, \pi_{-i}^r} - v_i^{\pi^{\mu,r}}\right) > 0$.
Also, since $\pi^{\ast}$ is the Nash equilibrium, we get:
\begin{align*}
    \sum_{i=1}^2 \left(1 - \frac{1}{\mu}\left(v_i^{\pi_i^{\ast}, \pi_{-i}^r} - v_i^{\pi^{\mu,r}}\right)\right) &= 2 + \frac{1}{\mu}\sum_{i=1}^2 \left(v_i^{\pi^{\mu,r}} - v_i^{\pi_i^{\ast}, \pi_{-i}^r}\right) \\
    &= 2 + \frac{1}{\mu}\sum_{i=1}^2 \left(- v_i^{\pi_i^{\ast}, \pi_{-i}^r} - v_i^{\pi^{\ast}}\right) \\
    &= 2 + \frac{1}{\mu}\sum_{i=1}^2 \left(v_i^{\pi_i^{\mu,r}, \pi_{-i}^{\ast}} - v_i^{\pi^{\ast}}\right) \leq 2,
\end{align*}
where the second equality follows from $\sum_{i=1}^2v_i^{\pi^{\mu,r}}=0$ and $\sum_{i=1}^2v_i^{\pi^{\ast}}=0$, and the last equality follows from $-v_1^{\pi_1^{\ast},\pi_2^r}=v_2^{\pi_1^{\ast},\pi_2^r}$ and $-v_2^{\pi_1^r,\pi_2^{\ast}}=v_1^{\pi_1^r,\pi_2^{\ast}}$ by the definition of two-player zero-sum games.
Thus, we get
\begin{align*}
    \mathrm{KL}(\pi^{\ast}, \pi^{\mu,r}) - \mathrm{KL}(\pi^{\ast}, r) &< 2(\ln 1)\leq 0.
\end{align*}
\end{proof}

\subsection{Proof of Lemma \ref{lem:recursive_point}}
\begin{proof}[Proof of Lemma \ref{lem:recursive_point}]
By using the ordinary differential equation (\ref{eq:rmd}), we have for all $i\in \{1,2\}$ and $a_i\in A_i$:
\begin{align*}
    &\pi_i^{\mu,r}(a_i)\left(q_i^{\pi^{\mu,r}}(a_i) - v_i^{\pi^{\mu,r}}\right) + \mu\left(r_i(a_i)-\pi_i^{\mu,r}(a_i)\right) = 0.
\end{align*}
Since $r=\pi^{\mu,r}$, we have for all $i\in \{1,2\}$ and $a_i\in A_i$:
\begin{align*}
    r_i(a_i)\left(q_i^{r}(a_i) - v_i^{r}\right) = 0.
\end{align*}
From the definition of the reference strategy, we have $r_i(a_i)>0$, and then $v_i^r = \max_{a_i\in A_i}q_i^r(a_i)$ for all $i\in \{1,2\}$.
Therefore, each player $i$'s strategy $r_i$ is a best response to the other player $-i$'s strategy $r_{-i}$.
Thus, $r$ is a Nash equilibrium of the original game.
\end{proof}

\clearpage

\section{ADDITIONAL EXPERIMENTAL RESULTS WITH NOISE}
\label{sec:appendix_experiments}

\subsection{Various Learning Rates}
\label{sec:appendix_experiments_noisy_varying_eta}
Figures \ref{fig:exploitability_noisy_varying_eta_BRPS} and \ref{fig:exploitability_noisy_varying_eta_MNe} show the numerical results with varying learning rates $\eta\in \{0.1, 0.05, 0.01, 0.005, 0.001\}$ in BRPS and M-Ne with noisy-information feedback.
We observe that M2WU exhibits lower exploitability than MWU and OMWU for all $\eta$.
Also, we can see that OMWU does not converge for all learning rates in both games (See also Remark \ref{rm:comparison_to_omwu}).
\begin{figure}[ht!]
    \centering
    \begin{minipage}[t]{0.49\columnwidth}
        \centering
        \includegraphics[width=1.0\textwidth]{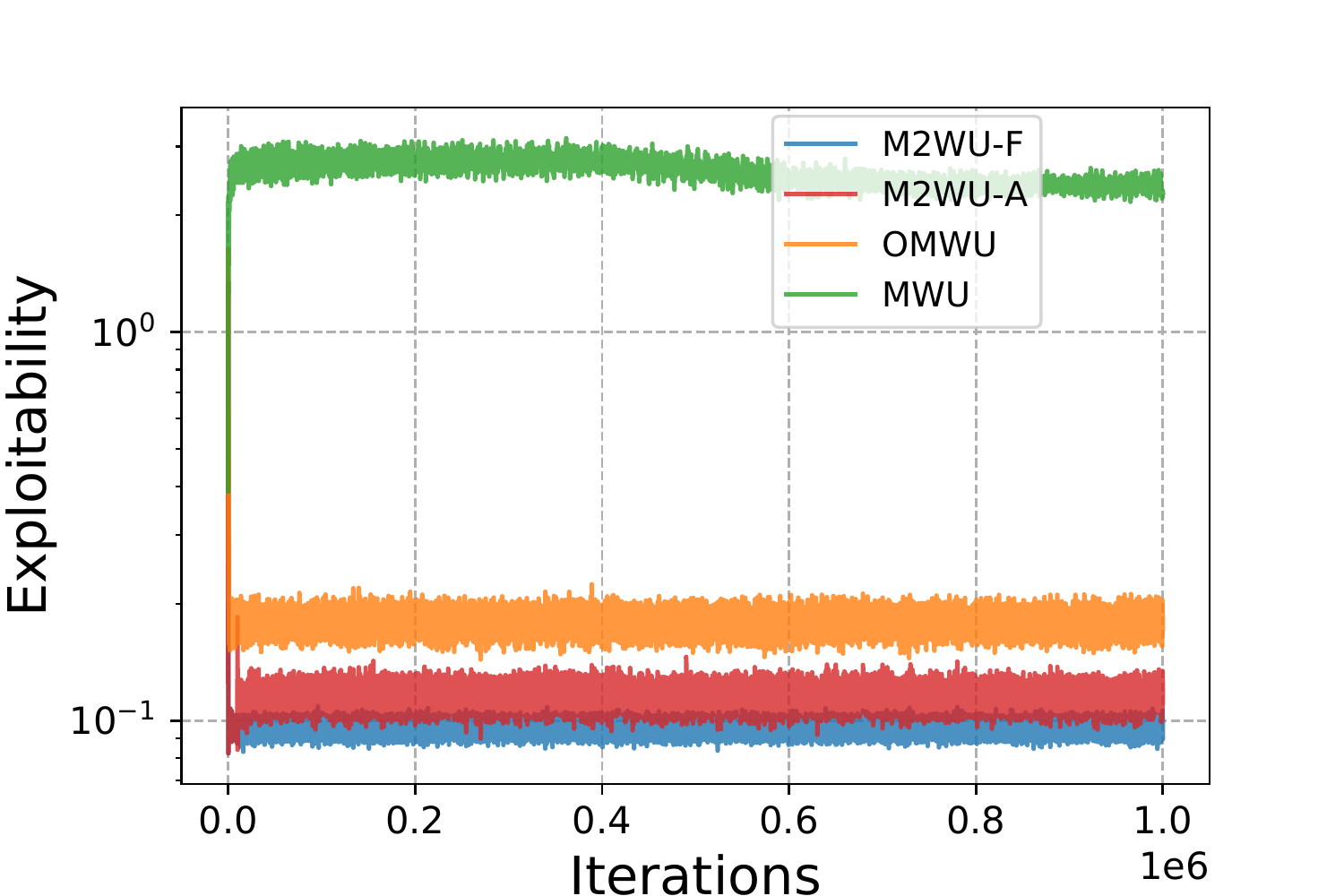}
        \subcaption{$\eta=0.1$}
    \end{minipage}
    \begin{minipage}[t]{0.49\columnwidth}
        \centering
        \includegraphics[width=1.0\textwidth]{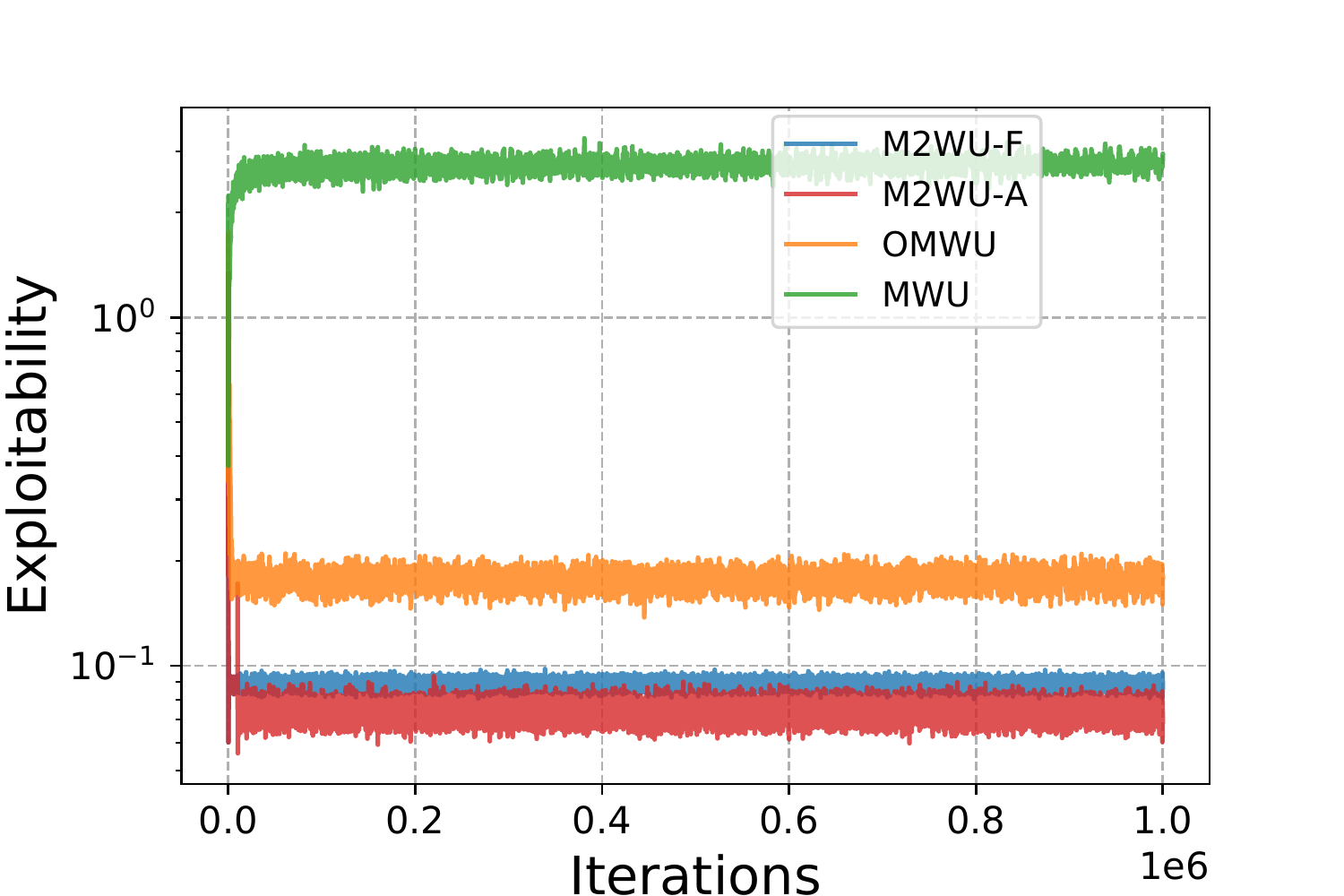}
        \subcaption{$\eta=0.05$}
    \end{minipage}
    \begin{minipage}[t]{0.49\columnwidth}
        \centering
        \includegraphics[width=1.0\textwidth]{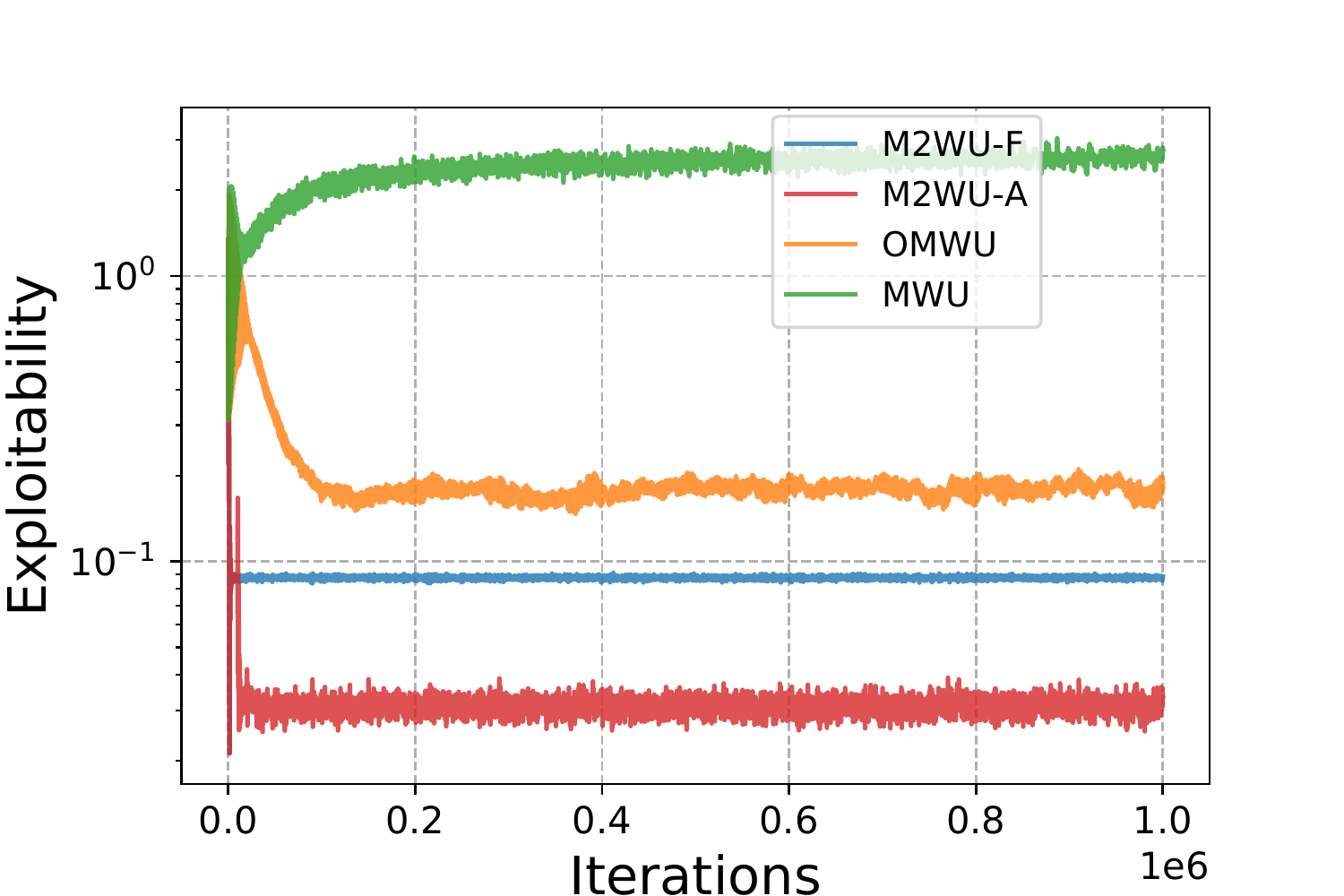}
        \subcaption{$\eta=0.01$}
    \end{minipage}
    \begin{minipage}[t]{0.49\columnwidth}
        \centering
        \includegraphics[width=1.0\textwidth]{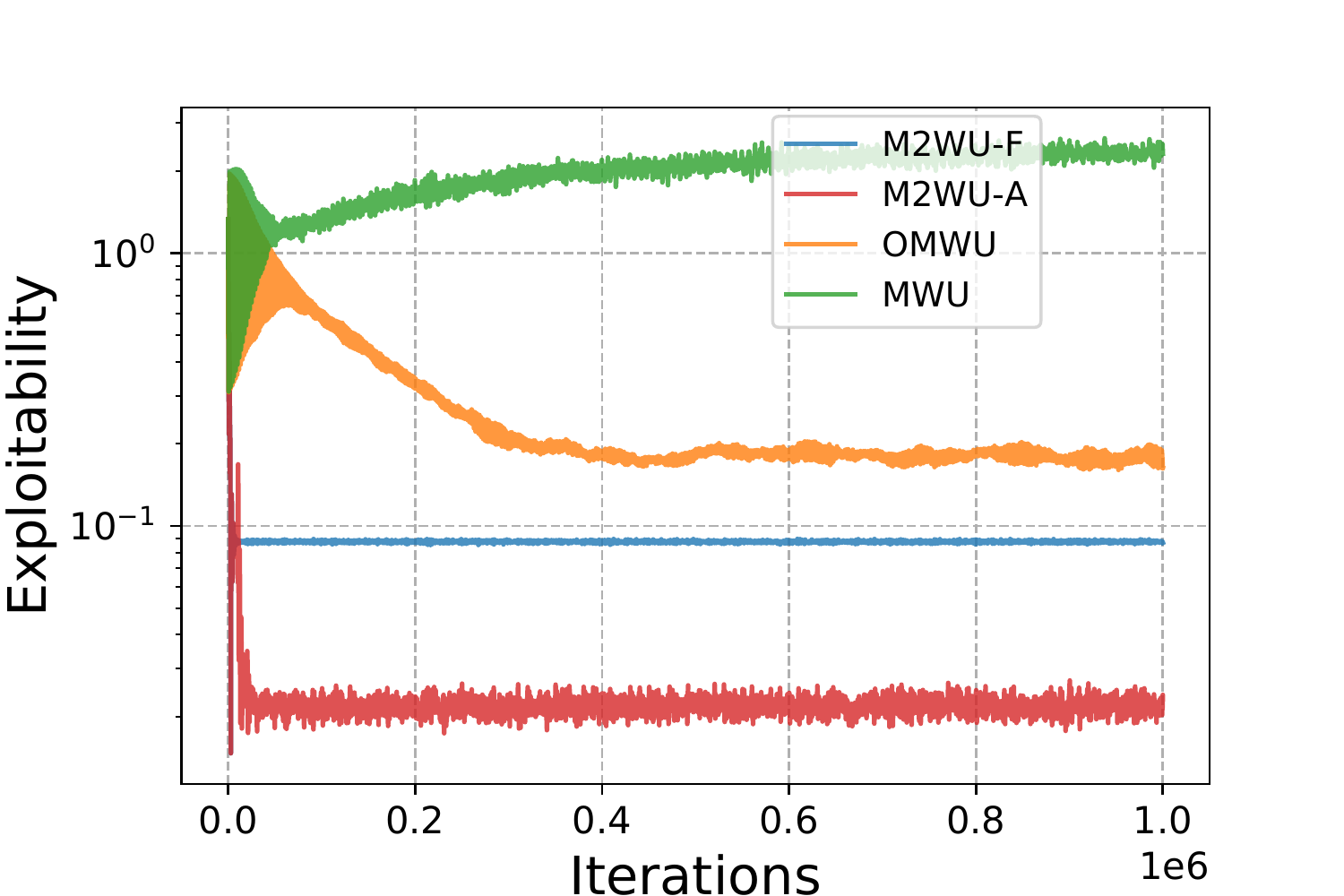}
        \subcaption{$\eta=0.005$}
    \end{minipage}
    \begin{minipage}[t]{0.49\columnwidth}
        \centering
        \includegraphics[width=1.0\textwidth]{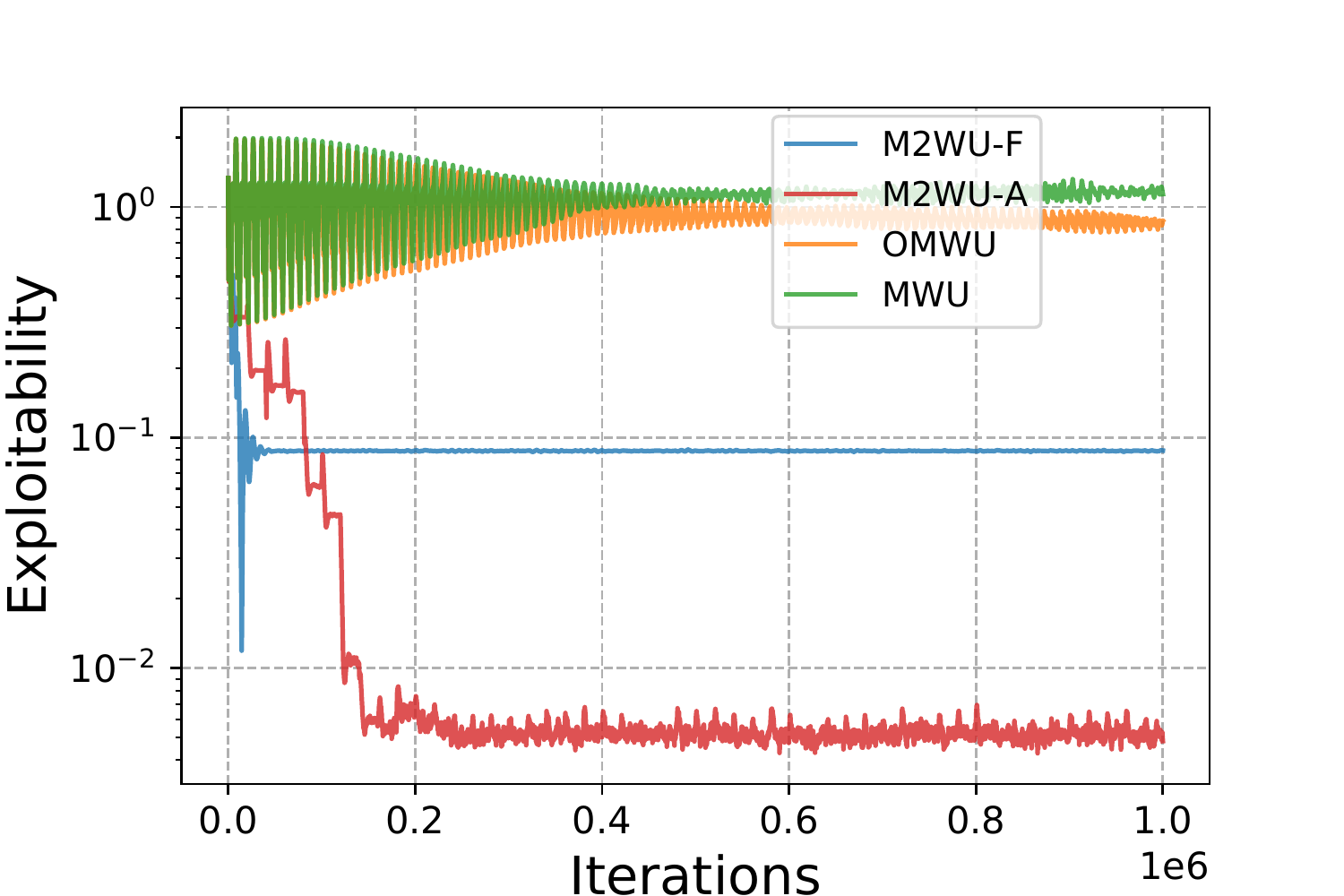}
        \subcaption{$\eta=0.001$}
    \end{minipage}
    \caption{
    Exploitability of $\pi^t$ for M2WU, MWU, and OMWU with varying $\eta\in \{0.1, 0.05, 0.01, 0.005, 0.001\}$ in BRPS with noisy-information feedback.
    }
    \label{fig:exploitability_noisy_varying_eta_BRPS}
\end{figure}

\begin{figure}[ht!]
    \centering
    \begin{minipage}[t]{0.49\columnwidth}
        \centering
        \includegraphics[width=1.0\textwidth]{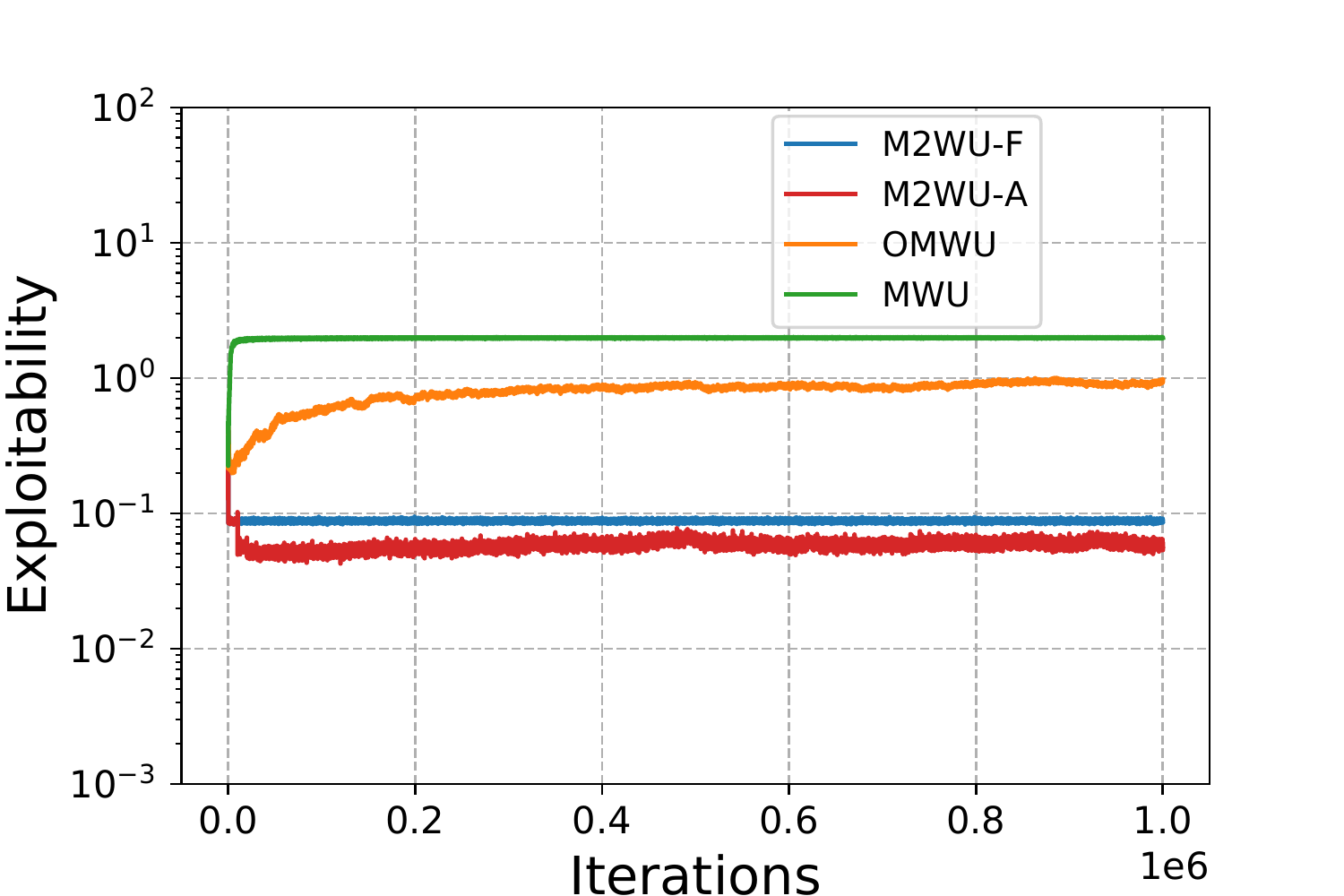}
        \subcaption{$\eta=0.1$}
    \end{minipage}
    \begin{minipage}[t]{0.49\columnwidth}
        \centering
        \includegraphics[width=1.0\textwidth]{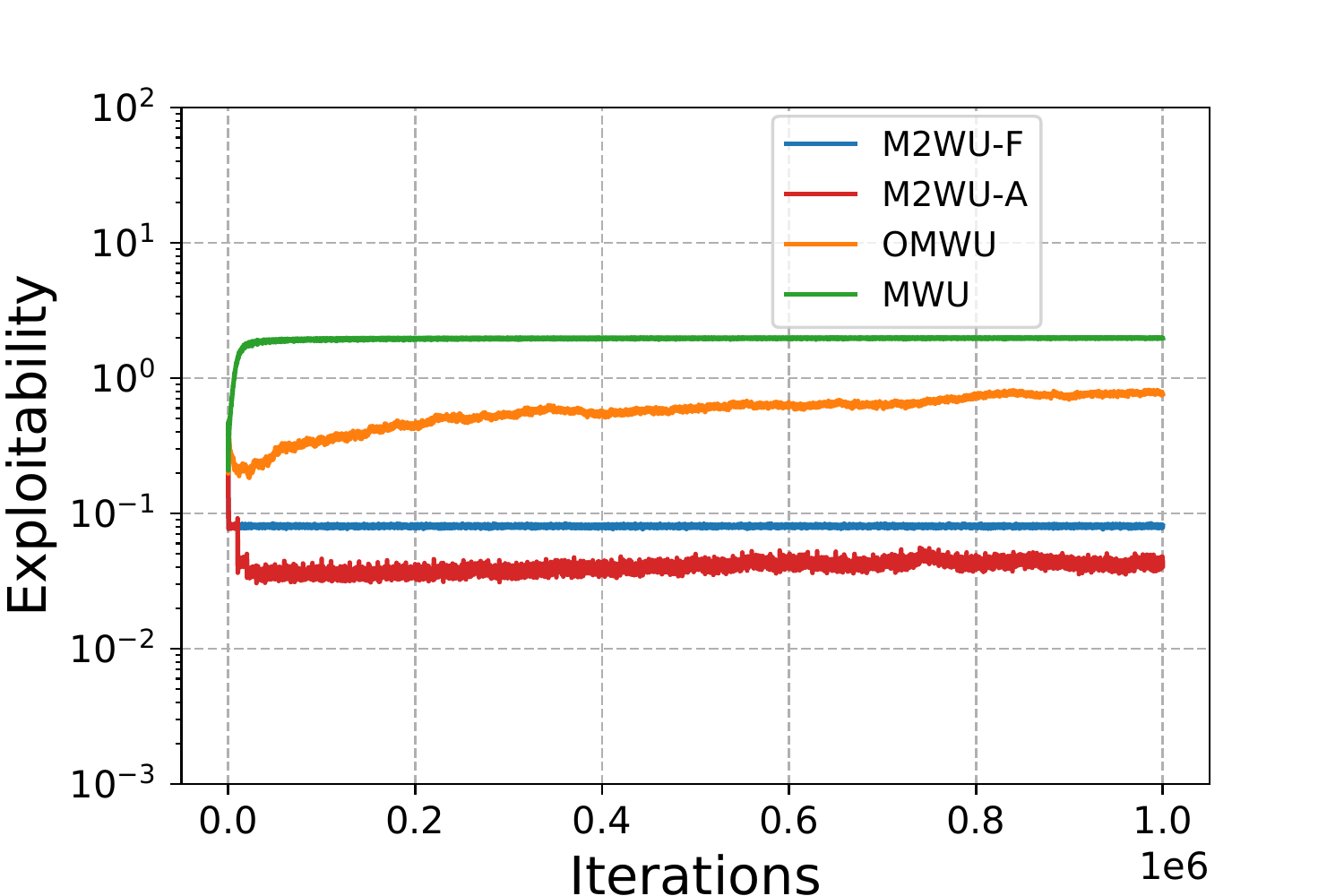}
        \subcaption{$\eta=0.05$}
    \end{minipage}
    \begin{minipage}[t]{0.49\columnwidth}
        \centering
        \includegraphics[width=1.0\textwidth]{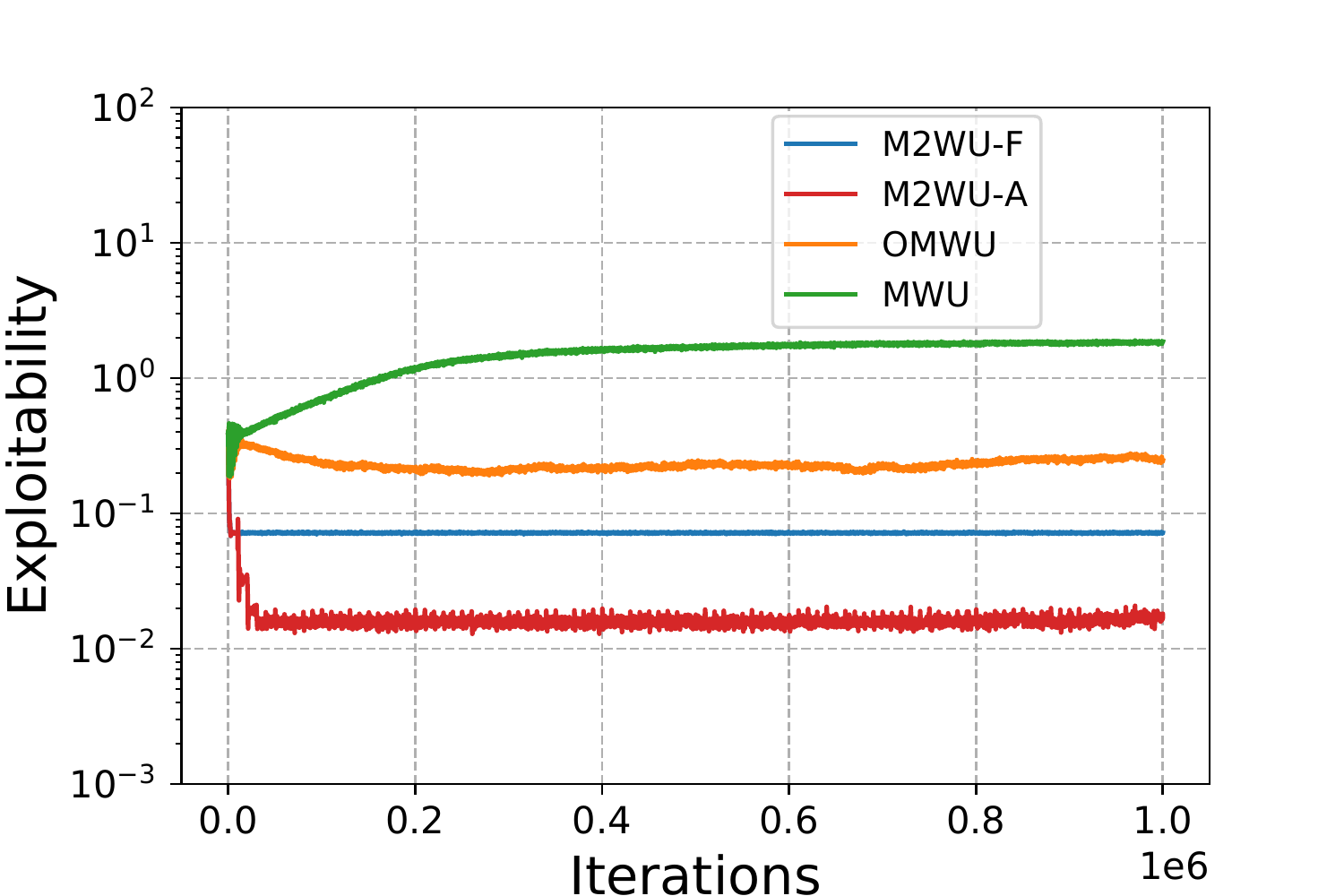}
        \subcaption{$\eta=0.01$}
    \end{minipage}
    \begin{minipage}[t]{0.49\columnwidth}
        \centering
        \includegraphics[width=1.0\textwidth]{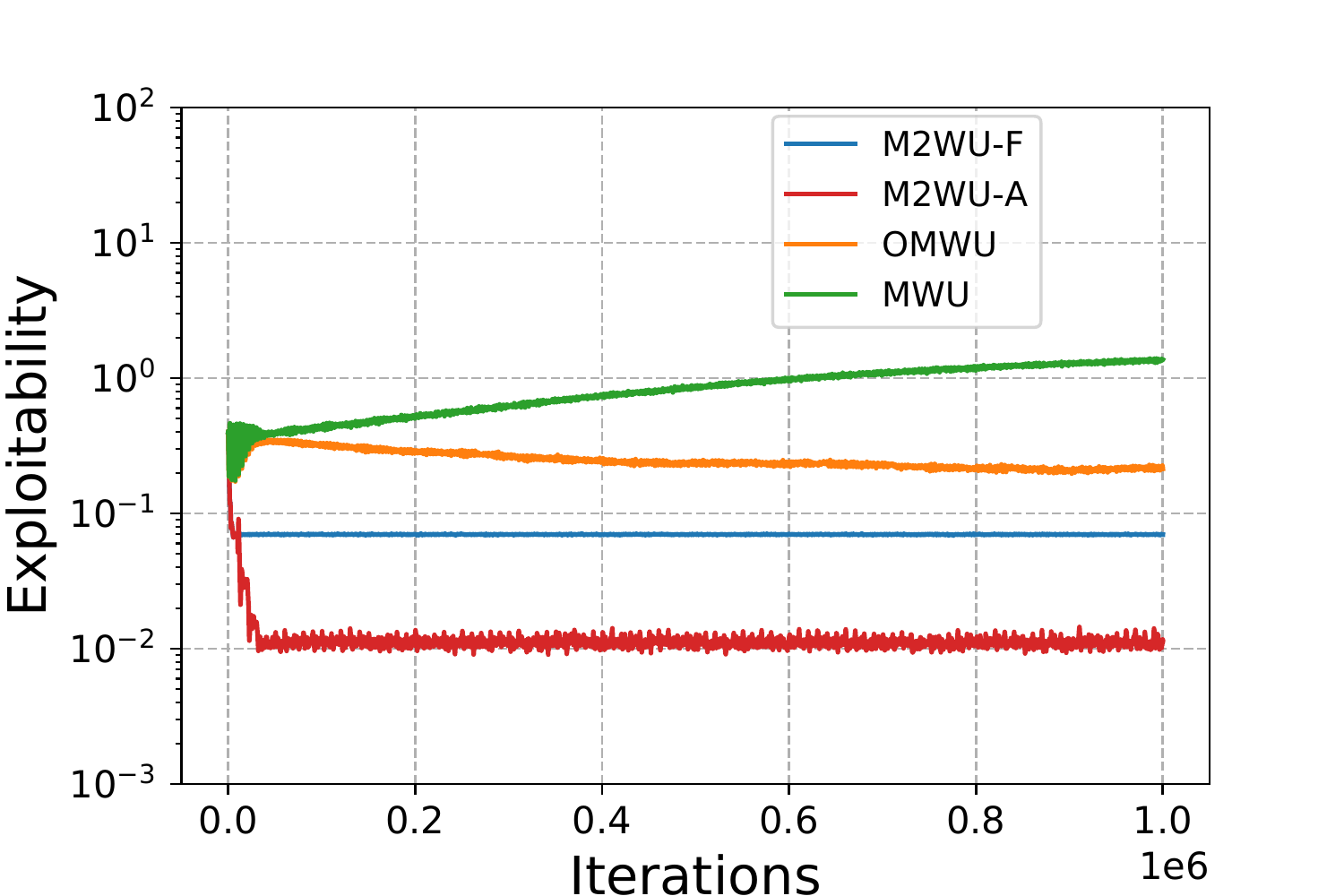}
        \subcaption{$\eta=0.005$}
    \end{minipage}
    \begin{minipage}[t]{0.49\columnwidth}
        \centering
        \includegraphics[width=1.0\textwidth]{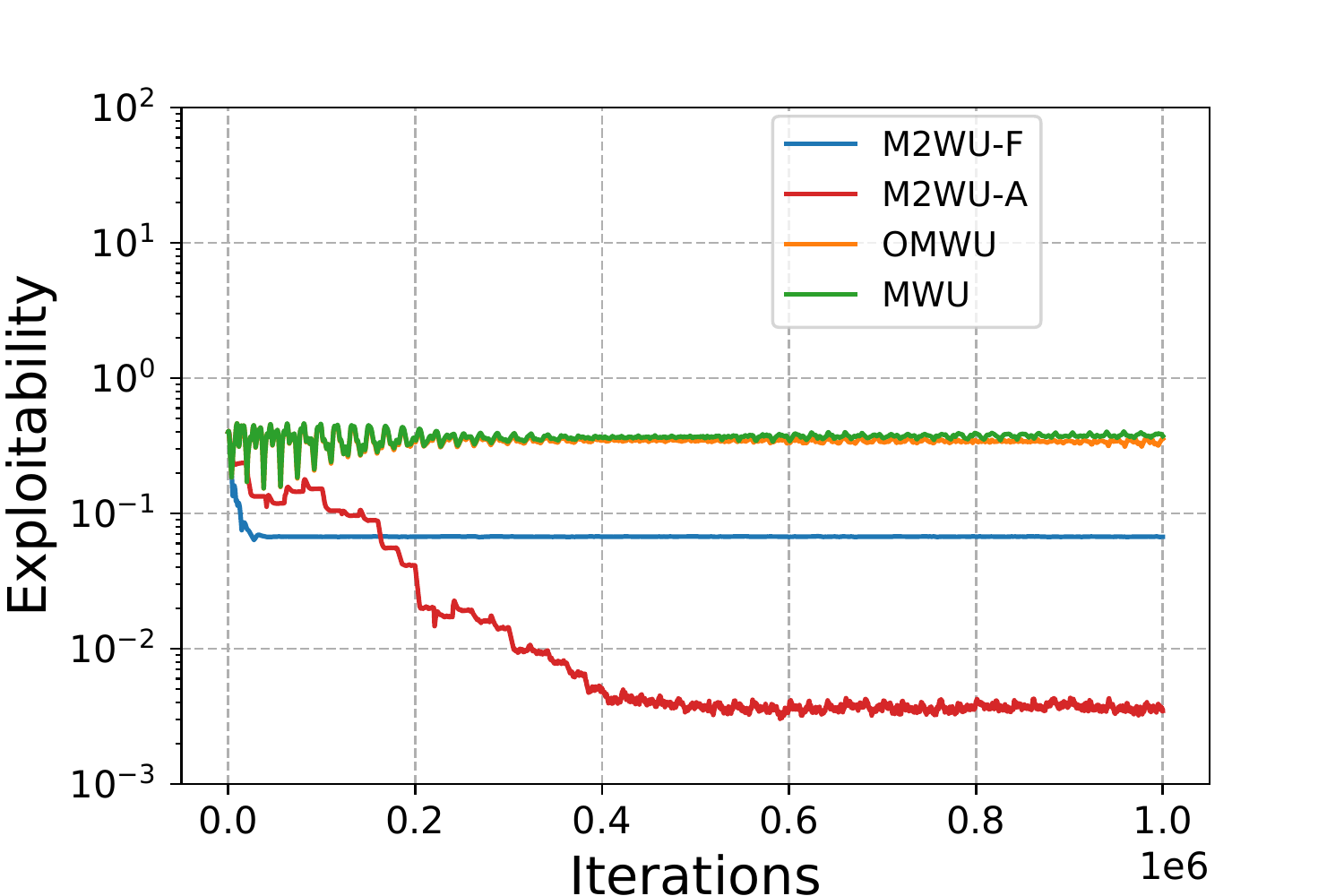}
        \subcaption{$\eta=0.001$}
    \end{minipage}
    \caption{
    Exploitability of $\pi^t$ for M2WU, MWU, and OMWU with varying $\eta\in \{0.1, 0.05, 0.01, 0.005, 0.001\}$ in M-Ne with noisy-information feedback.
    }
    \label{fig:exploitability_noisy_varying_eta_MNe}
\end{figure}

\subsection{Decreasing Learning Rates}
\label{sec:appendix_experiments_noisy_decay}
In this section, we investigate the performance of M2WU with decreasing learning rates under the noisy-information feedback setting.
We set the learning rates to $\eta_t=t^{-\frac{3}{4}}$ for all algorithms.
Other settings are equivalent to the noisy-information feedback experiments in Section \ref{sec:exp_noisy_feedback}.
Figure \ref{fig:exploitability_noisy_decay} shows the average exploitability of $\pi^t$ on $100$ instances.
Even with the decreasing learning rates, M2WU demonstrates better performance than MWU and OMWU.

\begin{figure}[ht!]
    \centering
    \begin{minipage}[t]{0.49\columnwidth}
        \centering
        \includegraphics[width=1.0\textwidth]{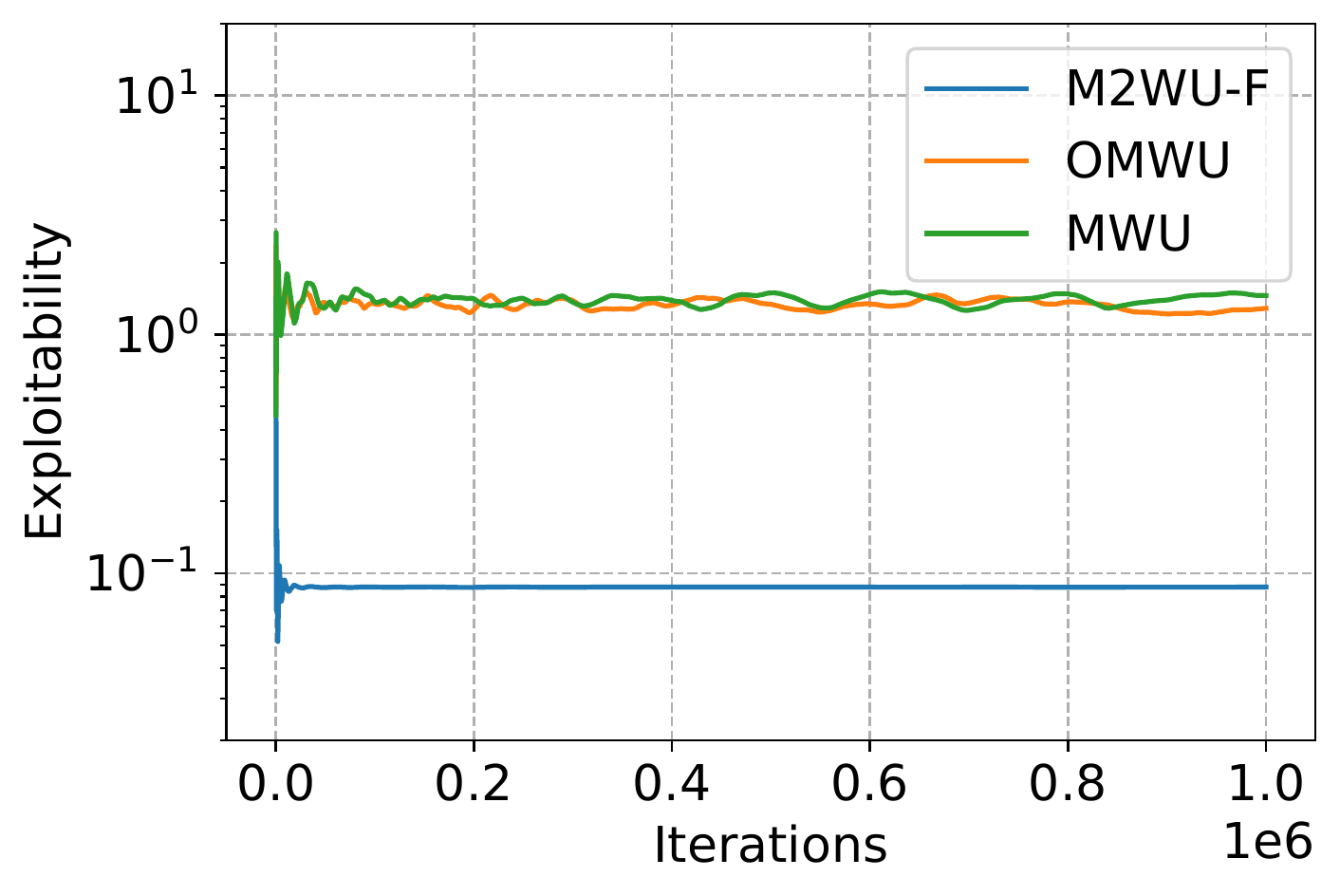}
        \subcaption{BRPS}
    \end{minipage}
    \begin{minipage}[t]{0.49\columnwidth}
        \centering
        \includegraphics[width=1.0\textwidth]{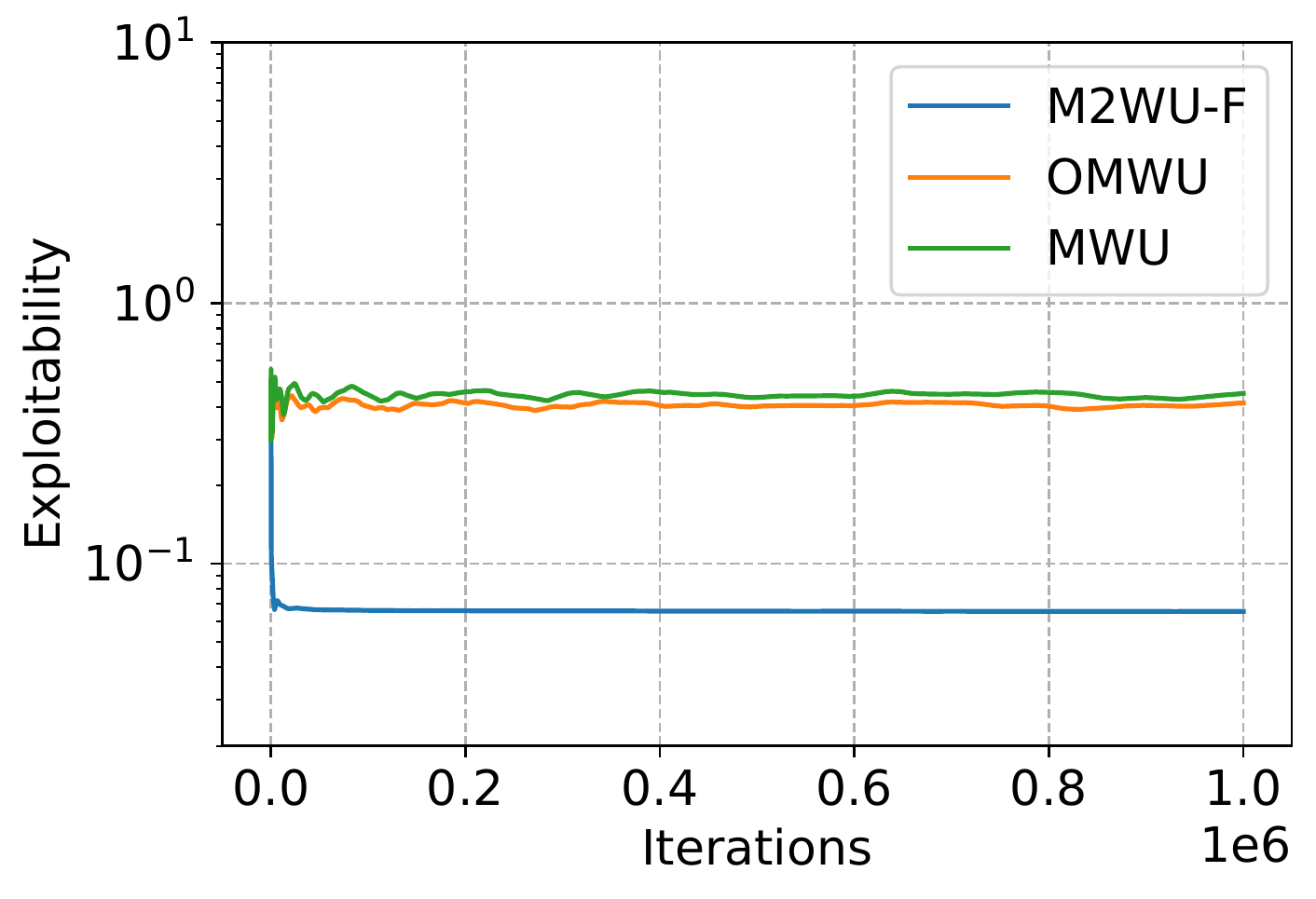}
        \subcaption{M-Ne}
    \end{minipage}
    \caption{
    Exploitability of $\pi^t$ for M2WU, MWU, and OMWU with decreasing learning rates under the noisy-information feedback setting.
    }
    \label{fig:exploitability_noisy_decay}
\end{figure}